\documentclass[10pt]{article}
\usepackage{amsmath,amssymb,graphicx,hyperref,bm}
\usepackage[margin=1in]{geometry}
\usepackage{microtype}
\usepackage{amsthm}
\usepackage{algorithm}
\usepackage{algpseudocode}
\usepackage{float}
\usepackage{booktabs}
\usepackage{comment}
\newtheorem{theorem}{Theorem}
\newtheorem{lemma}{Lemma}
\newtheorem{remark}{Remark}
\newtheorem{notation}{Notation}
\newtheorem{assumption}{Assumption}
\newtheorem{proposition}{Proposition}
\newtheorem{conjecture}{Conjecture}
\usepackage{xcolor}
\colorlet{blue}{black}
\newcommand{\blue}[1]{\textcolor{black}{#1}}
\newcommand{\red}[1]{\textcolor{black}{#1}}

\usepackage{subcaption}
\usepackage{placeins}

\newif\ifshowfigures
\showfigurestrue

\newif\ifconvertedfigures
\convertedfigurestrue

\usepackage[normalem]{ulem}

\begin{document}
\Large \textbf{{Single- and Multilevel Quadrature with Error Control for Fourier Pricing under the Rough Heston Model}} \\[1.2em]
\normalsize
Chiheb Ben Hammouda$^{1}$, Abderrahmene Ben Romdhane$^{3}$, Michael Samet$^{2}$ and Raúl F. Tempone$^{2,3}$ \\[1.2em]
\small
$^{1}$Mathematical Institute, Utrecht University, 3584 CD Utrecht, the Netherlands.\\
$^{2}$Mathematics for Uncertainty Quantification, RWTH Aachen University, Aachen, Germany.\\
$^{3}$King Abdullah University of Science and Technology (KAUST), Thuwal, Saudi Arabia.\\

\begin{abstract}
\noindent
Unlike the classical Heston model, Fourier pricing under the rough Heston model requires solving a fractional Riccati equation at every quadrature point. Since the required resolution varies with model parameters and quadrature point, a single uniform time discretization can be inefficient. We develop single- and multilevel Gauss-Laguerre quadrature methods that balance the time discretization and Fourier quadrature errors. Both methods scale the laguerre weight to the estimated Fourier integrand decay. The single-level method allocates a prescribed tolerance between the two errors. The multilevel method splits the integrand into a level-zero term and level differences, selecting quadrature points separately at each level. Suppose that the Fourier integrand discretization error is $\mathcal{O}(\Delta t^p)$, that evaluating the characteristic function once costs $\mathcal{O}(\Delta t^{-\beta})$, and that the algebraic Gauss-Laguerre quadrature error is $\mathcal{O}(N^{-s_\mathrm{SL}/2})$, where $s_{\mathrm{SL}}$ is the smoothness index. Under this estimate and assumptions on the regularity and decay of level differences, we prove that the proposed single-level method requires $\mathcal{O}(\varepsilon^{-(\beta/p+2/s_{\mathrm{SL}})})$ computational work to achieve accuracy $\varepsilon$, whereas the proposed multilevel method requires $\mathcal{O}(\varepsilon^{-\beta/p})$ computational work. We also study root-exponential Gauss-Laguerre error models for practical multilevel quadrature allocation. Numerical experiments support the observed fractional Riccati and Fourier integrand convergence rates and root-exponential quadrature behavior, and show substantial reductions in quadrature cost from the proposed scaling. The multilevel method provides clear computational savings over the single-level method. We further benchmark the multilevel fractional Riccati method against the BL2 Markovian approximation~\cite{bayer2023weak} and report lower total CPU time in the tested configurations.

\end{abstract}

\vspace{0.4em}

\vspace{0.4em}

\clearpage
\tableofcontents
\clearpage

\section{Introduction}
Fourier methods \cite{lewis2001simple,carr1999option,fang2009novel} form the basis of many of the most efficient pricing and calibration routines used in Markovian stochastic volatility \cite{gatheral2022volatility}, jump-diffusion \cite{duffie2002affine}, and L\'evy models \cite{schoutens2003levy}. This efficiency stems from the availability of the characteristic function of the {log price} process, which allows option valuation problems to be recast as Fourier inversion problems, replacing stochastic simulation-based methods by the numerical evaluation of deterministic contour integrals of analytic, but possibly  oscillatory, functions in the Fourier domain \cite{schmelzle2010option}.

In Markovian affine stochastic volatility models, the characteristic function is available in closed form as the explicit solution of systems of generalized Riccati ordinary differential equations (ODEs) \cite{duffie2002affine}. A canonical example is the Heston model \cite{heston1993}, which remains widely used in equity and foreign exchange markets \cite{gatheral2022volatility}. However, empirical evidence \cite{bayer2023rough} has revealed systematic shortcomings of such models in reproducing the behavior of implied volatility surfaces, in particular the pronounced steepness of the at-the-money skew at short maturities \cite{bayer2016pricing}. These observations motivated the development of rough volatility models, which \blue{incorporate memory effects by modeling the driving noise of the volatility process as a fractional Brownian motion}. Prominent examples include the rough Bergomi model \cite{bayer2016pricing} and the rough Heston model \cite{el2019roughening}.

The loss of Markovianity precludes the direct application of classical partial differential equation (PDE)-based pricing techniques and increases the computational complexity of Monte Carlo simulation methods due to the path-dependent nature of the volatility dynamics \cite{bayer2024efficient,bayer2016pricing}. Despite these challenges, Fourier-based pricing methods remain attractive in the rough Heston setting, following the seminal work of \cite{el2019characteristic}, who showed that the characteristic function admits a semi-explicit exponential-affine representation. As a consequence, option pricing can be reformulated as the numerical evaluation of a deterministic Fourier integral; however, in contrast to the Markovian affine case, this requires solving a fractional Riccati equation independently at each {quadrature point}, which becomes a major computational bottleneck in calibration settings where the characteristic function must be evaluated repeatedly across many parameter values.

This computational challenge has motivated two complementary lines of research aimed at restoring numerical tractability while preserving the  features of rough volatility dynamics. 
A first line of work replaces the original non-Markovian volatility process by a Markovian surrogate, constructed by approximating the fractional kernel  by a finite sum of exponentials \cite{abi2019multifactor,abi2019lifting,bayer2023weak}. 
{The exponential terms introduce a collection of mean-reverting factors and yield a finite-dimensional Markovian representation.} The resulting model retains the ability to reproduce key empirical features of rough volatility while replacing the fractional Riccati equation governing the characteristic function by a finite-dimensional system of classical Riccati ODEs.  

A second line of research retains the original rough volatility model and instead seeks to accelerate valuation directly. 
A substantial body of work within this line focuses on improving the numerical solution of the fractional Riccati equation itself. 
The standard approach remains the fractional Adams predictor-corrector scheme \cite{diethelm2002predictor}, which has been widely adopted in the rough volatility literature. Its direct implementation evaluates a discrete convolution over the full history at each time step and therefore requires $\mathcal{O}(\Delta t^{-2})$ operations over a fixed time interval \cite{li2009fractional,callegaro2018rough}. 
Several refinements have been proposed to enhance its efficiency, including Richardson-Romberg extrapolation combined with fractional power series expansions \cite{callegaro2018rough}, rational approximation techniques that construct fast surrogates for the Riccati solution \cite{gatheral2019rational}, as well as modified Adams-type schemes with improved stability and performance properties \cite{boyarchenko2025fast}. 
In parallel, other works have sought to accelerate Fourier pricing by focusing on the numerical evaluation of the Fourier integral itself, for instance through contour deformations and SINH-type accelerations \cite{boyarchenko2025fast}, or by adopting alternative transform-based methodologies such as the SINC method \cite{baschetti2022sinc}. 
These approaches improve the numerical resolution of the fractional Riccati solver or the efficiency of the quadrature, but treat the two components largely in isolation.

Despite this progress, several theoretical and numerical issues remain only partially understood and continue to constrain robust and systematically optimized implementations.
First, a full characterization of the relevant strip of analyticity of the rough Heston characteristic function is not available in explicit form, complicating principled contour selection and the tuning of damping parameters \cite{boyarchenko2025fast,bayer2023optimal}. In this work, we propose a practical heuristic that extends the {guideline for contour selection} of \cite{bayer2023optimal} to settings in which the analyticity strip is not available in closed form. Second, in contrast to the classical Heston model, for which the asymptotic behaviour of the characteristic function for large values of the Fourier integration variable is available in explicit form \cite{andersen2018robust}, no comparably rigorous asymptotic characterization has been derived for the rough Heston model. This gap complicates the construction of quadrature rules adapted to the tail behaviour and integral transformations, including the domain transformation strategies for quasi-Monte Carlo \cite{bayer2024quasi} and the double exponential quadrature method proposed in \cite{andersen2022high}.

The present work follows the second line of research in that it retains the exact rough Heston model and its fractional Riccati representation. Rather than studying the fractional Riccati solver and the Fourier quadrature separately, {we control their errors jointly and allocate the work required to meet a prescribed tolerance optimally. A related work derives an allocation of computational work between numerical quadrature and optimization that minimizes the cost of computing multivariate shortfall risk measures \cite{hammouda2026single}. While the time step determines the accuracy and cost of each characteristic function evaluation, the quadrature error estimate allows us to determine the number of required quadrature points in the Fourier domain. We first propose a single-level method that allocates the tolerance between the time discretization and quadrature errors and uses a scaled Gauss-Laguerre quadrature rule whose weight is matched to the estimated decay of the Fourier integrand to prevent issues documented for short maturities in \cite{boyarchenko2025fast}. We then extend the construction to a multilevel method, which applies the scaled Gauss-Laguerre quadrature rule to a level zero term and a hierarchy of level differences. \red{This construction follows the deterministic telescoping-quadrature paradigm \cite{harbrecht2012multilevel,griebel2020multilevel}; its specific contribution here is to apply the hierarchy jointly to the Fourier quadrature and the fractional-Riccati time discretization.} The proposed multilevel method selects the number of quadrature points separately at each level to optimally distribute the work over the hierarchy.} Our main contributions are as follows.
\begin{itemize}
	\item {\color{blue}
	We develop a single-level Fourier pricing method based on scaled Gauss-Laguerre quadrature.
	The pricing error is decomposed into a fractional-Riccati time-discretization error and a Fourier
	quadrature error, and the prescribed tolerance is allocated between them. The Laguerre scaling
	is chosen to reflect the decay of the Fourier integrand, and both algebraic and root-exponential
	quadrature error models are analyzed.
	}

	\item {\color{blue}
	Let $p$ denote the convergence order of the Fourier integrand under refinement of the fractional-
	Riccati time discretization, $\beta$ the per-node cost exponent, and $s_{\mathrm{SL}}$ the algebraic Gauss-
	Laguerre convergence exponent of the single-level integrand. We show that balancing the
	time-discretization and quadrature errors gives
	\[
	W_{\mathrm{SL}}(\varepsilon)
	=
	\mathcal{O}\!\left(
	\varepsilon^{-\beta/p-2/s_{\mathrm{SL}}}
	\right).
	\]
	We then develop a multilevel method based on a telescopic decomposition of the discretized
	Fourier integrand and level-dependent quadrature allocation. By exploiting the decay of the
	multilevel corrections, the method achieves
	\[
	W_{\mathrm{ML}}(\varepsilon)
	=
	\mathcal{O}\!\left(
	\varepsilon^{-\beta/p}
	\right)
	\]
	under the stated algebraic regularity and correction-decay assumptions. Root-exponential
	quadrature models are additionally used for the practical multilevel allocation.
	}

	\item {\color{blue}
	Numerical experiments support the observed fractional-Riccati/Fourier-integrand convergence
	rates and root-exponential quadrature behavior, and demonstrate substantial reductions in
	quadrature cost from the proposed scaling. We compare the single- and multilevel methods
	in terms of computational work and CPU time and observe clear savings from the multilevel
	construction. Finally, we benchmark the multilevel direct fractional-Riccati method against the
	\red{BL2 Markovian approximation~\cite{bayer2023weak} and report lower total CPU time in the four tested
	configurations, subject to reference-assisted BL2 factor selection.}
	}
\end{itemize}

The remainder of the paper is organized as follows. Section~\ref{sec:problem_setting} introduces the pricing framework and the Fourier representation. Section~\ref{sec:methodology} {develops the single- and multilevel methods and establishes their error and complexity estimates}. {Section~\ref{sec:numerical_results} tests these estimates, examines the effects of scaling the Gauss-Laguerre quadrature rule and allocating quadrature points among the levels, and presents the comparison with the BL2 benchmark.} Finally, Section~\ref{sec:conclusion} concludes and discusses possible extensions.

\section{Problem Setting and Pricing Framework}
\label{sec:problem_setting}
In this section, we formulate the valuation problem in Fourier space, state the underlying assumptions on the model and the payoff function, and fix the notation used throughout the paper. Section~\ref{sec:valuation_formula} introduces the Fourier pricing representation and the associated notation, while Section~\ref{sec:models_used} specifies the rough Heston dynamics and the payoff considered in this work.
\subsection{Fourier Pricing Valuation Formula}
\label{sec:valuation_formula}
Before stating the  Fourier-based valuation formula, we introduce the necessary notation, definitions, and assumptions.

\begin{notation}[Notations and definitions]
\mbox{}\par\vspace{0.5em}  
\begin{itemize}

    \item For $\xi \in \mathbb{C}$, $\Re[\xi]$ and $\Im[\xi]$ denote the real and imaginary parts of $\xi$.

    \item For $f \in L^1(\mathbb{R})$, $\widehat{f}$ denotes the Fourier transform of 
    $f:\mathbb{R}\to\mathbb{C}$, defined by
    \[
    \widehat{f}(\xi)=\int_{\mathbb{R}} e^{-\mathrm{i} \xi x}\,f(x)\,dx,
    \]
    with inverse transform
    \[
    f(x)=\frac{1}{2\pi}\int_{\mathbb{R}} e^{\mathrm{i} \xi x}\,\widehat{f}(\xi)\,d\xi.
    \]

    \item $X_t := (X_t)_{0 \le t \le T}$ denotes the {log price process} at time $t$,
    where $T$ is the maturity of the option. The dynamics of $X_t$ follow a stochastic volatility model 
    with parameter vector $\Theta_X$.

    \item The extended characteristic function of $X_t$ is defined as
    \[
    \phi(\xi,t) := \mathbb{E}\!\left[e^{\mathrm{i} \xi X_t}\right], \quad \xi \in \mathbb{C},
    \]
    when the expectation is {well defined} and admits an analytic continuation \cite{andersen2018robust}.

    \item The {characteristic function at maturity $T$} is denoted by
    \[
    \Phi(\xi) := \phi(\xi, T), \quad \xi \in \mathbb{C}.
    \]

\item The payoff function is denoted by $P:\mathbb{R}\to\mathbb{R}_+$. 

\item The generalized Fourier transform of the payoff is defined 
as
\[
\widehat{P}(\xi) := \int_{\mathbb{R}} e^{-\mathrm{i} \xi x} P(x)\,dx, \quad \xi\in\mathbb{C},
\]
whenever the integral converges.

    \item $\Theta_P = (K, T, r)$ denotes the vector of payoff and market parameters, where
    $K$ is the strike price, $T$ is the maturity, and $r$ is the risk-free interest rate.

    \item For stochastic processes $X$ and $Y$, $\langle X, Y \rangle_t$ denotes their quadratic 
    covariation process up to time $t$.

    \item For an interval $I \subset \overline{\mathbb{R}} := \mathbb{R}\cup\{-\infty,+\infty\}$,
    $L^p(I)$ denotes the space of measurable functions $f:I\to\mathbb{C}$ such that
    \[
    \|f\|_{L^p(I)} < \infty, \qquad 1 \le p \le \infty.
    \]

    \item For $f \in L^p([a,b])$, the $L^p$ {norm} is defined by
    \[
    \|f\|_{L^p([a,b])}
    =\left(\int_{a}^{b} |f(x)|^{p}\,dx\right)^{1/p}, 
    \quad 1\le p<\infty,
    \]
    and
    \[
    \|f\|_{L^\infty([a,b])}
    =\sup_{x\in[a,b]}|f(x)|.
    \]
    The interval $[a,b]$ may be finite, semi-infinite, or equal to $\mathbb{R}$.

 \item $
L_{bc}^1(\mathbb{R})
:=
\Big\{
f:\mathbb{R}\to\mathbb{C}
\ \Big|\
f \ \text{bounded and continuous},\ \|f\|_{L^1(\mathbb{R})}<\infty
\Big\}.$

    \item $AC_{\mathrm{loc}}(0,\infty)$ denotes the space of functions on $(0,\infty)$ that are 
    absolutely continuous on every compact subinterval.

    \item $\Gamma(\cdot)$ denotes the complex Gamma function, defined {for $\xi \in \mathbb{C}$}, $\Re[\xi]>0$ by
    \[
    \Gamma(\xi) = \int_0^\infty t^{\xi-1} e^{-t}\, dt.
    \]

\end{itemize}
\end{notation}

\begin{assumption}[Assumptions on the payoff]\label{ass:payoff}
\[
\delta_P
:=
\left\{
R \in \mathbb{R}
\;\middle|\;
x \mapsto e^{R x} P(x) \in L_{bc}^1(\mathbb{R}), \; and \; u \mapsto \widehat{P}(u + \mathrm{i}R) \in L^1(\mathbb{R})
\right\}
\neq \varnothing .
\]
\end{assumption}
\begin{assumption}[Assumptions on the model]\label{ass:model}
{\[
\delta_X
:=
\left\{
R \in \mathbb{R}
\;\middle|\;
\mathbb{E}\!\left[e^{-R X_T}\right] < \infty
\right\}
\neq \varnothing .
\]}
\end{assumption}

A general treatment of admissibility conditions for Fourier pricing formulas, covering different regularity and integrability assumptions on the payoff and the characteristic function, is given in \cite{eberlein2010analysis}. In the rough Heston model, however, the admissible strip of analyticity of the characteristic function is not known explicitly, in contrast to many Lévy models. This complicates the selection of the damping parameter. {We therefore propose, in Section~\ref{sec:numerical_results}, a practical selection rule that extends the approach of \cite{bayer2023optimal}.} Alternative choices of the damping parameter and the integration contour are discussed, for example, in \cite{boyarchenko2025fast}.

\begin{proposition}[Fourier pricing valuation formula]
\label{prop:fourier_formula}
\mbox{}\par\vspace{0.5em}  

We suppose that Assumptions~\ref{ass:payoff} and~\ref{ass:model} hold, and that $\delta_V = \delta_X \cap \delta_P \neq \emptyset$. {By~\cite{eberlein2010analysis,bayer2023optimal}, for any $R \in \delta_V$, the option value is given by}
\begin{equation}
\label{eq:fourier_pricing}
V(\Theta_X, \Theta_P)
= \frac{e^{-rT}}{2\pi}
\, \Re \!\left[
\int_{\mathbb{R}}
\Phi(u+\mathrm{i}R)
\, \widehat{P}(u+\mathrm{i}R)
\, du
\right].
\end{equation}
\end{proposition}

\begin{proof}
We refer the reader to \cite{eberlein2010analysis,bayer2023optimal} for the proof
of the Fourier valuation formula.
\end{proof}

\vspace{1.2em}

\noindent

In what follows, based on equation~\eqref{eq:fourier_pricing}, we define the integrand of interest by
\begin{equation}
\label{eq:integrand}
g(u; R, \Theta_X, \Theta_P)
:= \frac{e^{-rT}}{2\pi}\,
\Re\!\left[\Phi(u+\mathrm{i}R)\,\widehat{P}(u+\mathrm{i}R)\right],
\qquad u \in \mathbb{R},\; R \in \delta_V.
\end{equation}

{\color{blue}
For $R\in\delta_V$, the integrand~\eqref{eq:integrand} is
even\footnote{Since $X_T$ and $P$ are real-valued,
$\Phi(-u+\mathrm{i}R)=\overline{\Phi(u+\mathrm{i}R)}$ and
$\widehat P(-u+\mathrm{i}R)=\overline{\widehat P(u+\mathrm{i}R)}$.
Therefore, the product $\Phi\widehat P$ is conjugated under
$u\mapsto-u$, and $g$ is its real part.}, hence
\begin{equation}
V(\Theta_X,\Theta_P)
=
\int_{\mathbb R}g(u)\,du
=
2\int_0^\infty g(u)\,du.
\label{eq:half_line_fourier_integral}
\end{equation}
All subsequent Fourier-domain analysis is therefore carried out for
$u\geq0$.
}

\subsection{Payoff and Asset Model}
\label{sec:models_used}
In this section, we specify the payoff of interest and introduce the underlying asset model used throughout this work.
{\subsubsection{Payoff and its Fourier Transform}}
Although the proposed framework can be extended to a wide variety of payoff functions (see \cite{eberlein2010analysis}), in this work we focus the European call option.  
Its payoff, for $K > 0$, is defined by
\begin{equation}
P(X_T) = \max(e^{X_T} - K,\, 0).
\label{eq:call_payoff}
\end{equation}
The corresponding Fourier transform of the payoff function is given by
\begin{equation}
\widehat{P}(\xi)
= -\,\frac{K^{\,1-\mathrm{i}\xi}}{\xi^2 + \mathrm{i} \xi},
\label{eq:call_payoff_transform}
\end{equation}
where $\xi := u + \mathrm{i} R$ with $u \in \mathbb{R}$ and $R \in \delta_P  =  \bigl\{ R \in \mathbb{R} \mid R < -1 \bigr\}.$

{\subsubsection{Model and the Corresponding Characteristic Function}}
Throughout this paper, we consider the rough Heston model, obtained by replacing the variance process in the classical Heston model~\cite{heston1993} with a fractional square-root process. Let $H\in(0,\frac12)$, $\nu,\gamma,\theta>0$, and $\rho\in[-1,1]$. The model dynamics under the risk-neutral measure are given by~\cite{el2019roughening}:
\begin{equation}
\begin{aligned}
dS_t
&=
S_t\bigl(r\,dt+\sqrt{V_t}\,dW_t^1\bigr),
\\[4pt]
V_t
&=
V_0
+
\frac{1}{\Gamma\!\left(H+\tfrac12\right)}
\int_0^t
(t-s)^{H-\tfrac12}
\bigl(
\gamma(\theta-V_s)\,ds
+
\nu\gamma\sqrt{V_s}\,d W_s^2
\bigr).
\end{aligned}
\label{eq:rHeston_model}
\end{equation}
Here, $S_0,V_0>0$, and all processes are defined on a filtered probability space $(\Omega,\mathcal F,(\mathcal F_t)_{t\in[0,T]})$. The processes $(W^1,W^2)$ are standard $\mathcal F$-Brownian motions satisfying $\mathrm{d}\langle W^1,W^2\rangle_t=\rho\,\mathrm{d}t$, and $\mathcal F_t:=\sigma(W_s^1,W_s^2;\,0\leq s\leq t)$.

We define the log price by $X_t:=\log S_t$ and set $\alpha:=H+\frac12\in(\frac12,1)$. It was proved in~\cite{el2019characteristic} that the characteristic function $\phi(\xi,t)$ of $X_t$ admits an exponential-affine representation in terms of the solution of a fractional Riccati equation. 
For $\xi=u+\mathrm{i}R$, with $u\in\mathbb R$ and
$R\in\delta_X$, let $h(\xi,\cdot)$ solve the fractional Riccati equation
\begin{equation}
D_t^\alpha h(\xi,t)
=
F\!\bigl(\xi,h(\xi,t)\bigr),
\qquad
0\leq t\leq T,
\qquad
{I_t^{1-\alpha}}h(\xi,0)=0,
\label{eq:riccati}
\end{equation}
where,
\begin{equation}
F(\xi,h)
:=
-\tfrac12\bigl(\xi^2+\mathrm{i}\xi\bigr)
+
\gamma\bigl(\mathrm{i}\xi\rho\nu-1\bigr)h
+
\tfrac{(\gamma\nu)^2}{2}h^2.
\label{eq:F_def}
\end{equation}

The operators ${I_t^\alpha}$ and ${D_t^\alpha}$ denote the Riemann-Liouville fractional integral and derivative:
$$
\begin{aligned}
I_t^\alpha u(t)
&=
\frac{1}{\Gamma(\alpha)}
\int_0^t (t-s)^{\alpha-1}u(s)\,ds,
\\[4pt]
D_t^\alpha u(t)
&=
\frac{1}{\Gamma(1-\alpha)}
\frac{d}{dt}
\int_0^t (t-s)^{-\alpha}u(s)\,ds.
\end{aligned}
$$

Equivalently, $h$ satisfies the Volterra equation
\begin{equation}
h(\xi,t)
=
\frac{1}{\Gamma(\alpha)}
\int_0^t
(t-s)^{\alpha-1}
F\!\bigl(\xi,h(\xi,s)\bigr)\,ds.
\label{eq:volterra_form}
\end{equation}
Since $h={I_t^\alpha} F(\xi,h)$, the identity ${I_t^{1-\alpha}I_t^\alpha=I_t^1}$ yields the equivalent representation given in Proposition~2.1 of~\cite{boyarchenko2025fast}:
$$
\phi(\xi,t)
=
\exp\!\left(
\mathrm{i}\xi\bigl(X_0+rt\bigr)
+
\int_0^t
\left[
\theta\gamma h(\xi,s)
+
V_0F\!\bigl(\xi,h(\xi,s)\bigr)
\right]\,ds
\right),
\qquad
0\leq t\leq T.
$$
This representation avoids a separate numerical approximation of the fractional integral ${I_t^{1-\alpha}}h(\xi,t)$. Since the Fourier pricing formula in Proposition~\ref{prop:fourier_formula} uses the characteristic function at maturity $T$, we write
\[
\Phi(\xi):=\phi(\xi,T)=\exp\!\bigl(G(\xi)\bigr),
\]
where the {characteristic function exponent} is given by
\begin{equation}
G(\xi)
:=
\mathrm{i}\xi\bigl(X_0+rT\bigr)
+
\int_0^T
\left[
\theta\gamma h(\xi,s)
+
V_0F\!\bigl(\xi,h(\xi,s)\bigr)
\right]\,ds.
\label{eq:terminal_characteristic_exponent}
\end{equation}

\section{Methodology and Numerical Analysis}
\label{sec:methodology}

This section develops and analyses single-level and multilevel Gauss-Laguerre quadrature {approximations} for Fourier pricing under the rough Heston model. At each quadrature {point}, the characteristic function is evaluated by numerically solving the Riccati-Volterra equation~\eqref{eq:volterra_form}.


\subsection{Discretization of the Fractional Riccati Equation}
\label{sec:riccati_discretization}

{\color{blue}
Let $\Delta t_0>0$ and define the hierarchy
$\Delta t_\ell:=2^{-\ell}\Delta t_0$, $\ell\geq0$. The coarsest time step $\Delta t_0$ is fixed and assumed to be
sufficiently small that the discrete Riccati solutions are well defined and
numerically stable on the relevant Fourier contour, and that the nodal and
time-integration error estimates introduced below hold from level zero
onward. Thus, the hierarchy considered in the analysis starts in the
asymptotic convergence regime. In the numerical implementation,
$\Delta t_0$ is treated as a tuning parameter subject to the
stability/refinement diagnostic stated below. Developing an automatic
cost-aware selection procedure is left for future work. We assume that
$M_\ell:=T/\Delta t_\ell$ is an integer and set
$t_{\ell,j}:=j\Delta t_\ell$, $j=0,\ldots,M_\ell$. The numerical solver
returns the nodal values
\[
h_{\ell,j}(\xi)
:=
h_\ell(\xi,t_{\ell,j}),
\qquad
j=0,\ldots,M_\ell,
\]
with $h_{\ell,0}(\xi)=0$. For a grid vector
$\mathbf v=(v_0,\ldots,v_{M_\ell})$, define the composite trapezoidal
functional
\begin{equation}
\mathcal T_\ell[\mathbf v]
:=
\Delta t_\ell
\left(
\frac12v_0
+
\sum_{j=1}^{M_\ell-1}v_j
+
\frac12v_{M_\ell}
\right).
\label{eq:time_trapezoidal_functional}
\end{equation}
We also introduce
\begin{equation}
J(\xi,z)
:=
\theta\gamma z+V_0F(\xi,z).
\label{eq:exponent_integrand_J}
\end{equation}

The exponent used in all computations is the fully discrete quantity
\begin{equation}
\begin{aligned}
G_\ell(\xi)
&:=
\mathrm{i}\xi(X_0+rT)
+
\mathcal T_\ell
\left[
\left(
J\!\left(\xi,h_{\ell,j}(\xi)\right)
\right)_{j=0}^{M_\ell}
\right]
\\
&=
\mathrm{i}\xi(X_0+rT)
+
\Delta t_\ell
\left[
\frac12J\!\left(\xi,h_{\ell,0}(\xi)\right)
+
\sum_{j=1}^{M_\ell-1}
J\!\left(\xi,h_{\ell,j}(\xi)\right)
+
\frac12J\!\left(\xi,h_{\ell,M_\ell}(\xi)\right)
\right].
\end{aligned}
\label{eq:discrete_characteristic_exponent}
\end{equation}
Thus, unlike the continuous exponent $G$ in
\eqref{eq:terminal_characteristic_exponent}, $G_\ell$ includes both the
nodal Riccati approximation and the numerical integration in time. The
corresponding implemented characteristic function is
$\Phi_\ell(\xi):=\exp(G_\ell(\xi))$, and $g_\ell$ denotes the Fourier
integrand obtained by replacing $\Phi$ with $\Phi_\ell$
in~\eqref{eq:integrand}, for $u\geq0$.
}

\red{The global integrability and quadrature-regularity conditions used below
are hypotheses of the analysis; we do not prove that the fully discrete
fractional-Adams characteristic function has the required large-frequency
decay on the entire Fourier half-line.}

\blue{Motivated by~\cite{li2009fractional}, we assume the following nodal
weighted error estimate with a generic convergence order $q>0$. This allows us
to apply the same analysis to any solver satisfying
Assumption~\ref{ass:riccati_pointwise_error}.}

{\color{blue}
\begin{assumption}[Nodal Riccati discretization error]
\label{ass:riccati_pointwise_error}
Fix $R\in\delta_V$. For each $\xi=u+\mathrm{i}R$, with
$u\geq0$, there exists a constant $C_\xi>0$, independent of $\ell$,
such that, for every $\ell\geq0$,
\begin{equation}
\left|h_{\ell,j}(\xi)-h(\xi,t_{\ell,j})\right|
\leq
C_\xi t_{\ell,j}^{\alpha-1}\Delta t_\ell^{\,q},
\qquad
j=1,\ldots,M_\ell.
\label{eq:riccati_pointwise_q}
\end{equation}
\end{assumption}
}

{\color{blue}
\begin{assumption}[Finite-horizon boundedness of the Riccati solution]
\label{ass:riccati_boundedness}
Fix $R\in\delta_V$. For each $\xi=u+\mathrm{i}R$, with
$u\geq0$, the fractional Riccati solution exists on $[0,T]$ and
satisfies
\[
h(\xi,\cdot)\in C([0,T];\mathbb C).
\]
In particular,
\(\|h(\xi,\cdot)\|_{L^\infty(0,T)}<\infty\).
\end{assumption}
}

{\color{blue}
We next account for the numerical integration in time used in the fully
discrete exponent $G_\ell$. For fixed $\xi$, define
\begin{equation}
\mathcal J_\xi(t)
:=
J\!\left(\xi,h(\xi,t)\right),
\qquad
0\leq t\leq T.
\label{eq:exact_exponent_time_integrand}
\end{equation}

\begin{assumption}[Time-integration error in the characteristic function exponent]
\label{ass:time_integration_error}
Fix $R\in\delta_V$. There exists an order $r_T>0$ such that, for each
$\xi=u+\mathrm{i}R$, with $u\geq0$, there exists a constant
$C_{\mathrm{time},\xi}>0$, independent of $\ell$, satisfying, for every
$\ell\geq0$,
\begin{equation}
\left|
\int_0^T\mathcal J_\xi(t)\,dt
-
\mathcal T_\ell
\left[
\left(
\mathcal J_\xi(t_{\ell,j})
\right)_{j=0}^{M_\ell}
\right]
\right|
\leq
C_{\mathrm{time},\xi}\Delta t_\ell^{\,r_T}.
\label{eq:time_integration_error}
\end{equation}
\end{assumption}

We now study how the nodal Riccati error and the time-integration error
propagate to the approximation errors of the characteristic function and the
Fourier integrand. Define
\begin{equation}
E_\ell(\xi):=G(\xi)-G_\ell(\xi),
\qquad
p_G:=\min\{q,r_T\}.
\label{eq:exponent_error_and_rate}
\end{equation}
For $u\geq0$, let
\[
A_E(u)
:=
\sup_{\ell\geq0}
\Delta t_\ell^{-p_G}
\left|E_\ell(u+\mathrm{i}R)\right|.
\]
}

{\color{blue}
\begin{lemma}[{Discretization error in the characteristic function exponent}]
\label{lemma:terminal_exponent_error}
Assume that $\delta_V=\delta_X\cap\delta_P\neq\emptyset$ and fix
$R\in\delta_V$. Suppose that
Assumptions~\ref{ass:riccati_pointwise_error},
\ref{ass:riccati_boundedness}, and~\ref{ass:time_integration_error} hold.
Then $A_E(u)<\infty$ for every $u\geq0$, and
\begin{equation}
\left|E_\ell(u+\mathrm{i}R)\right|
\leq
A_E(u)\Delta t_\ell^{\,p_G},
\qquad
u\geq0,\quad \ell\geq0.
\label{eq:exponent_error_pG}
\end{equation}
\end{lemma}
}

\begin{proof}
{The proof for Lemma~\ref{lemma:terminal_exponent_error} is presented in
Appendix~\ref{sec:proof_terminal_exponent_error}.}
\end{proof}

For the subsequent single-level and multilevel analyses, the relevant discretization rate is the $L^1$ convergence rate of the Fourier integrand. We denote this rate by $p>0$.

\begin{assumption}[{$L^1$ discretization error of the Fourier integrand}]
\label{ass:fourier_integrand_rate}
There exists a constant $C_g>0$, independent of $\ell$, such that, for sufficiently small $\Delta t_\ell$,
\begin{equation}
\left\|g-g_\ell\right\|_{L^1(0,\infty)}
\leq
C_g\Delta t_\ell^{\,p}.
\label{eq:integrand_rate_p}
\end{equation}
\end{assumption}

{
We next give sufficient conditions under which
Assumption~\ref{ass:fourier_integrand_rate} holds with
$p=\blue{p_G}$. The
estimate~\eqref{eq:exponent_error_pG} is pointwise in the Fourier integration
variable. For $u\geq0$, define
\[
\Psi(u)
:=
\sup_{\ell\geq0}
\Delta t_\ell^{-\blue{p_G}}
\left|
1-\exp\!\left(-E_\ell(u+\mathrm{i}R)\right)
\right|.
\]
Since $E_\ell=G-G_\ell$ and
$\Phi_\ell=\Phi\exp(-E_\ell)$, we introduce the domination function
\begin{equation}
\begin{aligned}
\mathcal D_R(u)
&:=
\left|\widehat P(u+\mathrm{i}R)\right|
\left|\Phi(u+\mathrm{i}R)\right|
\Psi(u)
\\
&=
\sup_{\ell\geq0}
\Delta t_\ell^{-\blue{p_G}}
\left|\widehat P(u+\mathrm{i}R)\right|
\left|\Phi(u+\mathrm{i}R)-\Phi_\ell(u+\mathrm{i}R)\right|.
\end{aligned}
\label{eq:domination_function}
\end{equation}

\begin{conjecture}[Integrability of the domination function]
\label{conj:contour_domination}
For the fixed damping parameter $R\in\delta_V$, the function
$\mathcal D_R$ defined in~\eqref{eq:domination_function} satisfies
\blue{$\mathcal D_R\in L^1(0,\infty)$}.
\end{conjecture}

The role of Conjecture~\ref{conj:contour_domination} is immediate from the
definition of $\mathcal D_R$. For every $u\geq0$ and $\ell\geq0$,
\begin{equation}
\begin{aligned}
|g(u)-g_\ell(u)|
&\leq
\frac{e^{-rT}}{2\pi}
\left|\widehat P(u+\mathrm{i}R)\right|
\left|\Phi(u+\mathrm{i}R)-\Phi_\ell(u+\mathrm{i}R)\right|
\\
&\leq
\frac{e^{-rT}}{2\pi}
\Delta t_\ell^{\blue{p_G}}\mathcal D_R(u).
\end{aligned}
\label{eq:pointwise_integrand_domination}
\end{equation}
Thus, the conjectured integrability of $\mathcal D_R$ is precisely the
condition that permits integration of the pointwise characteristic-function
error over the Fourier variable. Numerical evidence is provided in
Appendix~\ref{app:contour_domination_numerics}, using the observed rate
$\blue{p_{G,\mathrm{obs}}=1+\alpha}$ reported in
Section~\ref{subsec:time_discretization_numerics}.
}

\blue{Assumptions~\ref{ass:riccati_pointwise_error},
\ref{ass:riccati_boundedness}, and~\ref{ass:time_integration_error}, together
with Conjecture~\ref{conj:contour_domination} yield the following estimate.}

\begin{proposition}[{$L^1$ discretization error of the Fourier integrand}]
\label{prop:discretization_error}
\blue{Suppose that Assumptions~\ref{ass:riccati_pointwise_error},
\ref{ass:riccati_boundedness}, and~\ref{ass:time_integration_error}, and
Conjecture~\ref{conj:contour_domination} hold. Then there exists a constant
$C_g>0$, independent of $\ell$, such that}
\begin{equation}
\left\|g-g_\ell\right\|_{L^1(0,\infty)}
\leq
C_g\Delta t_\ell^{\,\blue{p_G}},
\qquad
\ell\geq0.
\label{eq:g_error_pG}
\end{equation}
\end{proposition}
{
\begin{proof}
We prove Proposition~\ref{prop:discretization_error} in
Appendix~\ref{sec:proof_fourier_integrand_error}.
\end{proof}}
\red{Under Assumptions~\ref{ass:riccati_pointwise_error},
\ref{ass:riccati_boundedness}, and~\ref{ass:time_integration_error}, and
Conjecture~\ref{conj:contour_domination},
Proposition~\ref{prop:discretization_error} yields
Assumption~\ref{ass:fourier_integrand_rate} with $p=p_G$. The required global
domination is not proved for the direct Adams implementation.}

\paragraph{The fractional Adams scheme.}

We solve the Volterra formulation~\eqref{eq:volterra_form} of the
fractional Riccati equation~\eqref{eq:riccati} using the fractional Adams
predictor-corrector method introduced in~\cite{diethelm2002predictor}, as
applied to the rough Heston model in~\cite{el2019characteristic}. Let
$t_j=j\Delta t$, $j=0,\ldots,M$, with $M\Delta t=T$, and let
$h_{\Delta t}(\xi,t_j)$ denote the corresponding approximation of
$h(\xi,t_j)$.

{\color{blue}
At time step $j$, the direct Adams implementation evaluates history sums over
the preceding time points. Summing this work over $j=1,\ldots,M$ gives
$\mathcal{O}(M^2)$ operations for one Fourier quadrature point
\cite{li2009fractional,callegaro2018rough}. Once the nodal Riccati values are
available, evaluating the composite trapezoidal rule in the fully discrete
characteristic function exponent requires $\mathcal{O}(M)$ additional operations.
Consequently, the total cost per characteristic-function evaluation is
$\mathcal{O}(M^2+M)=\mathcal{O}(M^2)$. Thus, the implementation used in this work has exponent $\beta=2$ in the cost
model. The generic analysis below retains $\beta$ so that it also applies to
solvers with different cost estimates.
}

{\color{blue}
For the fractional Adams implementation, Assumption~\ref{ass:riccati_pointwise_error}
is used with a generic nodal convergence exponent $q>0$. We do not provide a
theoretical value of $q$ for the fractional Adams scheme in the present
setting. Proposition~\ref{prop:discretization_error} then yields the
Fourier-integrand rate $p=p_G=\min\{q,r_T\}$.

In what follows, the analysis uses the generic exponent $p$ of
Assumption~\ref{ass:fourier_integrand_rate} and the exponent $\beta>0$ in the
solver cost model, defined by
$W_\ell=\mathcal{O}(\Delta t_\ell^{-\beta})$. The results therefore apply to
any solver for which such a pair $(p,\beta)$ is available. For the direct
fractional Adams implementation used in this work, only the cost exponent is
specialized theoretically, namely $\beta=2$.

\begin{remark}[Observed rate]
\label{rem:observed_rate}
The numerical experiments in
Section~\ref{subsec:time_discretization_numerics} indicate that the measured
integrated Riccati and Fourier-integrand errors converge with slopes close to
$1+\alpha$ for the tested configurations. This observed rate is consistent
with a convergence rate available for the fractional Adams method
in~\cite{li2009fractional} under assumptions that we do not verify in this
work. Accordingly, we use $1+\alpha$ as an empirical rate in the numerical
estimators, while keeping $q$, $r_T$, and $p$ generic in the theoretical
analysis.
\end{remark}
}

\subsection{Single-Level Fourier Gauss-Laguerre  Quadrature Method}
\label{sec:SL_method}
In this section, we design a single-level Fourier pricing method based on a
scaled Gauss-Laguerre quadrature rule for the option value~\eqref{eq:fourier_pricing} with
error control. The approximation combines a time discretization
of~\eqref{eq:volterra_form}, which controls the error of each evaluation of the
characteristic function, with a Gauss-Laguerre quadrature rule, which
determines the number of such evaluations. The two choices are coupled because
each quadrature {point} requires one solve of~\eqref{eq:volterra_form}, so the
discretization level fixes the {cost per characteristic function evaluation} while the Gauss-Laguerre quadrature rule fixes
the {number of quadrature points}. We prescribe a total tolerance $\varepsilon>0$,
decompose the error into a discretization and a quadrature contribution, and
allocate the tolerance between them. We then select the {scaling factor} of the scaled Gauss-Laguerre quadrature rule
in~\eqref{eq:sigma_choice}, the discretization level
in~\eqref{eq:practical_level}, and the {number of quadrature points}
in~\eqref{eq:optimal_N}, and we derive the resulting computational complexity
in Proposition~\ref{prop:work_scaling_single}.

\red{In practice, unavailable quantities in the error bounds are replaced by
Richardson-type indicators and fitted quadrature models; attainment of the
target tolerance is therefore assessed numerically rather than certified a
posteriori.}

Gauss-Laguerre quadrature is a natural choice for the semi-infinite integral
in~\eqref{eq:half_line_fourier_integral}. This choice is further supported by
the behaviour of the integrand at large values of $u$. In fact, it is {conjectured and numerically supported}
in \cite{boyarchenko2025fast} that characteristic function of the rough Heston model decays exponentially along the integration contour with a rate  determined by the model parameters
$\Theta_X$ and the maturity $T$.

\subsubsection{The scaled Gauss-Laguerre quadrature rule.}

The standard Gauss-Laguerre quadrature rule typically employed in the
Fourier pricing literature
\cite{schmelzle2010option,von2015benchop,bayer2023optimal,boyarchenko2025fast}
is associated with the weight function $e^{-u}$. However, the effective decay
rate of the integrand $g$ can be heavily affected by the model parameters, the
maturity, the damping parameter, and the payoff parameters, as shown
in~\cite{bayer2024quasi} for several L\'evy models. We refer also to Section~\ref{subsec:quadrature_numerics} for numerical evidence in the rough Heston setting. We therefore introduce a
{scaling factor} $\tilde\sigma>0$, replacing the weight $e^{-u}$ by
$e^{-\tilde\sigma u}$. Since the quadrature is then applied to the transformed
integrand $e^{\tilde\sigma u}g(u)$, \blue{we choose the factor
$\tilde\sigma$ heuristically to reflect the estimated asymptotic decay of the
characteristic function.}

We fix $\tilde\sigma>0$ throughout and denote by $\{u_n,w_n\}_{n=1}^N$
the nodes and weights of the $N$-point Gauss-Laguerre quadrature rule associated with
the weight $e^{-\tilde\sigma u}$ on $(0,\infty)$; the dependence on
$\tilde\sigma$ is left implicit. We refer to this as the scaled Gauss-Laguerre quadrature rule
\cite{wang2023laguerre}. It is the particular case of the
generalized Gauss-Laguerre quadrature rule, whose weight function is
$u^{\lambda}e^{-\tilde\sigma u}$ with $\lambda=0$. Selecting $\tilde\sigma$ to match the decay
of the integrand follows the domain transformation principle
of~\cite{bayer2024quasi}, where the transformation is chosen to reproduce the
functional form of the asymptotically dominant part of the Fourier integrand.

Writing
\begin{equation}
\int_0^\infty g_\ell(u)\,du
=
\int_0^\infty
e^{-\tilde\sigma u}\,\tilde g_\ell(u)\,du,
\qquad
\tilde g_\ell(u):=e^{\tilde\sigma u}g_\ell(u),
\label{eq:weighted_integrand}
\end{equation}
and absorbing the factor $2$ of~\eqref{eq:half_line_fourier_integral}
into the quadrature approximation, we define
\begin{equation}
\mathcal Q_N^{\tilde\sigma}[g_\ell]
:=
2\sum_{n=1}^N
w_n\,
\tilde g_\ell(u_n),
\label{eq:scaled_GL_rule}
\end{equation}
so that $\mathcal Q_N^{\tilde\sigma}[g_\ell]$ approximates
$2\int_0^\infty g_\ell(u)\,du$. The standard Gauss-Laguerre quadrature rule corresponds to
$\tilde\sigma=1$. Increasing $\tilde\sigma$ concentrates the nodes near the
origin; decreasing $\tilde\sigma$ spreads them towards larger values of $u$. We
write
\begin{equation}
R_N^{\tilde\sigma}(g_\ell)
:=
2\int_0^\infty g_\ell(u)\,du
-
\mathcal Q_N^{\tilde\sigma}[g_\ell]
\label{eq:quadrature_error}
\end{equation}
for the quadrature remainder of the scaled Gauss-Laguerre quadrature rule~\eqref{eq:scaled_GL_rule} at level $\ell$, and $\left|R_N^{\tilde\sigma}(g_\ell)\right|$ for the corresponding quadrature error.

\subsubsection{Error analysis of the single-level method.}

We now approximate the option value $V$ by combining the time
discretization of Section~\ref{sec:riccati_discretization} with the scaled
Gauss-Laguerre quadrature rule~\eqref{eq:scaled_GL_rule}, and split the resulting error into the two
contributions that determine the choice of $L$ and $N$. We denote by
\begin{equation}
V_\ell
:=
2\int_0^\infty g_\ell(u)\,du
\label{eq:level_approximation}
\end{equation}
the value obtained by replacing $\Phi$ with $\Phi_\ell$
in~\eqref{eq:half_line_fourier_integral}, and by
$V_{N,\ell}:=\mathcal Q_N^{\tilde\sigma}[g_\ell]$ its approximation by the scaled
Gauss-Laguerre quadrature rule~\eqref{eq:scaled_GL_rule} with {$N$ quadrature points}. The single-level approximation
at level $L$ is thus $V_{N,L}$. Adding and subtracting $V_L$ gives
\begin{equation}
\left|V-V_{N,L}\right|
\;\leq\;
\underbrace{\left|V-V_L\right|}_{=:\;e_{\mathrm{disc}}(L)}
\;+\;
\underbrace{\left|V_L-V_{N,L}\right|}_{=:\;e_{\mathrm{quad}}(N,L)} ,
\label{eq:error_decomposition}
\end{equation}
where  we recall that $V=V(\Theta_X,\Theta_P)$. The first term is the error of the time
discretization and depends on $L$ only. By~\eqref{eq:half_line_fourier_integral}, \eqref{eq:level_approximation} and Assumption~\ref{ass:fourier_integrand_rate},
$$
e_{\mathrm{disc}}(L)
=
2\left|\int_0^\infty\bigl(g(u)-g_L(u)\bigr)\,du\right|
\leq
2\left\|g-g_L\right\|_{L^1(0,\infty)}
\leq
2C_g\Delta t_L^{\,p}.
$$

The second term is the error from the
Gauss-Laguerre quadrature rule at the fixed level $L$, that is,
$e_{\mathrm{quad}}(N,L)=|R_N^{\tilde\sigma}(g_L)|$ defined in \eqref{eq:quadrature_error}. Given a tolerance $\varepsilon>0$, we split the error as follows
$\varepsilon=\varepsilon_{\mathrm{disc}}+\varepsilon_{\mathrm{quad}}$ and
require
\begin{equation}
e_{\mathrm{disc}}(L)\leq\varepsilon_{\mathrm{disc}},
\qquad
e_{\mathrm{quad}}(N,L)\leq\varepsilon_{\mathrm{quad}}.
\label{eq:tolerance_allocation}
\end{equation}
Unless stated otherwise we take
$\varepsilon_{\mathrm{disc}}=\varepsilon_{\mathrm{quad}}=\varepsilon/2$. Each of
the $N$ terms in~\eqref{eq:scaled_GL_rule} requires one solve
of~\eqref{eq:volterra_form} at cost
$W_L=\mathcal O(\Delta t_L^{-\beta})$. Thus, the computational work required by the single-level method is
\begin{equation}
W_{\mathrm{SL}}(N,L)
:=
N W_L.
\label{eq:sl_work}
\end{equation}
We fix $\tilde\sigma$ first, and then select $L$ and $N$ from the two
constraints in~\eqref{eq:tolerance_allocation}. Before selecting them, we recall the
estimates available for the Gauss-Laguerre quadrature
remainder \eqref{eq:quadrature_error}.

\paragraph{Analysis of the Gauss-Laguerre quadrature error.}

{\color{blue}
We consider both algebraic and root-exponential estimates for the
Gauss-Laguerre quadrature error. As observed numerically in
Section~\ref{subsubsec:quadrature_assumptions}, the root-exponential estimate
provides a closer description of the quadrature errors, whereas the algebraic
estimate is more conservative. The algebraic estimate, however, permits the
explicit optimization and complexity analysis developed in
Propositions~\ref{prop:quadrature_order_algebraic},
\ref{prop:optimal_N_ml}, \ref{prop:work_scaling_single},
and~\ref{prop:work_scaling_multilevel}, including the closed-form expression
for the minimum multilevel work in~\eqref{eq:ml_work_optimal}.

The algebraic estimate is applied to the real transformed integrand
$\tilde g_\ell$ defined in~\eqref{eq:weighted_integrand}. For the
root-exponential estimate, analyticity is instead imposed on the underlying
complex transformed integrand
\begin{equation}
\tilde f_\ell(z)
:=
e^{\tilde\sigma z}
\frac{e^{-rT}}{2\pi}
\Phi_\ell(z+\mathrm{i}R)\widehat P(z+\mathrm{i}R),
\qquad
\tilde g_\ell(u)=\Re\tilde f_\ell(u),\quad u\geq0.
\label{eq:complex_transformed_integrand}
\end{equation}
Both estimates are stated at a fixed level $\ell$, hence for a fixed
$\Delta t_\ell$.
}

\begin{assumption}[Regularity of the transformed integrand]
\label{ass:weighted_smoothness}
Let $\tilde\sigma>0$ and $\ell\in\mathbb N_0$. There exists an integer
$s_\ell\geq1$ such that
$$
\tilde g_\ell^{(s_\ell-1)}
\in
AC_{\mathrm{loc}}(0,\infty),
\qquad
\int_0^\infty
u^{s_\ell/2}
\left|\tilde g_\ell^{(s_\ell)}(u)\right|
e^{-\tilde\sigma u}\,du
<\infty .
$$
\end{assumption}

\begin{theorem}[\cite{mastroianni1994error}]
\label{thm:algebraic_error}
Let Assumption~\ref{ass:weighted_smoothness} hold at level $\ell$.
Then there exists a constant $C_{s_\ell,\tilde\sigma}>0$, independent
of $N$, such that
\begin{equation}
\left|R_N^{\tilde\sigma}(g_\ell)\right|
\leq
A_\ell\,N^{-s_\ell/2},
\qquad
N\geq1,
\label{eq:algebraic_error}
\end{equation}
where
$$
A_\ell
:=
C_{s_\ell,\tilde\sigma}
\int_0^\infty
u^{s_\ell/2}
\left|\tilde g_\ell^{(s_\ell)}(u)\right|
e^{-\tilde\sigma u}\,du .
$$
\end{theorem}

The estimate in~\cite{mastroianni1994error} is stated for the standard
Laguerre weight $e^{-u}$. The form~\eqref{eq:algebraic_error} follows by
rescaling the nodes and weights as in~\eqref{eq:scaled_GL_rule}; the fixed
factors introduced by this rescaling are absorbed into
$C_{s_\ell,\tilde\sigma}$. We let $s_\ell$ be the {smoothness index} and $A_\ell$
the constant in the error estimate. \blue{The second estimate replaces the
finite regularity of Assumption~\ref{ass:weighted_smoothness} by analyticity
and growth conditions on the complex transformed integrand $\tilde f_\ell$,
and yields a rate in $N$ faster than any algebraic one.}

{\color{blue}
\begin{assumption}[Analyticity of the complex transformed integrand]
\label{ass:root_exponential}
Let $R\in\delta_V$, $\tilde\sigma>0$, and $\ell\in\mathbb N_0$. Assume that
$\tilde f_\ell$ defined in~\eqref{eq:complex_transformed_integrand} is
analytic in the interior of a parabola containing the positive real axis and
continuous on its boundary. Assume further that there exist constants
$C_{\mathrm{an},\ell}<\infty$ and $m_\ell\in\mathbb R$ such that
$$
\left|\tilde f_\ell(z)\right|
\leq
C_{\mathrm{an},\ell}\,|z|^{m_\ell}
$$
as $|z|\to\infty$ in the interior of that parabola.
\end{assumption}

Define the corresponding complex quadrature remainder by
\begin{equation}
\mathcal R_N^{\tilde\sigma}(\tilde f_\ell)
:=
2\int_0^\infty
e^{-\tilde\sigma u}\tilde f_\ell(u)\,du
-
2\sum_{n=1}^N w_n\tilde f_\ell(u_n).
\label{eq:complex_quadrature_remainder}
\end{equation}
Since $\tilde g_\ell(u)=\Re\tilde f_\ell(u)$ and the quadrature weights are
real,
\begin{equation}
R_N^{\tilde\sigma}(g_\ell)
=
\Re\!\left[
\mathcal R_N^{\tilde\sigma}(\tilde f_\ell)
\right].
\label{eq:real_complex_remainder_relation}
\end{equation}

\begin{theorem}[\cite{wang2023laguerre}]
\label{thm:root_exponential}
Let Assumption~\ref{ass:root_exponential} hold at level $\ell$. Then there
exist constants $A_\ell^{\mathrm{re}}<\infty$ and $B_\ell>0$ such that
\begin{equation}
\left|\mathcal R_N^{\tilde\sigma}(\tilde f_\ell)\right|
\leq
A_\ell^{\mathrm{re}}\exp\!\left(-B_\ell\sqrt N\right),
\qquad
N\geq1 .
\label{eq:complex_root_exponential_estimate}
\end{equation}
Consequently,
\begin{equation}
\left|R_N^{\tilde\sigma}(g_\ell)\right|
\leq
A_\ell^{\mathrm{re}}\exp\!\left(-B_\ell\sqrt N\right),
\qquad
N\geq1 .
\label{eq:root_exponential_estimate}
\end{equation}
\end{theorem}

The estimate in Theorem~6.3 of~\cite{wang2023laguerre} is stated for the
standard Laguerre weight $e^{-u}$ and for all sufficiently large $N$.
Rescaling the nodes and weights as in~\eqref{eq:scaled_GL_rule} gives
\eqref{eq:complex_root_exponential_estimate}, with the fixed factors absorbed
into $A_\ell^{\mathrm{re}}$ and $B_\ell$. For each fixed level $\ell$,
$A_\ell^{\mathrm{re}}$ can be increased to cover the finitely many remaining
positive integers. The real estimate~\eqref{eq:root_exponential_estimate}
then follows from~\eqref{eq:real_complex_remainder_relation}. Such an
analyticity condition is consistent with those used for contour deformations
in Fourier pricing~\cite{boyarchenko2025fast}, but is not available in closed
form for the rough Heston model and is therefore assumed. The parabola and
the growth condition depend on $\tilde\sigma$ and the damping parameter $R$.
}

\subsubsection{\blue{Numerical parameter selection for the single-level method}}
\label{subsubsec:sl_numerical_parameter_selection}

\paragraph{\blue{Scaling factor.}}

The nodes of the scaled Gauss-Laguerre quadrature rule~\eqref{eq:scaled_GL_rule} are those of the standard
Gauss-Laguerre quadrature rule divided by $\tilde\sigma$, so $\tilde\sigma$ determines the range of values of $u$
at which {the integrand is evaluated}. We therefore \blue{choose $\tilde\sigma$ heuristically} to reflect the estimated decay
rate of $g$, which is primarily determined by the characteristic function. In addition, by the explicit
formula for $\widehat P$ in Section~\ref{sec:models_used}, we know that the payoff
transform contributes an algebraic factor and \blue{does not affect the exponential decay rate} \cite{lewis2001simple}.
As $|\xi|\to\infty$ along the contour, it is conjectured in~\cite{boyarchenko2025fast}
that the solution of~\eqref{eq:volterra_form} behaves as
$h(\xi,\tau)=h_\infty\xi+\mathcal O(1)$ with
$h_\infty=-(\mathrm i\rho+\sqrt{1-\rho^2})/(\gamma\nu)$, and that the {characteristic function
exponent}~\eqref{eq:terminal_characteristic_exponent} satisfies
$\bigl(G(\xi)-\mathrm i\xi(X_0+rT)\bigr)/\xi\to-c_\infty(T)$ with
\begin{equation}
\operatorname{Re}c_\infty(T)
=
\frac{\sqrt{1-\rho^2}}{\gamma\nu}
\left(
\gamma\theta T+\frac{V_0\,T^{1-\alpha}}{\Gamma(2-\alpha)}
\right) ,
\label{eq:c_infty}
\end{equation}
under which $|\Phi(u+\mathrm iR)|$ decays like
$e^{-\operatorname{Re}c_\infty(T)u}$. This conjectured asymptotic behaviour is
supported by the numerical experiments in~\cite{boyarchenko2025fast} and motivates the choice
\begin{equation}
\tilde\sigma
:=
\operatorname{Re}c_\infty(T) .
\label{eq:sigma_choice}
\end{equation}
\red{For $|\rho|<1$, this choice aligns the exponential weight with the
estimated large-$u$ decay of the Fourier integrand. At $|\rho|=1$, the scale
degenerates to zero; these endpoint cases are not considered here and would
require a separate positive scale.}
Relative to this estimate, the standard Gauss-Laguerre quadrature rule can leave $\tilde g$
exponentially growing at short maturities and can place its nodes over an unnecessarily wide range of
$u$ at long maturities. The resulting loss of accuracy
at short maturities and at strikes far from the spot price is documented
in~\cite{boyarchenko2025fast}. {Section~\ref{subsec:quadrature_numerics} provides corresponding numerical evidence for the \blue{parameter sets} considered here.}

\paragraph{\blue{Discretization level.}}
By Assumption~\ref{ass:fourier_integrand_rate}, the discretization error
satisfies $e_{\mathrm{disc}}(L)\leq2C_g\Delta t_L^{\,p}$, so the first
constraint in~\eqref{eq:tolerance_allocation} is satisfied whenever
$\Delta t_L\leq(\varepsilon_{\mathrm{disc}}/2C_g)^{1/p}$. Since
$\Delta t_\ell=2^{-\ell}\Delta t_0$, this is equivalent to
$$
\ell
\geq
\frac{1}{p}\log_2\!\left(
\frac{2C_g\Delta t_0^{\,p}}{\varepsilon_{\mathrm{disc}}}
\right),
$$
and the smallest required theoretical level is
\begin{equation}
{
\left\lceil
\frac{1}{p}\log_2\!\left(
\frac{2C_g\Delta t_0^{\,p}}{\varepsilon_{\mathrm{disc}}}
\right)
\right\rceil
},
\label{eq:theoretical_level_bound}
\end{equation}
where $\lceil x\rceil$ denotes the smallest integer greater than or equal to $x \in \mathbb{R}$. The constant $C_g$ of
Assumption~\ref{ass:fourier_integrand_rate} is not available, so
\eqref{eq:theoretical_level_bound} is used in the complexity analysis rather
than in the computations. \blue{To derive a practical rule for selecting the
level $L$, we instead use a Richardson estimator based on the difference
between two consecutive levels.}
{\color{blue}
By the definitions of $V$ and $V_\ell$
in~\eqref{eq:half_line_fourier_integral} and
\eqref{eq:level_approximation}, respectively,
\[
|V_\ell-V|
=
2\left|\int_0^\infty(g-g_\ell)(u)\,du\right|
\leq
2\|g-g_\ell\|_{L^1(0,\infty)},
\]
However, this bound only shows that the price error is at most of order
$\Delta t_\ell^{\,p}$; it does not guarantee that the price error is
asymptotically proportional to $\Delta t_\ell^{\,p}$. Positive and negative
parts of $g-g_\ell$ may cancel after integration, producing a smaller leading
price error. The following assumption rules out this cancellation and provides
the leading-order expansion required by the Richardson estimator.
}

{
\begin{assumption}[Sharp price-discretization expansion]
\label{ass:sharp_price_expansion}
There exists a constant $c_V\neq0$ such that
\begin{equation}
V_\ell-V
=
c_V\Delta t_\ell^{\,p}
+
o\!\left(\Delta t_\ell^{\,p}\right),
\qquad
\ell\to\infty,
\label{eq:sharp_expansion}
\end{equation}
where $p>0$ is the exponent in
Assumption~\ref{ass:fourier_integrand_rate}.
Thus, the price has the same leading convergence order as the $L^1$ integrand
error, with no cancellation of its leading contribution after integration.
\end{assumption}

\blue{Assumption~\ref{ass:sharp_price_expansion} is stronger than
Assumption~\ref{ass:fourier_integrand_rate} and is used only to justify the
Richardson estimator; it is not used in
Propositions~\ref{prop:work_scaling_single}--\ref{prop:work_scaling_multilevel}.}

We provide numerical evidence for
Assumption~\ref{ass:sharp_price_expansion} in
Section~\ref{subsec:time_discretization_numerics}.
\blue{Under this assumption, Proposition~\ref{prop:richardson} relates the
discretization error to the difference between two consecutive levels and
thereby yields a computable rule for selecting $L$.}
}

\begin{proposition}[Asymptotic Richardson error estimate]
\label{prop:richardson}
{
Let $\Delta t_\ell=2^{-\ell}\Delta t_0$ and suppose that
Assumption~\ref{ass:sharp_price_expansion} holds.
Define the price error and its Richardson estimator by
\begin{equation}
E_\ell^{\,V}
:=
e_{\mathrm{disc}}(\ell)
=
|V-V_\ell|,
\qquad
\widehat E_\ell^{\,V}
:=
\frac{|V_\ell-V_{\ell-1}|}{2^{\,p}-1}.
\label{eq:richardson_price_notation}
\end{equation}
Then
\begin{equation}
E_\ell^{\,V}
=
\widehat E_\ell^{\,V}
\left(1+o(1)\right),
\qquad
\ell\to\infty .
\label{eq:richardson_error_estimate}
\end{equation}
}
\end{proposition}
\begin{proof}
{We prove Proposition~\ref{prop:richardson} in
Appendix~\ref{sec:proof_richardson}.}
\end{proof}
For the computable approximations, define the {signed quadrature error}
\begin{equation}
q_{N,\ell}
:=
V_{N,\ell}-V_\ell.
\label{eq:quadrature_perturbation}
\end{equation}
{\color{blue}
Suppose that $V_{N,\ell}$ and $V_{N,\ell-1}$ are evaluated with the same
scaled Gauss--Laguerre rule, that is, with the same $N$, scaling parameter
$\tilde\sigma$, nodes, and weights. By linearity of the quadrature rule, the
quadrature perturbation of the level difference is
\begin{equation}
\begin{aligned}
e_{N,\ell}^{\Delta}
&:=
\left|
(V_{N,\ell}-V_{N,\ell-1})-(V_\ell-V_{\ell-1})
\right| \\
&=
\left|q_{N,\ell}-q_{N,\ell-1}\right|
=
\left|R_N^{\tilde\sigma}(g_\ell-g_{\ell-1})\right|.
\end{aligned}
\label{eq:level_difference_quadrature_perturbation}
\end{equation}
Consequently, the reverse triangle inequality gives
\begin{equation}
\bigl|
|V_{N,\ell}-V_{N,\ell-1}|
-|V_\ell-V_{\ell-1}|
\bigr|
\leq
e_{N,\ell}^{\Delta}
=
\left|R_N^{\tilde\sigma}(g_\ell-g_{\ell-1})\right|.
\label{eq:absolute_level_difference_perturbation}
\end{equation}
Suppose that, for the chosen quadrature order $N$, there exists a constant
$C_{Q,N}<\infty$, independent of $\ell$, such that
\begin{equation}
\left|R_N^{\tilde\sigma}(g_\ell-g_{\ell-1})\right|
\leq
C_{Q,N}\Delta t_\ell^{\,p}.
\label{eq:quadrature_condition_order_p}
\end{equation}
Define the computable Richardson estimator by
\begin{equation}
\widehat E_{N,\ell}^{\,V}
:=
\frac{|V_{N,\ell}-V_{N,\ell-1}|}{2^p-1}.
\label{eq:computable_richardson_estimator}
\end{equation}
Proposition~\ref{prop:richardson} and
\eqref{eq:absolute_level_difference_perturbation} then give
\begin{equation}
\left|
e_{\mathrm{disc}}(\ell)-\widehat E_{N,\ell}^{\,V}
\right|
\leq
\frac{C_{Q,N}}{2^p-1}\Delta t_\ell^{\,p}
+o\!\left(\Delta t_\ell^{\,p}\right).
\label{eq:computable_richardson_perturbation}
\end{equation}
Thus, the accuracy of the computable Richardson estimator
then depends on $C_{Q,N}$ being small relative to the leading level-difference
coefficient $(2^p-1)|c_V|$. The theoretical complexity analysis continues to
use the certified $L^1$ bound and level~\eqref{eq:theoretical_level_bound};
the rule below is its practical, asymptotic replacement.
}

{
Evaluating this Richardson error estimate at every level would require the
corresponding level approximations. Instead, we compute the first level
difference with \red{$\bar N^{\mathrm{comp}}$} quadrature points and define
\begin{equation}
D_1
:=
\left|
\red{V_{\bar N^{\mathrm{comp}},1}-V_{\bar N^{\mathrm{comp}},0}}
\right|.
\label{eq:pilot_richardson_difference}
\end{equation}
For $\ell\geq1$, we use the \blue{asymptotic Richardson-type discretization indicator}
\begin{equation}
\eta_{\mathrm{disc}}(\ell)
:=
\frac{D_1}{2^{\,p}-1}
2^{-p(\ell-1)}.
\label{eq:practical_disc_indicator}
\end{equation}
\blue{This asymptotic indicator} extrapolates the asymptotic rate from the first level
difference and is not a certified error bound unless levels $0$ and $1$ are
already in the asymptotic regime and \blue{the quadrature error of their level
difference is sufficiently small}.
\blue{Accordingly, $\Delta t_0$ must be chosen sufficiently small that $D_1$
is numerically stable under further time refinement.}
We select the smallest level $L\geq1$ satisfying
$\eta_{\mathrm{disc}}(L)\leq\varepsilon_{\mathrm{disc}}$, namely
\begin{equation}
L
:=
\blue{\max\!\Biggl\{1,}
{
\left\lceil
1
+\frac{1}{p}\log_2\!\left(
\frac{D_1}
{\left(2^{\,p}-1\right)\varepsilon_{\mathrm{disc}}}
\right)
\right\rceil
}
\blue{\Biggr\}}.
\label{eq:practical_level}
\end{equation}
The behaviour of this level selection rule is assessed in
Section~\ref{sec:numerical_results}.
}

\paragraph{\blue{Number of quadrature points.}}

The level $L$ is fixed by~\eqref{eq:practical_level}. The second
constraint in~\eqref{eq:tolerance_allocation} is then a separate condition on $N$.  We select $N$ differently based on the used quadrature error estimate. 

\begin{remark}[Relation between the two quadrature estimates]
\label{rem:quadrature_estimates}
We note that the two estimates are not interchangeable. If the quadrature error
satisfies~\eqref{eq:root_exponential_estimate}, then, for every fixed $s>0$, it
also satisfies an estimate of the form~\eqref{eq:algebraic_error} with a
suitable constant and for $N$ large enough. However,  the converse does not hold.
The algebraic estimate is in this sense the conservative one, and the number of
{quadrature points} it selects through~\eqref{eq:optimal_N} exceeds the number
selected by Proposition~\ref{prop:quadrature_order_root_exponential} for
$\varepsilon_{\mathrm{quad}}$ small enough. Where~\eqref{eq:root_exponential_estimate}
holds, the latter therefore gives the sharper estimate of the work required to reach
a prescribed tolerance.
\end{remark}

{
\begin{proposition}[Number of quadrature points under the algebraic error estimate]
\label{prop:quadrature_order_algebraic}
Let Assumption~\ref{ass:weighted_smoothness} hold at level $L$ and let
$\varepsilon_{\mathrm{quad}}>0$. Then every $N\geq N_{\mathrm{alg}}$
satisfies the quadrature constraint in~\eqref{eq:tolerance_allocation}, where
\begin{equation}
N_{\mathrm{alg}}
:=
{
\left\lceil
\left(
\frac{A_L}{\varepsilon_{\mathrm{quad}}}
\right)^{2/s_L}
\right\rceil
}.
\label{eq:optimal_N}
\end{equation}
Moreover, $N_{\mathrm{alg}}$ is the smallest integer {number of quadrature points} certified by
the algebraic estimate~\eqref{eq:algebraic_error} to satisfy the constraint.
\end{proposition}

\begin{proof}
{By~\eqref{eq:algebraic_error}, the quadrature constraint holds whenever
$N\geq(A_L/\varepsilon_{\mathrm{quad}})^{2/s_L}$. Therefore,
\eqref{eq:optimal_N} is the smallest integer certified by
\eqref{eq:algebraic_error}, and every larger integer also satisfies the
constraint.}
\end{proof}

Under Assumption~\ref{ass:root_exponential}, the corresponding certified
{number of quadrature points} is obtained by inverting the root-exponential error estimate.

\begin{proposition}[Number of quadrature points under the root-exponential error estimate]
\label{prop:quadrature_order_root_exponential}
Let Assumption~\ref{ass:root_exponential} hold for fixed $L$ and let
$0<\varepsilon_{\mathrm{quad}}<A_L^{\mathrm{re}}$. Then every
$N\geq N_{\mathrm{re}}$ satisfies the quadrature constraint
in~\eqref{eq:tolerance_allocation}, where
\begin{equation}
N_{\mathrm{re}}
:=
\left\lceil
\frac{1}{B_L^2}
\log^2\!\left(
\frac{A_L^{\mathrm{re}}}{\varepsilon_{\mathrm{quad}}}
\right)
\right\rceil .
\label{eq:optimal_N_root_exponential}
\end{equation}
Moreover, $N_{\mathrm{re}}$ is the smallest integer {number of quadrature points} certified by
the root-exponential estimate~\eqref{eq:root_exponential_estimate} to satisfy
the constraint.
\end{proposition}

\begin{proof}
By~\eqref{eq:root_exponential_estimate}, the quadrature constraint holds
whenever
$A_L^{\mathrm{re}}\exp(-B_L\sqrt N)\leq
\varepsilon_{\mathrm{quad}}$. Since
$0<\varepsilon_{\mathrm{quad}}<A_L^{\mathrm{re}}$, this sufficient condition
is equivalent to
$N\geq B_L^{-2}\log^2(A_L^{\mathrm{re}}/\varepsilon_{\mathrm{quad}})$.
Therefore,~\eqref{eq:optimal_N_root_exponential} is the smallest integer
certified by~\eqref{eq:root_exponential_estimate}, and every larger integer
also satisfies the constraint.
\end{proof}
}

{
\begin{remark}[Practical estimation of the quadrature parameters]
The theoretical selection rules~\eqref{eq:optimal_N}
and~\eqref{eq:optimal_N_root_exponential} use the quadrature parameters at the
selected finest level $L$. Estimating these parameters directly would require
repeated quadrature computations at the most expensive discretization level.
In practice, we instead estimate the algebraic parameters $A_0$ and $s_0$ and
the root-exponential parameters $A_0^{\mathrm{re}}$ and $B_0$ at level zero,
and use them as surrogates for $A_L$, $s_L$, $A_L^{\mathrm{re}}$, and $B_L$,
respectively. Section~\ref{subsubsec:quadrature_assumptions}, in particular
Figure~\ref{fig:single_level_cover_fits}, provides numerical evidence that
these parameters vary only weakly across the discretization levels considered.
This replacement is an empirical implementation choice and is not part of the
certified bounds in Propositions~\ref{prop:quadrature_order_algebraic}
and~\ref{prop:quadrature_order_root_exponential}, which retain the actual
level-$L$ parameters.
\end{remark}
}

As
$\varepsilon_{\mathrm{quad}}\to0$, the {number of quadrature points selected by}~\eqref{eq:optimal_N} grows
algebraically in $1/\varepsilon_{\mathrm{quad}}$, whereas
\eqref{eq:optimal_N_root_exponential} grows as
$\log^2(1/\varepsilon_{\mathrm{quad}})$.

\subsubsection{\blue{Computational complexity}}

We study how the computational work required by the proposed single-level method scales as $\varepsilon\to0$.

\begin{proposition}[Complexity of the single-level method under the algebraic
quadrature error estimate]
\label{prop:work_scaling_single}
Let Assumption~\ref{ass:fourier_integrand_rate} hold. Suppose
that there exist an integer {$s_{\mathrm{SL}}\geq1$}, a constant $A<\infty$ and
$\varepsilon_0>0$ such that, for every
$\varepsilon\in(0,\varepsilon_0]$, Assumption~\ref{ass:weighted_smoothness}
holds at level $L(\varepsilon)$ with
$$
{s_{L(\varepsilon)}\geq s_{\mathrm{SL}}},
\qquad
A_{L(\varepsilon)}\leq A.
$$
{
Suppose also that, for some constants $C_W<\infty$ and $\beta>0$, the work
of one integrand evaluation at level $\ell$ satisfies
\[
W_\ell\leq C_W\Delta t_\ell^{-\beta},
\qquad \ell\geq0.
\]
}
Let $L$ be the smallest level given by
\eqref{eq:theoretical_level_bound}, and let {$N=N_{\mathrm{alg}}$} be given by
Proposition~\ref{prop:quadrature_order_algebraic}, with
$\varepsilon_{\mathrm{disc}}
=\varepsilon_{\mathrm{quad}}=\varepsilon/2$. Then
$|V-V_{N,L}|\leq\varepsilon$ and
\begin{equation}
W_{\mathrm{SL}}(\varepsilon)
:=
W_{\mathrm{SL}}(N,L)
=
\mathcal O\!\left(
{\varepsilon^{-\left(\frac{\beta}{p}+\frac{2}{s_{\mathrm{SL}}}\right)}}
\right),
\qquad
\varepsilon\to0 .
\label{eq:work_scaling_single}
\end{equation}
\end{proposition}

\begin{proof}
{We prove Proposition~\ref{prop:work_scaling_single} in
Appendix~\ref{app:complexity_proof_single}.}
\end{proof}

{
\begin{proposition}[Complexity of the single-level method under the root-exponential
quadrature error estimate]
\label{prop:work_scaling_single_root_exponential}
Let Assumption~\ref{ass:fourier_integrand_rate} hold and let
$\varepsilon_{\mathrm{disc}}=\varepsilon_{\mathrm{quad}}=\varepsilon/2$.
Let $L$ be the smallest level given by~\eqref{eq:theoretical_level_bound}.
Assume that Assumption~\ref{ass:root_exponential} holds at level $L$ and that
$0<\varepsilon_{\mathrm{quad}}<A_L^{\mathrm{re}}$. Let
$N=N_{\mathrm{re}}$ be given by
Proposition~\ref{prop:quadrature_order_root_exponential}. Then
$|V-V_{N,L}|\leq\varepsilon$.

For the asymptotic work estimate, suppose in addition that there exist
constants $A^{\mathrm{re}}<\infty$, $B>0$, and $\varepsilon_0>0$ such that,
for every $\varepsilon\in(0,\varepsilon_0]$, the constants at the selected
level $L=L(\varepsilon)$ satisfy
\[
A_{L(\varepsilon)}^{\mathrm{re}}
\leq
A^{\mathrm{re}},
\qquad
B_{L(\varepsilon)}
\geq
B.
\]
{
Assume also that, for some constants $C_W<\infty$ and $\beta>0$, the work
of one integrand evaluation at level $\ell$ satisfies
\[
W_\ell\leq C_W\Delta t_\ell^{-\beta},
\qquad \ell\geq0.
\]
}
Then
\begin{equation}
W_{\mathrm{SL}}(\varepsilon)
:=
W_{\mathrm{SL}}(N,L)
=
\mathcal O\!\left(
\varepsilon^{-\beta/p}\log^2(1/\varepsilon)
\right),
\qquad
\varepsilon\to0 .
\label{eq:work_root_exponential}
\end{equation}
\end{proposition}
}

\begin{proof}
{We prove
Proposition~\ref{prop:work_scaling_single_root_exponential} in
Appendix~\ref{app:complexity_proof_single_root_exponential}.}
\end{proof}

\subsection{Multilevel Fourier Gauss-Laguerre  Quadrature Method}
\label{sec:ML_method}

In this section, we design a multilevel Fourier pricing method based on the
scaled Gauss-Laguerre quadrature rule for the option value~\eqref{eq:fourier_pricing}. The
single-level method of Section~\ref{sec:SL_method} evaluates the integrand at
every quadrature {point} at the level $L$ selected by~\eqref{eq:practical_level},
so all $N$ solves of~\eqref{eq:volterra_form} are performed at the finest and
most expensive level. The multilevel method instead splits the integrand into a
{level zero term} and a hierarchy of level differences, and applies the scaled
Gauss-Laguerre quadrature rule to each. The {number of quadrature points} is then chosen
separately at each level such that the
total work is minimized by distributing the {quadrature points} across the
hierarchy, placing few of them at the levels where one evaluation is expensive, and more quadrature points where the fractional Riccati equation is solved at a coarser mesh.
The level $L$ remains given by~\eqref{eq:practical_level}, and the quadrature
tolerance $\varepsilon_{\mathrm{quad}}$ of~\eqref{eq:tolerance_allocation} is
distributed among the levels. We state the assumptions on the regularity and
the level decay of the level differences, derive the allocation of the
{quadrature points}, and obtain the resulting computational complexity in
Proposition~\ref{prop:work_scaling_multilevel}.

We first write the integrand at level $L$ as a telescoping sum. For
$\ell\geq1$ we define the level differences
\begin{equation}
\Delta g_\ell
:=
g_\ell-g_{\ell-1},
\qquad
\Delta\tilde g_\ell
:=
\tilde g_\ell-\tilde g_{\ell-1} ,
\label{eq:level_difference}
\end{equation}
with $\tilde g_\ell$ given by~\eqref{eq:weighted_integrand}, and we set
{the level-zero integrands $g_0$ and $\tilde g_0$ apart from the level
differences.} Derivatives of the level differences are denoted for $m \in \mathbb{N}$  by
$\left(\Delta\tilde g_\ell\right)^{(m)}
=\tilde g_\ell^{(m)}-\tilde g_{\ell-1}^{(m)}$.
By definition of \eqref{eq:level_difference},
\begin{equation}
{
g_L
=
g_0+\sum_{\ell=1}^L\Delta g_\ell .
}
\label{eq:telescoping_g}
\end{equation}
Let $\mathbf N:=(N_0,\ldots,N_L)$, with $N_\ell\geq1$ for every $\ell$.
{Applying the scaled Gauss-Laguerre quadrature
rule~\eqref{eq:scaled_GL_rule} with $N_0$ quadrature points to $g_0$ and with
$N_\ell$ quadrature points to $\Delta g_\ell$ for $\ell\geq1$,} we define the
multilevel approximation
\begin{equation}
{
V_{\mathbf N,L}
:=
\mathcal Q_{N_0}^{\tilde\sigma}[g_0]
+
\sum_{\ell=1}^L
\mathcal Q_{N_\ell}^{\tilde\sigma}[\Delta g_\ell] .
}
\label{eq:ml_approximation}
\end{equation}
For $\ell\geq1$, evaluating $\Delta g_\ell$ at one {quadrature point}
requires one solve of~\eqref{eq:volterra_form} at level $\ell$ and one at level
$\ell-1$, at a cost of $W_\ell+W_{\ell-1}$. {The level-zero term costs
$W_0$ per quadrature point.} The computational work
required by the multilevel method is
\begin{equation}
{
W_{\mathrm{ML}}(\mathbf N,L)
=
W_0N_0
+
\sum_{\ell=1}^L
\left(W_\ell+W_{\ell-1}\right)N_\ell .
}
\label{eq:ml_work}
\end{equation}

\subsubsection{\blue{Error analysis and practical parameter selection for the multilevel method}}
\label{subsubsec:ml_error_parameter_selection}

We now split the error of~\eqref{eq:ml_approximation} into the two
contributions of~\eqref{eq:error_decomposition}. Adding and subtracting $V_L$
gives
\begin{equation}
\left|V-V_{\mathbf N,L}\right|
\;\leq\;
\underbrace{\left|V-V_L\right|}_{=\;e_{\mathrm{disc}}(L)}
\;+\;
\underbrace{\left|V_L-V_{\mathbf N,L}\right|}_{=:\;e_{\mathrm{quad}}(\mathbf N,L)} .
\label{eq:ml_error_decomposition}
\end{equation}
The discretization error is unchanged, and $L$ is therefore selected
by~\eqref{eq:practical_level}. The Richardson estimate of
Proposition~\ref{prop:richardson}, together with the condition on the
quadrature errors stated after that proposition, applies without modification.
By~\eqref{eq:level_approximation}, \eqref{eq:telescoping_g},
\eqref{eq:ml_approximation} and~\eqref{eq:quadrature_error},
$$
{
V_L-V_{\mathbf N,L}
=
R_{N_0}^{\tilde\sigma}(g_0)
+
\sum_{\ell=1}^L
R_{N_\ell}^{\tilde\sigma}(\Delta g_\ell) .
}
$$
hence
\begin{equation}
{
e_{\mathrm{quad}}(\mathbf N,L)
\leq
|R_{N_0}^{\tilde\sigma}(g_0)|
+
\sum_{\ell=1}^L
\left|R_{N_\ell}^{\tilde\sigma}(\Delta g_\ell)\right| .
}
\label{eq:ml_quadrature_split}
\end{equation}
We impose on the level differences the regularity required by
Theorem~\ref{thm:algebraic_error}.

\begin{assumption}[Regularity of the level differences]
\label{ass:level_difference_regularity}
{Let $\tilde\sigma>0$ and $\ell\in\mathbb N$. There exists an integer
$s_\ell\geq1$ such that
$$
\left(\Delta\tilde g_\ell\right)^{(s_\ell-1)}
\in
AC_{\mathrm{loc}}(0,\infty),
\qquad
\int_0^\infty
u^{s_\ell/2}
\left|\left(\Delta\tilde g_\ell\right)^{(s_\ell)}(u)\right|
e^{-\tilde\sigma u}\,du
<\infty .
$$}
\end{assumption}

{At level zero, Assumption~\ref{ass:weighted_smoothness} with smoothness
index $s_0$ and Theorem~\ref{thm:algebraic_error} give
\[
\left|R_N^{\tilde\sigma}(g_0)\right|
\leq
A_0N^{-s_0/2}.
\]
For $\ell\geq1$, Assumption~\ref{ass:level_difference_regularity} is
Assumption~\ref{ass:weighted_smoothness} applied to the integrand
$\Delta g_\ell$. Hence, Theorem~\ref{thm:algebraic_error} gives}
\begin{equation}
{
\left|R_N^{\tilde\sigma}(\Delta g_\ell)\right|
\leq
A_\ell\,N^{-s_\ell/2},
\quad
A_\ell
:=
C_{s_\ell,\tilde\sigma}
\int_0^\infty
u^{s_\ell/2}
\left|\left(\Delta\tilde g_\ell\right)^{(s_\ell)}(u)\right|
e^{-\tilde\sigma u}\,du,
\qquad
\ell\geq1.
}
\label{eq:level_difference_estimate}
\end{equation}

{The allocation below uses a common smoothness exponent for the level
differences. We therefore work with an integer $s\geq1$ satisfying}
\begin{equation}
{
s\leq s_\ell,
\qquad
\ell=1,\ldots,L .
}
\label{eq:common_order}
\end{equation}
{The largest possible value is
$s=\min_{1\leq\ell\leq L}s_\ell$. Since $N_\ell\geq1$,
\eqref{eq:level_difference_estimate} and~\eqref{eq:common_order} give}
\begin{equation}
{
\left|R_{N_\ell}^{\tilde\sigma}(\Delta g_\ell)\right|
\leq
A_\ell\,N_\ell^{-s/2},
\qquad
\ell=1,\ldots,L .
}
\label{eq:common_order_estimate}
\end{equation}
{For a prescribed tolerance, \eqref{eq:common_order} involves the
finitely many correction levels $1,\ldots,L$. The complexity statement
concerns $L\to\infty$ and requires in addition that $s$ can be chosen
independently of $L$. Section~\ref{subsubsec:quadrature_assumptions}, in
particular Figure~\ref{fig:correction_algebraic_cover}, provides numerical
evidence that the fitted smoothness indices of the level differences are
approximately constant across the levels considered, supporting the use of a
common correction index $s$. The level-zero index $s_0$ is treated separately
and is examined in Figure~\ref{fig:single_level_cover_fits}. For the complexity
analysis, we impose an additional assumption on the level dependence of the
constants $A_\ell$.}

\begin{assumption}[Decay of the quadrature constants of the level differences]
\label{ass:level_difference_decay}
There exists a constant $C_A<\infty$, independent of $\ell$, such that the
constants $A_\ell$ of~\eqref{eq:level_difference_estimate} satisfy
$$
A_\ell
\leq
C_A\,\Delta t_\ell^{\,p} ,
\qquad
\ell\geq1 .
$$
\end{assumption}

Assumption~\ref{ass:level_difference_decay} is separate from
Assumption~\ref{ass:fourier_integrand_rate}. The latter controls
$\|g-g_\ell\|_{L^1(0,\infty)}$, whereas $A_\ell$ is defined
in~\eqref{eq:level_difference_estimate} by a weighted seminorm of the
derivative of order $s_\ell$ of $\Delta\tilde g_\ell$.
{Section~\ref{subsubsec:quadrature_assumptions}, in particular
Figure~\ref{fig:correction_algebraic_cover}, provides numerical evidence that
the fitted constants $A_\ell$ decay with an exponent close to $p$, supporting
Assumption~\ref{ass:level_difference_decay} for the levels considered.}
\paragraph{Allocation of {quadrature points across levels}.}
We now select the {number of quadrature points at each level}.
By~\eqref{eq:ml_quadrature_split} and~\eqref{eq:common_order_estimate}, the
second constraint in~\eqref{eq:tolerance_allocation} is satisfied if we
{allocate one half of $\varepsilon_{\mathrm{quad}}$ to the level-zero
term and one half to the level differences, and impose}
\begin{equation}
{
\begin{aligned}
A_0N_0^{-s_0/2}
&\leq
\frac{\varepsilon_{\mathrm{quad}}}{2},
\\
\sum_{\ell=1}^L
A_\ell N_\ell^{-s/2}
&\leq
\frac{\varepsilon_{\mathrm{quad}}}{2}.
\end{aligned}
}
\label{eq:ml_quadrature_constraint}
\end{equation}
{These two constraints imply
$e_{\mathrm{quad}}(\mathbf N,L)\leq\varepsilon_{\mathrm{quad}}$.}
The entries of $\mathbf N$ are integers, so the minimization of the
work~\eqref{eq:ml_work} subject to~\eqref{eq:ml_quadrature_constraint} is an
integer optimization problem. We relax it and minimize over  the set of real positive numbers 
$N_0,\ldots,N_L \in \mathbb{R}_+$ instead. Proposition~\ref{prop:optimal_N_ml} solves the
relaxed problem in closed form, and we then set
$N_\ell=\lceil N_\ell^\star\rceil$.

\begin{proposition}[Allocation of {quadrature points across levels}]
\label{prop:optimal_N_ml}
{Let Assumption~\ref{ass:weighted_smoothness} hold at level zero with
smoothness index $s_0$, let
Assumption~\ref{ass:level_difference_regularity} hold for
$\ell=1,\ldots,L$, let $s$ satisfy~\eqref{eq:common_order}, and let
$\varepsilon_{\mathrm{quad}}>0$. Assume that $A_0>0$ and $A_\ell>0$ for
$\ell=1,\ldots,L$. The minimizer of~\eqref{eq:ml_work} over
$N_0,\ldots,N_L>0$ subject to the split
constraints~\eqref{eq:ml_quadrature_constraint} is given by}
\begin{equation}
{
\begin{aligned}
N_0^\star
&=
\left(
\frac{2A_0}{\varepsilon_{\mathrm{quad}}}
\right)^{2/s_0},
\\
N_\ell^\star
&=
\left(
\frac{A_\ell}{W_\ell+W_{\ell-1}}
\right)^{\frac{2}{s+2}}
\left(
\frac{2}{\varepsilon_{\mathrm{quad}}}
\sum_{k=1}^L
A_k^{\frac{2}{s+2}}
\left(W_k+W_{k-1}\right)^{\frac{s}{s+2}}
\right)^{\frac{2}{s}},
\qquad
\ell=1,\ldots,L .
\end{aligned}
}
\label{eq:optimal_N_ml}
\end{equation}
\end{proposition}

\begin{proof}
{We prove Proposition~\ref{prop:optimal_N_ml} in
Appendix~\ref{app:optimal_N_ml}.}
\end{proof}
{The two constraints in~\eqref{eq:ml_quadrature_constraint} are active
at the minimizer. Writing
$\mathbf N^\star:=(N_0^\star,\ldots,N_L^\star)$, the resulting computational
work is}
\begin{equation}
{
\begin{aligned}
W_{\mathrm{ML}}(\mathbf N^\star,L)
={}&
W_0
\left(
\frac{2A_0}{\varepsilon_{\mathrm{quad}}}
\right)^{2/s_0}
\\
&+
\left(
\frac{2}{\varepsilon_{\mathrm{quad}}}
\right)^{2/s}
\left(
\sum_{\ell=1}^L
A_\ell^{\frac{2}{s+2}}
\left(W_\ell+W_{\ell-1}\right)^{\frac{s}{s+2}}
\right)^{\frac{s+2}{s}} .
\end{aligned}
}
\label{eq:ml_work_optimal}
\end{equation}
In practice we take $N_\ell=\lceil N_\ell^\star\rceil$.
\red{This upward rounding preserves the split error constraints but is not
claimed to solve the exact integer optimization problem.}

\blue{\subsubsection{Computational complexity.}}

The allocation~\eqref{eq:optimal_N_ml} is valid for any prescribed
tolerance. {For the complexity estimate, we additionally assume
$s>2p/\beta$. Under this condition, the finest correction level determines
the asymptotic order of the correction contribution
in~\eqref{eq:ml_work_optimal}. The separate level-zero contribution has order
$\mathcal O(\varepsilon^{-2/s_0})$. \blue{For the direct fractional Adams
scheme, $\beta=2$, so the condition becomes $s>p$. The numerical evidence in
Section~\ref{subsubsec:quadrature_assumptions} indicates that the fitted
effective regularity indices are above the corresponding thresholds for the
parameter sets and finite-range convergence rates considered.}}

\begin{proposition}[Complexity of the multilevel method under the algebraic
quadrature error estimate]
\label{prop:work_scaling_multilevel}
{Let Assumption~\ref{ass:weighted_smoothness} hold at level zero with
smoothness index $s_0$, and let Assumptions~\ref{ass:fourier_integrand_rate},
\ref{ass:level_difference_regularity}, and
\ref{ass:level_difference_decay} hold.} Suppose that the costs satisfy
{
\begin{equation}
W_\ell\leq C_W\Delta t_\ell^{-\beta},
\qquad \ell\geq0,
\label{eq:multilevel_cost_bound}
\end{equation}
}
{for some $C_W<\infty$ and $\beta>0$. Suppose in addition that there
exists an integer $s$,
independent of $L$, satisfying~\eqref{eq:common_order} and
$s>2p/\beta$. Assume also that $A_0>0$ and $A_\ell>0$ for
$\ell=1,\ldots,L$.} Let
$\varepsilon_{\mathrm{disc}}=\varepsilon_{\mathrm{quad}}=\varepsilon/2$,
let $L$ be the smallest level given by~\eqref{eq:theoretical_level_bound}, and set
$N_\ell=\lceil N_\ell^\star\rceil$, where $N_\ell^\star$ is given
by~\eqref{eq:optimal_N_ml}. Then
\[
|V-V_{\mathbf N,L}|\leq\varepsilon
\]
{and}
\begin{equation}
{
W_{\mathrm{ML}}(\mathbf N,L)
=
\mathcal O\!\left(
\varepsilon^{-2/s_0}
+
\varepsilon^{-\beta/p}
\right),
\qquad
\varepsilon\to0 .
}
\label{eq:work_scaling_multilevel_general}
\end{equation}
{If, in addition,
\[
s_0\geq\frac{2p}{\beta},
\]
then}
\begin{equation}
{
W_{\mathrm{ML}}(\mathbf N,L)
=
\mathcal O\!\left(\varepsilon^{-\beta/p}\right),
\qquad
\varepsilon\to0.
}
\label{eq:work_scaling_multilevel}
\end{equation}
\end{proposition}

We prove Proposition~\ref{prop:work_scaling_multilevel} in
Appendix~\ref{app:complexity_proof_multilevel}.

\paragraph{Allocation of $\mathbf{N}$ under the root-exponential estimate.}

{
The allocation of~\eqref{eq:optimal_N_ml} is based on the algebraic error
estimate~\eqref{eq:algebraic_error}, which we use mainly for the complexity
analysis. In the numerical implementation we instead allocate the quadrature
points based on the root-exponential quadrature error estimates observed for
the level-zero term and the level differences. Let
$A_0^{\mathrm{re}}>0$ and $B_0>0$ denote the level-zero prefactor and rate.
For the correction levels, let $A_\ell^{\mathrm{re}}>0$ and $B_\ell>0$
denote the corresponding prefactors and rates. Section~\ref{subsubsec:quadrature_assumptions},
in particular Figure~\ref{fig:correction_exponential_cover}, provides
numerical evidence that the fitted rates $B_\ell$ are approximately constant
across the correction levels considered. We therefore use a common correction
rate $B>0$, that is, $B_\ell=B$ for $\ell=1,\ldots,L$, and write the empirical
error estimates as
\begin{equation}
\begin{aligned}
\left|R_N^{\tilde\sigma}(g_0)\right|
&\leq
A_0^{\mathrm{re}}
\exp\!\left(-B_0\sqrt N\right),
\\
\left|R_N^{\tilde\sigma}(\Delta g_\ell)\right|
&\leq
A_\ell^{\mathrm{re}}
\exp\!\left(-B\sqrt N\right),
\qquad \ell=1,\ldots,L.
\end{aligned}
\label{eq:ml_root_exponential_estimate}
\end{equation}
We then determine the number of quadrature points at each level by solving
\begin{equation}
\begin{aligned}
\min_{N_0,\ldots,N_L}
\quad&
W_0N_0
+
\sum_{\ell=1}^L
\left(W_\ell+W_{\ell-1}\right)N_\ell
\\
\text{subject to}\quad&
A_0^{\mathrm{re}}
\exp\!\left(-B_0\sqrt{N_0}\right)
+
\sum_{\ell=1}^L
A_\ell^{\mathrm{re}}
\exp\!\left(-B\sqrt{N_\ell}\right)
\leq
\varepsilon_{\mathrm{quad}},
\\
&
N_0\geq1,
\qquad
N_\ell\geq1,
\quad
\ell=1,\ldots,L .
\end{aligned}
\label{eq:ml_root_exponential_optimization}
\end{equation}
We solve the
resulting problem~\eqref{eq:ml_root_exponential_optimization} numerically, and take the ceilings of the resulting values,
which preserves the constraint
in~\eqref{eq:ml_root_exponential_optimization} because each error term is
decreasing in its corresponding number of quadrature points.
}

{
\begin{remark}[Practical construction of the multilevel quadrature parameters]
\label{rem:practical_ml_quadrature_parameters}
The algebraic allocation~\eqref{eq:optimal_N_ml} and the root-exponential
optimization~\eqref{eq:ml_root_exponential_optimization} involve the
prefactors $A_\ell$ and $A_\ell^{\mathrm{re}}$ at all correction levels.
Fitting these prefactors separately would require pilot quadrature
computations at every increasingly expensive level. In practice, we estimate
the level-zero parameters $A_0$, $s_0$, $A_0^{\mathrm{re}}$, and $B_0$, and
the first-correction prefactors $A_1$ and $A_1^{\mathrm{re}}$. For
$\ell=1,\ldots,L$, we then use the propagated surrogates
\[
\widehat A_\ell
:=
A_1
\left(
\frac{\Delta t_\ell}{\Delta t_1}
\right)^p,
\qquad
\widehat A_\ell^{\mathrm{re}}
:=
A_1^{\mathrm{re}}
\left(
\frac{\Delta t_\ell}{\Delta t_1}
\right)^p.
\]
The common algebraic index $s$ and the common root-exponential rate $B$ are
chosen as described above. Figures~\ref{fig:correction_algebraic_cover}
and~\ref{fig:correction_exponential_cover} provide numerical evidence that
the correction rates are approximately level independent and that both fitted
prefactor sequences decay with an exponent close to $p$. Thus, the
quadrature parameters on the finer correction levels can be constructed
without fitting their quadrature errors separately. The propagated
prefactors are empirical implementation choices; the algebraic complexity
analysis continues to rely on Assumption~\ref{ass:level_difference_decay}.
\end{remark}
}

\section{Numerical Results}
\label{sec:numerical_results}

{In this section, we \red{provide finite-range numerical diagnostics for} the assumptions and error estimates
introduced in Section~\ref{sec:methodology} and assess the performance of the
single-level and multilevel pricing methods developed in
Sections~\ref{sec:SL_method} and~\ref{sec:ML_method}.
\red{These diagnostics do not prove the corresponding bounds on the entire
Fourier half-line.}
\blue{Section~\ref{subsec:time_discretization_numerics} studies the convergence
of the fractional Adams approximation and of the corresponding fully discrete
Fourier integrand. We compare their measured finite-range slopes with the
empirical reference rate $1+\alpha$ discussed in
Remark~\ref{rem:observed_rate}; these diagnostics do not assign values to the
generic theoretical exponents $q$, $r_T$, or $p$.} {We also test the sharp price expansion
in Assumption~\ref{ass:sharp_price_expansion} and the Richardson estimator of
Proposition~\ref{prop:richardson}.}
Section~\ref{subsec:quadrature_numerics} examines the choice of the
Gauss-Laguerre scaling factor in~\eqref{eq:sigma_choice} and the algebraic and
root-exponential quadrature error estimates of
Section~\ref{sec:SL_method}. For the multilevel method, we additionally examine
the regularity of the level differences $\Delta g_\ell$, the existence of a
common lower bound $s$ for the smoothness indices $s_\ell$
in~\eqref{eq:common_order}, and the decay of the constants $A_\ell$ defined
in~\eqref{eq:level_difference_estimate}, as assumed in
Assumption~\ref{ass:level_difference_decay}.
Finally, Section~\ref{subsec:pricing_comparison} compares the accuracy and
computational cost of the single-level and multilevel methods, relates the
observed work to the complexity results of
Sections~\ref{sec:SL_method} and~\ref{sec:ML_method}, and compares the multilevel method
with the BL2 Markovian approximation in ~\cite{bayer2023weak}.}

\paragraph{Benchmark parameter sets.}

We consider European call options throughout the numerical experiments. We use
the EuRos and SPY rough Heston parameter sets reported
in \cite{boyarchenko2025fast}{, both originating from calibrations to market
implied volatility data,} as representative examples of two distinct
parameter regimes. The used model parameters are given in
Table~\ref{tab:numerical_benchmarks}.

\begin{table}[H]
    \centering
    \small
    \caption{
    Rough Heston parameter sets used in the numerical experiments.
    }
    \label{tab:numerical_benchmarks}
    \begin{tabular}{ccccccc}
        \toprule
        Case
        & $\alpha$
        & $\gamma$
        & $\nu$
        & $\rho$
        & $V_0$
        & $\theta$
        \\
        \midrule
        EuRos
        & $0.62$
        & $0.1$
        & $0.331$
        & $-0.681$
        & $0.0392$
        & $0.3156$
        \\
        SPY
        & $0.7151$
        & $1.8967$
        & $0.6144356$
        & $-0.6704$
        & $0.06246$
        & $0.03848$
        \\
        \bottomrule
    \end{tabular}
\end{table}

\blue{Unless stated otherwise, the time-step hierarchy is initialized with
$\Delta t_0=T/32$. The refinement results in
Section~\ref{subsec:time_discretization_numerics} provide numerical evidence,
for the tested configurations, that this choice yields stable fractional
Adams approximations and that the first levels already exhibit the
asymptotic behavior used by the Richardson indicator.}

\paragraph{{Selection of the damping parameter.}}
\label{par:practical_damping_selection}
{The damping rule proposed in~\cite{bayer2023optimal} selects $R$ by solving}
\[
{
R^\star
\in
\operatorname*{arg\,min}_{R\in\delta_V}
\sup_{u\in\mathbb{R}}
|g(u;R,\Theta_X,\Theta_P)|
=
\operatorname*{arg\,min}_{R\in\delta_V}
\blue{\left|g(0;R,\Theta_X,\Theta_P)\right|}.
}
\]
{\color{blue}
For the rough Heston model, this optimization problem cannot be applied
directly because the model admissibility set $\delta_X$ in
Assumption~\ref{ass:model}, and hence
$\delta_V=\delta_P\cap\delta_X$, is not characterized analytically. Define
\red{Moment explosions and critical moments in the rough Heston model are
studied in~\cite{gerhold2019moment,kellerressel2020comparison}.}
the right critical log-price moment by
\[
u_+(T)
:=
\sup\left\{
u>1:\mathbb E\!\left[e^{uX_T}\right]<\infty
\right\}.
\]
For the European call payoff~\eqref{eq:call_payoff}, the explicit transform
in~\eqref{eq:call_payoff_transform} gives
$\delta_P=(-\infty,-1)$, and full Fourier admissibility requires
\[
R\in(-u_+(T),-1).
\]
The payoff boundary $R=-1$, which corresponds to the pole at
$\xi=-\mathrm{i}$, can be enforced exactly. We therefore take the payoff
margin $\eta_P=0.5$, use $R_0=-1-\eta_P=-1.5$, and minimize the numerical
approximation of $|g(0;R,\Theta_X,\Theta_P)|$ using L-BFGS-B subject to the
one-sided constraint
\[
R\leq-1-\eta_P=-1.5.
\]
This constraint guarantees payoff-transform admissibility. However, because
no certified value or lower bound for $u_+(T)$ is available for the parameter
regimes considered, the model-moment boundary $R>-u_+(T)$ cannot be enforced.
The resulting damping selection is therefore heuristic with respect to model
admissibility. Although the numerical Riccati solution is finite at the
selected $R$, this does not certify model-moment admissibility. The selected
damping parameters for every distinct numerical configuration are reported in
Table~\ref{tab:numerical_damping_values} in
Appendix~\ref{app:numerical_damping_values}.

The scaling parameter of the
Gauss--Laguerre rule is subsequently chosen according
to~\eqref{eq:sigma_choice}.
}

{\color{blue}
\subsection{Numerical evidence for the time-discretization assumptions}
\label{subsec:time_discretization_numerics}

We examine the time-discretization assumptions for the two configurations used
in Section~\ref{subsubsec:single_multilevel_comparison}: EuRos with $T=2$,
$S_0=K=1000$, and $r=0$, and SPY with $T=2/365$, $S_0=K=1$, and $r=0$.
The damping parameter and the default time-step hierarchy are selected as
described above. All reported reference quantities use
$\Delta t_{\mathrm{ref}}=T/4096$. Unless stated otherwise, the reference
Fourier integrals are evaluated using $\bar N=128$ Gauss--Laguerre nodes.

We first assess Assumption~\ref{ass:riccati_pointwise_error}. Since the exact
Riccati solution is unavailable, let $\bar h$ denote the Adams approximation
computed on the reference grid. For numerical validation, we use the reference
error
\begin{equation}
\bar E_\ell^{\,h,\mathrm{w}}(u)
:=
\max_{1\leq j\leq M_\ell}
t_{\ell,j}^{1-\alpha}
\left|
h_{\ell,j}(u+\mathrm{i}R)
-\bar h(u+\mathrm{i}R,t_{\ell,j})
\right|.
\label{eq:numerical_weighted_nodal_error}
\end{equation}
Figure~\ref{fig:nodal_adams_rate} provides numerical evidence for
Assumption~\ref{ass:riccati_pointwise_error} while estimating the numerical
rate $q$ from the decay of $\bar E_\ell^{\,h,\mathrm{w}}(u)$ under refinement.

\begin{figure}[!htbp]
    \centering
    \begin{subfigure}[t]{0.49\textwidth}
        \centering
        \includegraphics[width=\linewidth]
        {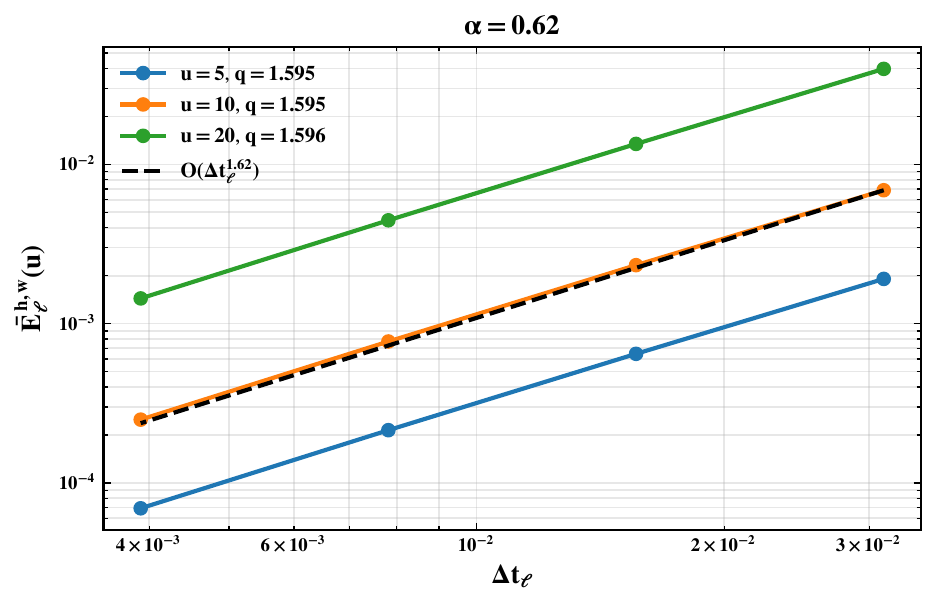}
        \caption{EuRos: $u\in\{5,10,20\}$.}
    \end{subfigure}
    \hfill
    \begin{subfigure}[t]{0.49\textwidth}
        \centering
        \includegraphics[width=\linewidth]
        {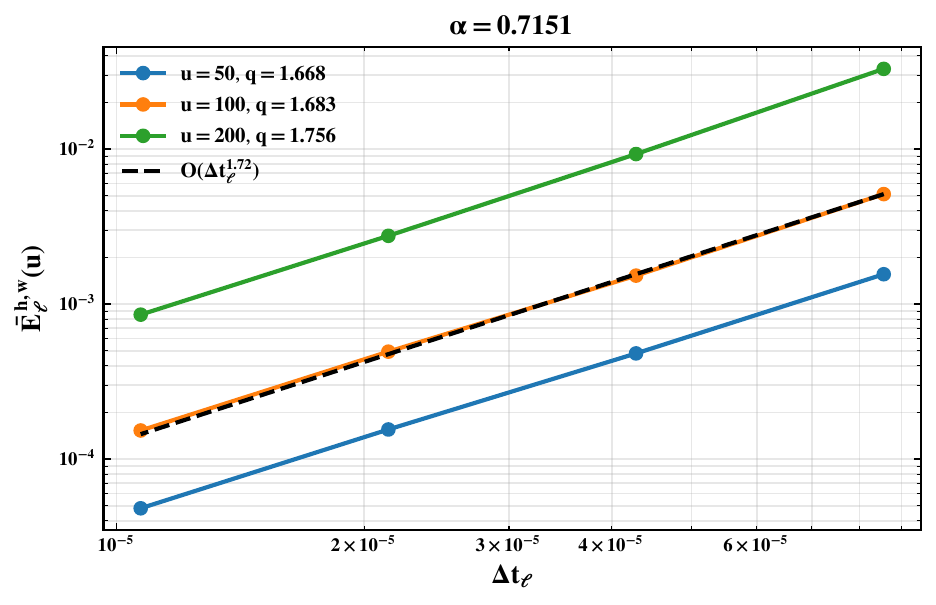}
        \caption{SPY: $u\in\{50,100,200\}$.}
    \end{subfigure}
    \caption{
    $\bar E_\ell^{\,h,\mathrm{w}}(u)$ for EuRos and SPY. The dashed lines
    indicate $\mathcal O(\Delta t_\ell^{1+\alpha})$.
    }
    \label{fig:nodal_adams_rate}
\end{figure}

The fitted rates are approximately $1.595$--$1.596$ for EuRos and
$1.668$--$1.756$ for SPY, broadly consistent with $1+\alpha=1.62$ and
$1+\alpha=1.7151$, respectively. As discussed in
Remark~\ref{rem:observed_rate}, these results provide numerical evidence for
the empirical rate $1+\alpha$ in the tested configurations, but do not provide
a theoretical proof or identify the generic exponent $q$.
Accordingly, we use $1+\alpha$ as the numerical rate in the estimators below.

We next assess the time-integration rate $r_T$ in
Assumption~\ref{ass:time_integration_error}. For
$\xi=u+\mathrm{i}R$, define the time-integration error
\[
E_\ell^{\,\mathrm{time}}(\xi)
:=
\left|
\int_0^T \mathcal J_\xi(t)\,dt
-
\mathcal T_\ell
\left[
\left(\mathcal J_\xi(t_{\ell,j})\right)_{j=0}^{M_\ell}
\right]
\right|.
\]
The exact function $\mathcal J_\xi$ is unavailable because it depends on the
unknown Riccati solution $h$. Using $\bar h$, define
\[
\bar{\mathcal J}_\xi(t)
:=
J\!\left(\xi,\bar h(\xi,t)\right).
\]
For numerical validation, we approximate $E_\ell^{\,\mathrm{time}}(\xi)$ by
\[
\bar E_\ell^{\,\mathrm{time}}(\xi)
:=
\left|
\mathcal T_{\mathrm{ref}}[\bar{\mathcal J}_\xi]
-\mathcal T_\ell[\bar{\mathcal J}_\xi]
\right|,
\]
where $\mathcal T_{\mathrm{ref}}$ and $\mathcal T_\ell$ denote the same
composite trapezoidal rule on the reference grid and on the grid with time
step $\Delta t_\ell$, respectively. The decay of this estimate under
refinement is shown in
Figure~\ref{fig:time_integration_rate}.

\begin{figure}[!htbp]
    \centering
    \begin{subfigure}[t]{0.49\textwidth}
        \centering
        \includegraphics[width=\linewidth]
        {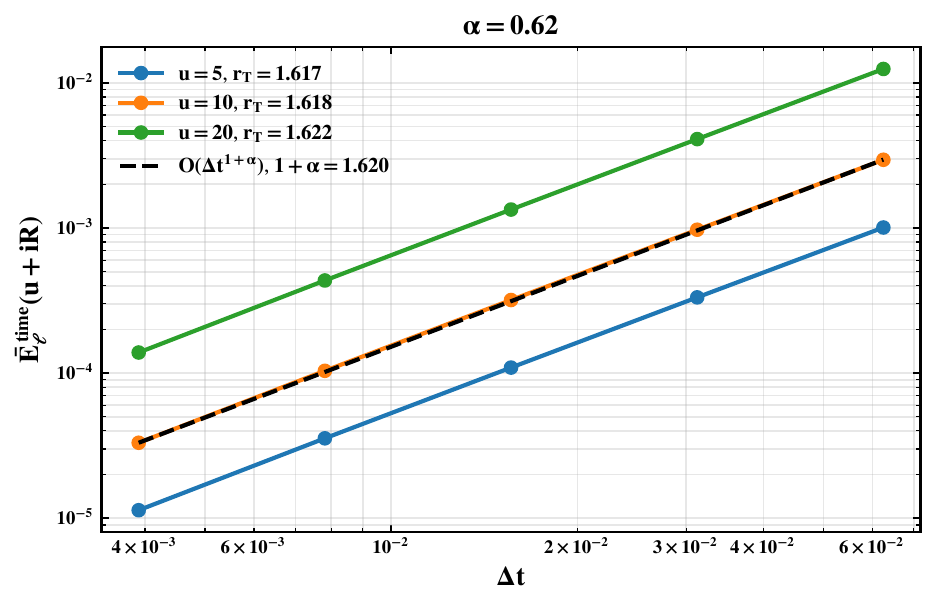}
        \caption{EuRos configuration.}
    \end{subfigure}
    \hfill
    \begin{subfigure}[t]{0.49\textwidth}
        \centering
        \includegraphics[width=\linewidth]
        {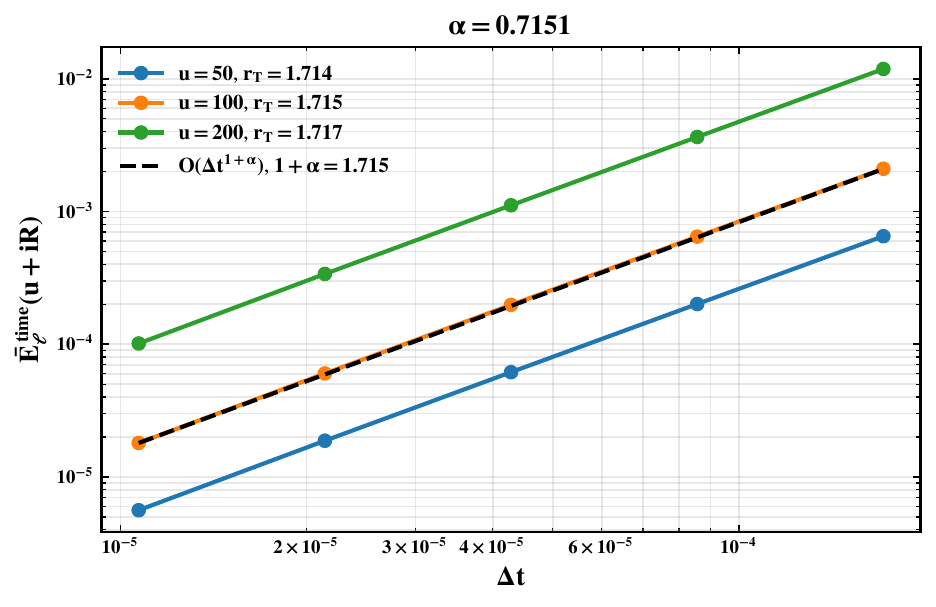}
        \caption{SPY configuration.}
    \end{subfigure}
    \caption{
    $\bar E_\ell^{\,\mathrm{time}}(u+\mathrm{i}R)$ for EuRos and SPY.
    The dashed lines indicate
    $\mathcal O(\Delta t_\ell^{1+\alpha})$.
    }
    \label{fig:time_integration_rate}
\end{figure}

The fitted rates are approximately $1.617$--$1.622$ for EuRos and
$1.714$--$1.717$ for SPY, which are close to $1+\alpha=1.62$ and
$1+\alpha=1.7151$, respectively. These results provide numerical evidence
for using $r_T=1+\alpha$ in the tested configurations, but a theoretical proof
of this rate is not provided in this work. Combining these observations with
the measured nodal rates and the definition
$p_G=\min\{q,r_T\}$ in~\eqref{eq:exponent_error_and_rate}, we obtain
empirically $p_G\approx1+\alpha$ and use $p_G=1+\alpha$ in the numerical
experiments below.

We next examine whether the Fourier-integrand error exhibits the same rate.
Define the exact Fourier-integrand error by
\[
E_\ell^{\,g}
:=
\left\|g-g_\ell\right\|_{L^1(0,\infty)}.
\]
Provided that Conjecture~\ref{conj:contour_domination} holds,
Proposition~\ref{prop:discretization_error} transfers the exponent rate $p_G$
to the $L^1$ error of the Fourier integrand. Thus, it predicts
$E_\ell^{\,g}=\mathcal O(\Delta t_\ell^{p_G})$ and hence $p=p_G$ in
Assumption~\ref{ass:fourier_integrand_rate}.

For numerical validation, let $\bar g$ denote the reference integrand computed
from the fully discrete exponent
in~\eqref{eq:discrete_characteristic_exponent} on the reference grid, and set
\[
\bar E_\ell^{\,g}
:=
\left\|\bar g-g_\ell\right\|_{L^1(0,\infty)}.
\]
Its convergence is shown in
Figure~\ref{fig:reference_time_discretization_errors}.

\begin{figure}[!htbp]
    \centering
    \begin{subfigure}[t]{0.49\textwidth}
        \centering
        \includegraphics[width=\linewidth]
        {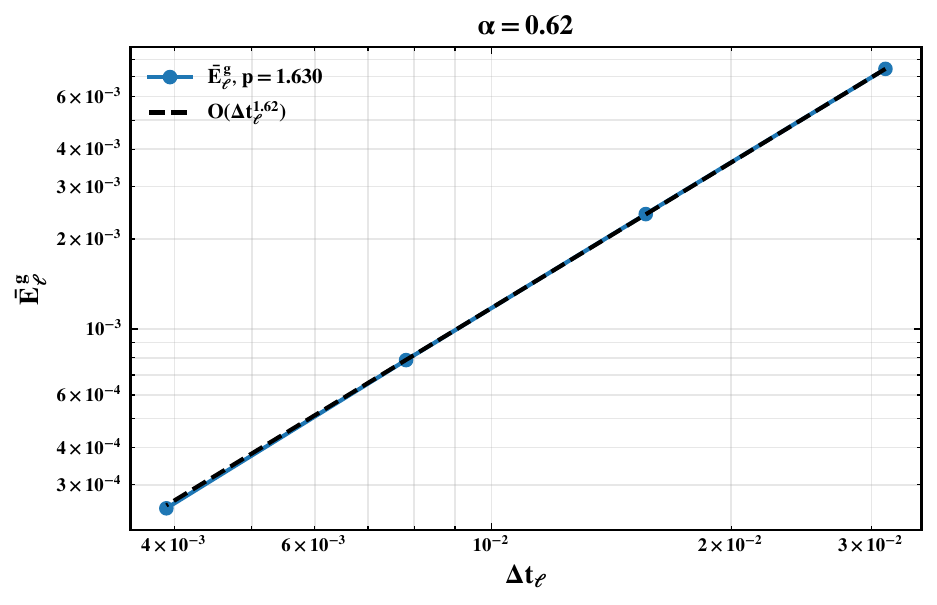}
        \caption{EuRos configuration.}
    \end{subfigure}
    \hfill
    \begin{subfigure}[t]{0.49\textwidth}
        \centering
        \includegraphics[width=\linewidth]
        {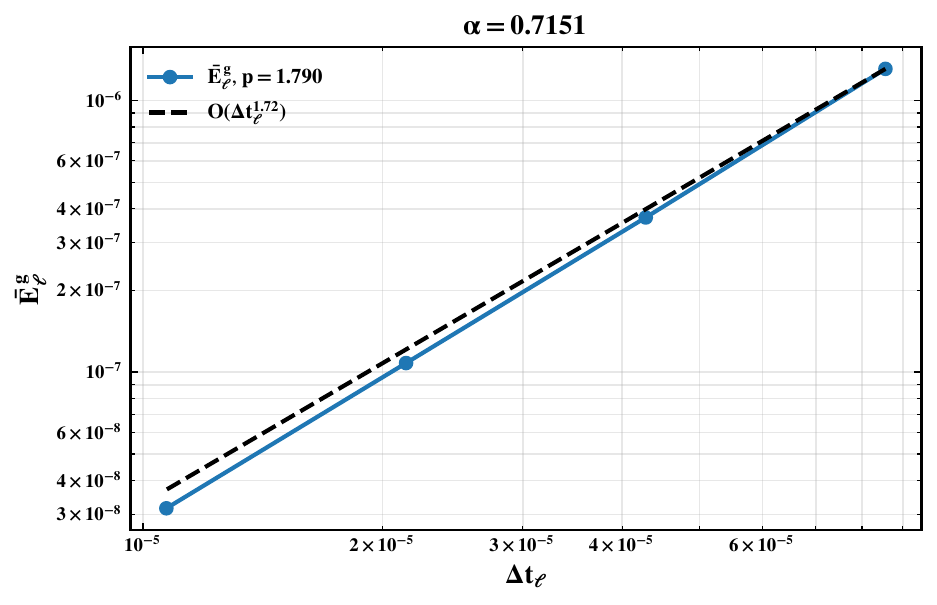}
        \caption{SPY configuration.}
    \end{subfigure}
    \caption{
    $\bar E_\ell^{\,g}$ for EuRos and SPY.
    The dashed lines indicate $\mathcal O(\Delta t_\ell^{1+\alpha})$.
    }
    \label{fig:reference_time_discretization_errors}
\end{figure}

The fitted rates of $\bar E_\ell^{\,g}$ in
Figure~\ref{fig:reference_time_discretization_errors} are
$1.630$ for EuRos and $1.790$ for SPY, supporting
Assumption~\ref{ass:fourier_integrand_rate} with $p\approx p_G\approx
1+\alpha$. We therefore use
\[
p=p_G=1+\alpha
\]
in the pricing experiments below.

The exact price error $E_\ell^{\,V}=|V-V_\ell|$ cannot be evaluated directly,
since it contains the unknown quantity of interest $V$. We therefore use the
Richardson estimator $\widehat E_\ell^{\,V}$ defined
in~\eqref{eq:richardson_price_notation}. Under
Assumption~\ref{ass:sharp_price_expansion},
Proposition~\ref{prop:richardson} shows that it is asymptotically equivalent to
$E_\ell^{\,V}$.

For numerical validation, define the reference value
\[
\bar V
:=
\mathcal Q_{\bar N}^{\widetilde\sigma}[\bar g],
\]
and approximate $E_\ell^{\,V}$ by
\[
\bar E_\ell^{\,V}
:=
|\bar V-V_\ell|.
\]
The same Gauss--Laguerre rule is used to evaluate $V_\ell$ in this comparison.

In practice, in the computable Richardson
estimator~\eqref{eq:computable_richardson_estimator}, we use
$\bar N^{\mathrm{comp}}=16$ quadrature points. To assess the resulting
quadrature error, let
$Q_{N,\ell}:=\mathcal Q_N^{\widetilde\sigma}[\Delta g_\ell]$. For
$\ell=1,\ldots,4$, we compare it with the reference estimator through
\[
\begin{aligned}
\frac{
\left|
\widehat E_{\bar N,\ell}^{\,V}
-
\widehat E_{\bar N^{\mathrm{comp}},\ell}^{\,V}
\right|
}{
\widehat E_{\bar N,\ell}^{\,V}
}
&=
\frac{
\left|
\left|Q_{\bar N,\ell}\right|
-
\left|Q_{\bar N^{\mathrm{comp}},\ell}\right|
\right|
}{
\left|Q_{\bar N,\ell}\right|
}
\\
&\leq
\frac{
\left|
Q_{\bar N,\ell}-Q_{\bar N^{\mathrm{comp}},\ell}
\right|
}{
\left|Q_{\bar N,\ell}\right|
}
=:
\rho_\ell.
\end{aligned}
\]
For EuRos, $\rho_\ell$ ranges from
$2.5215\times10^{-9}$ to $5.5594\times10^{-9}$, corresponding to only
$2.5215\times10^{-7}\%$--$5.5594\times10^{-7}\%$. For SPY, it ranges from
$1.0290\times10^{-3}$ to $2.1979\times10^{-2}$, corresponding to
$0.1029\%$--$2.1979\%$. These small ratios show that the quadrature error is
negligible relative to the level difference and are consistent with
condition~\eqref{eq:quadrature_condition_order_p} for both configurations.
The practical indicator $\eta_{\mathrm{disc}}(\ell)$
in~\eqref{eq:practical_disc_indicator} is computed using the same pilot
resolution $\bar N^{\mathrm{comp}}=16$. Both quantities are therefore
computable without evaluating $V$. For numerical validation, we examine
$(V_\ell-\bar V)/\Delta t_\ell^p$ and
$\widehat E_{\bar N^{\mathrm{comp}},\ell}^{\,V}/\bar E_\ell^{\,V}$. Under
Assumption~\ref{ass:sharp_price_expansion}, the normalized error
$(V_\ell-V)/\Delta t_\ell^p$ converges to $c_V\neq0$. Thus, stabilization of
its reference approximation $(V_\ell-\bar V)/\Delta t_\ell^p$ around a
nonzero value supports the assumed expansion, while a ratio close to one
shows that the computable Richardson estimate accurately approximates the
discretization error.

\begin{figure}[!htbp]
    \centering
    \begin{subfigure}[t]{0.49\textwidth}
        \centering
        \includegraphics[width=\linewidth]
        {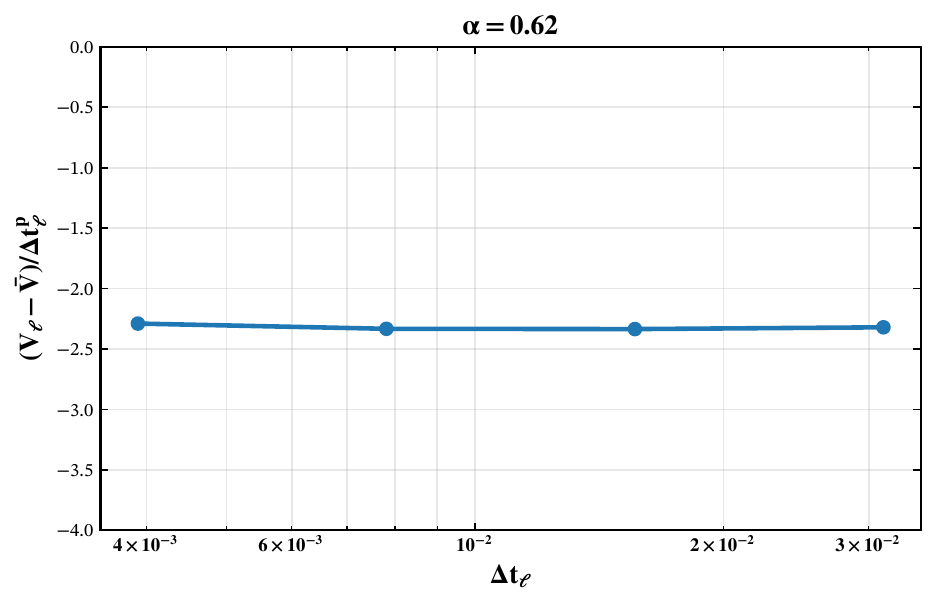}
        \caption{EuRos: signed normalized error.}
    \end{subfigure}
    \hfill
    \begin{subfigure}[t]{0.49\textwidth}
        \centering
        \includegraphics[width=\linewidth]
        {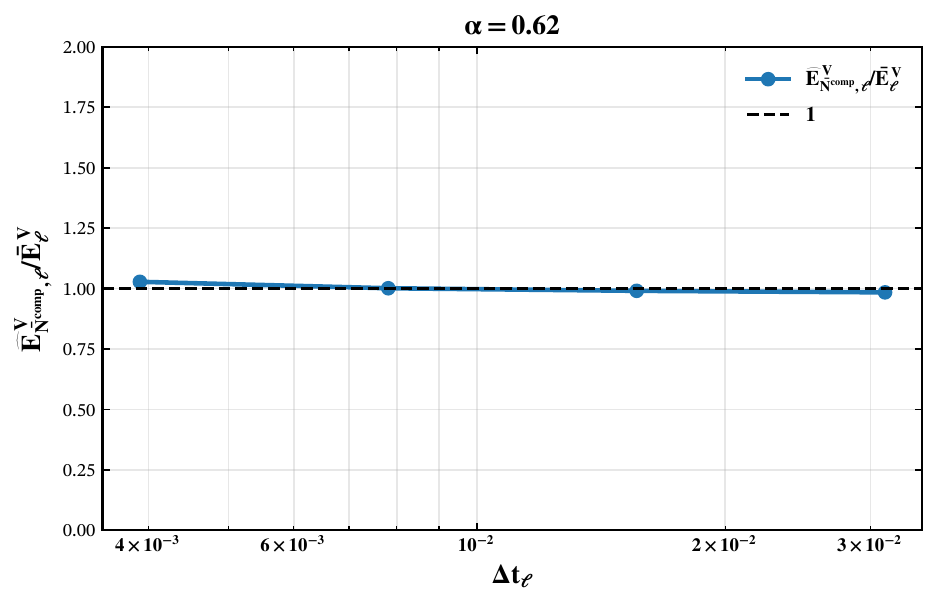}
        \caption{EuRos: Richardson ratio.}
    \end{subfigure}

    \medskip

    \begin{subfigure}[t]{0.49\textwidth}
        \centering
        \includegraphics[width=\linewidth]
        {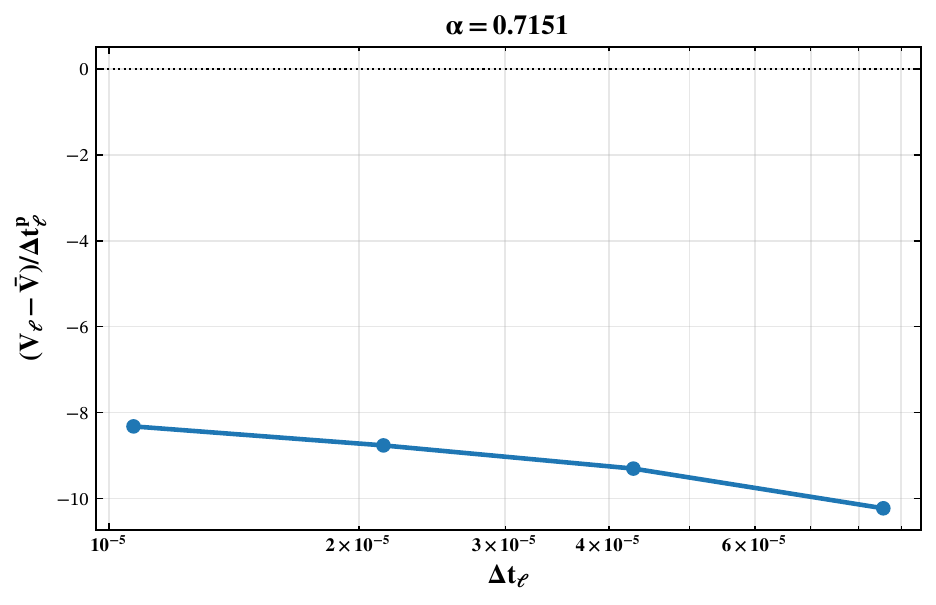}
        \caption{SPY: signed normalized error.}
    \end{subfigure}
    \hfill
    \begin{subfigure}[t]{0.49\textwidth}
        \centering
        \includegraphics[width=\linewidth]
        {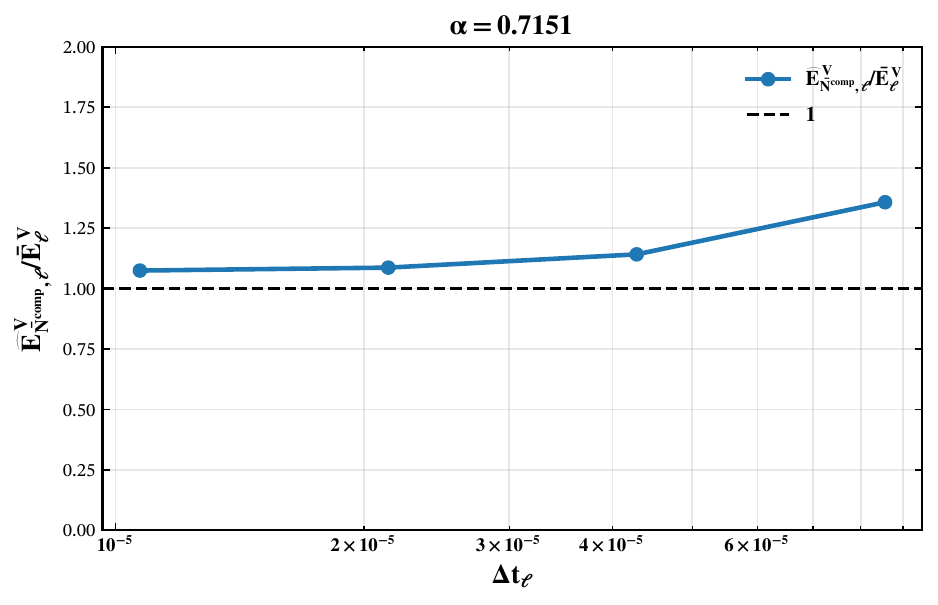}
        \caption{SPY: Richardson ratio.}
    \end{subfigure}
    \caption{
    Numerical evidence for Assumption~\ref{ass:sharp_price_expansion} and
    Proposition~\ref{prop:richardson}, using $p=1+\alpha$ for each parameter
    set. The reference error $\bar E_\ell^{\,V}$ uses $\bar N=128$
    Gauss--Laguerre nodes, while the Richardson estimate
    $\widehat E_{\bar N^{\mathrm{comp}},\ell}^{\,V}$ uses
    $\bar N^{\mathrm{comp}}=16$ nodes.
    }
    \label{fig:sharp_price_expansion_validation}
\end{figure}

For both parameter sets, the normalized errors are approximately stable and
remain bounded away from zero over the tested levels, supporting the nonzero
leading coefficient $c_V$ in Assumption~\ref{ass:sharp_price_expansion}.
Moreover, the Richardson error
estimate $\widehat E_{\bar N^{\mathrm{comp}},\ell}^{\,V}$ remains close to
$\bar E_\ell^{\,V}$ for both parameter sets.

We also compare $\bar E_\ell^{\,V}$ with
$\widehat E_{\bar N^{\mathrm{comp}},\ell}^{\,V}$ and the practical indicator
$\eta_{\mathrm{disc}}(\ell)$.

\begin{figure}[!htbp]
    \centering
    \begin{subfigure}[t]{0.49\textwidth}
        \centering
        \includegraphics[width=\linewidth]
        {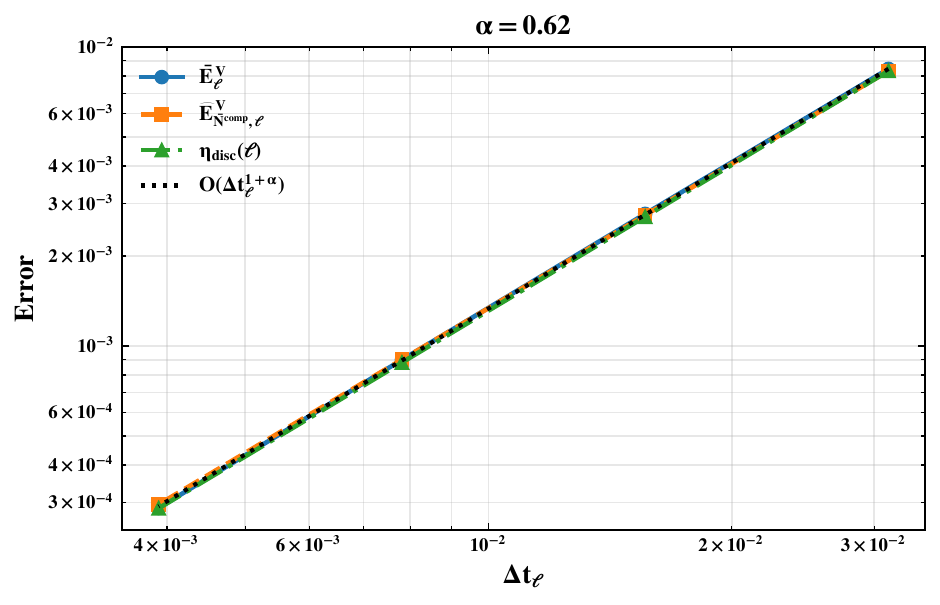}
        \caption{EuRos configuration.}
    \end{subfigure}
    \hfill
    \begin{subfigure}[t]{0.49\textwidth}
        \centering
        \includegraphics[width=\linewidth]
        {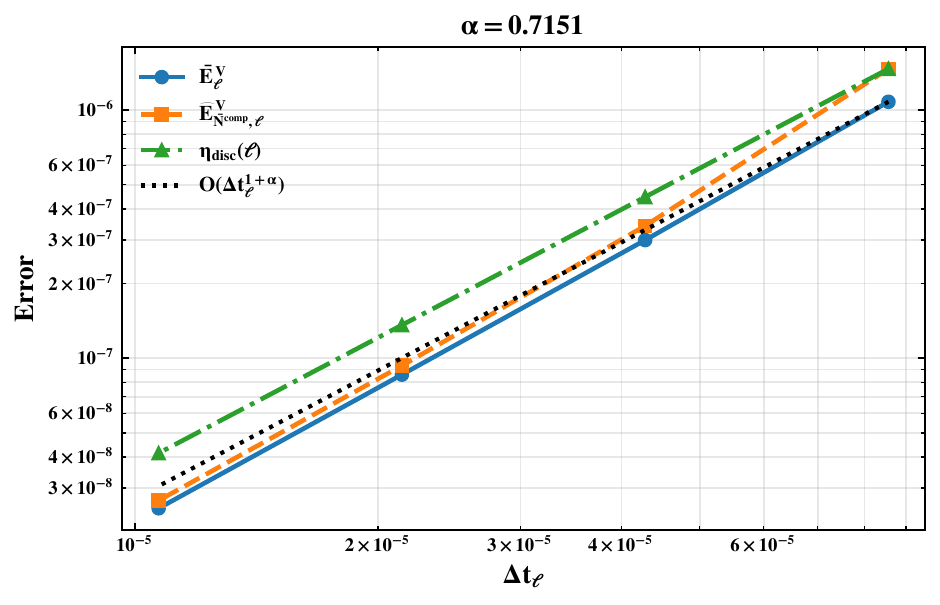}
        \caption{SPY configuration.}
    \end{subfigure}
    \caption{
    Reference error $\bar E_\ell^{\,V}$, Richardson estimator
    $\widehat E_{\bar N^{\mathrm{comp}},\ell}^{\,V}$, and practical
    indicator $\eta_{\mathrm{disc}}(\ell)$. The Richardson estimator and
    practical indicator use $\bar N^{\mathrm{comp}}=16$ Gauss--Laguerre
    nodes. The dotted lines
    indicate the respective rates $\mathcal O(\Delta t_\ell^{1+\alpha})$.
    }
    \label{fig:richardson_price_validation}
\end{figure}

Figure~\ref{fig:richardson_price_validation} shows that both computable
quantities reproduce the decay of $\bar E_\ell^{\,V}$, which
supports their use for practical level selection.

}

\subsection{\blue{Scaled Gauss--Laguerre Quadrature and Error Diagnostics}}
\label{subsec:quadrature_numerics}

{We begin by examining the choice of the scaling parameter
$\widetilde\sigma$ in~\eqref{eq:sigma_choice}; see
Section~\ref{subsubsec:scaling_parameter_selection}. We then turn to the
quadrature assumptions underlying the single-level and multilevel analyses of
Sections~\ref{sec:SL_method} and~\ref{sec:ML_method}; see
Section~\ref{subsubsec:quadrature_assumptions}. In particular, we consider the
algebraic quadrature error estimates and the fitted root-exponential error
estimates for the discretized integrands $g_\ell$ and the level differences
$\Delta g_\ell$.}

\subsubsection{Selection of the scaling parameter}
\label{subsubsec:scaling_parameter_selection}

{The decay of $g_\ell(u)$ and $\Delta g_\ell(u)$ as $u\to\infty$
depends on the rough Heston parameters, the maturity $T$, and the damping
parameter $R\in\delta_V$. To illustrate the dependence on maturity, we fix the
EuRos parameter set, optimize $R$ separately for each maturity, and consider
$T\in\{0.1,1,10\}$.
Figure~\ref{fig:GL_decay_maturity} shows
$|g_0(u)|/\|g_0\|_\infty$ and
$|\Delta g_1(u)|/\|\Delta g_1\|_\infty$, where
$\Delta g_1=g_1-g_0$. The decay is substantially slower for short maturities.
This dependence is relevant for the choice of $\widetilde\sigma$ because,
by~\eqref{eq:weighted_integrand}, the quadrature is applied to
$\widetilde g_\ell(u)=e^{\widetilde\sigma u}g_\ell(u)$.
\red{The scale controls the placement of the Laguerre nodes and the constants
in the quadrature remainder estimates. Matching it to the estimated Fourier
decay can reduce the quadrature error and improve the observed convergence
rate.} The same principle of adapting the
numerical transformation to the decay of the integrand has also been exploited
in the quasi Monte Carlo (QMC) and multilevel QMC settings; see
\cite{bayer2024quasi,hammouda2026single}.}

We next {compare the quadrature convergence obtained with the standard
Gauss--Laguerre rule, corresponding to $\widetilde\sigma=1$, with that obtained
using the scaling in~\eqref{eq:sigma_choice}.} For each maturity, the damping
parameter $R$ is optimized once and then held fixed across the two scaling
choices. The {discretization level} and all remaining model and payoff
parameters are also kept fixed. The quadrature errors are measured relative
to a rule with $\bar N=64$ {quadrature points}.
{Figures~\ref{fig:GL_scaling_comparison_level_zero}
and~\ref{fig:GL_scaling_comparison_level_difference} report the quadrature
errors for the level-zero integrand $g_0$ and the first level difference
$\Delta g_1$, respectively.}

\begin{figure}[!htbp]
    \centering
    \begin{subfigure}[t]{0.49\textwidth}
        \centering
        \includegraphics[width=\linewidth]
        {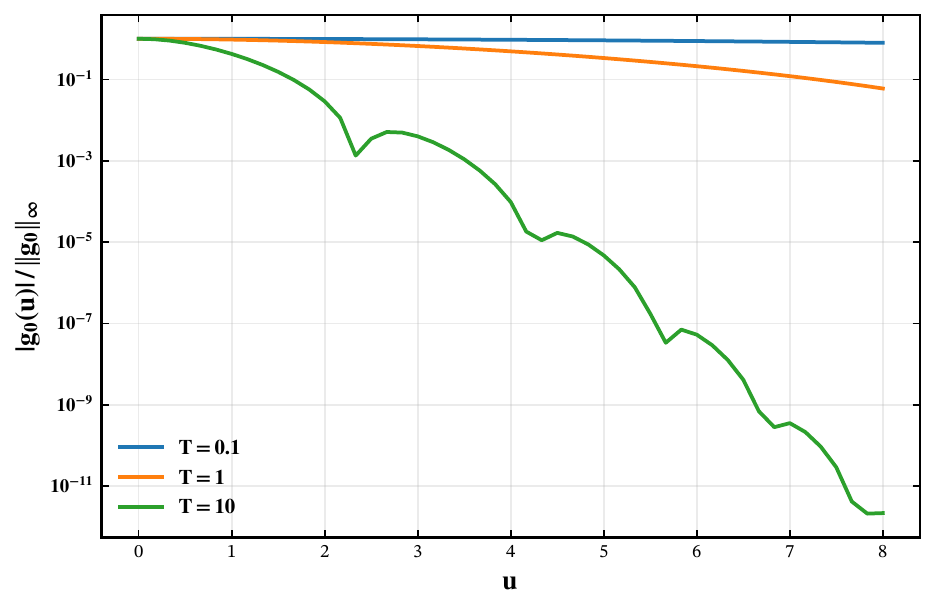}
        \caption{{Level-zero integrand,
        $|g_0(u)|/\|g_0\|_{\infty}$.}}
        \label{fig:GL_integrand_support}
    \end{subfigure}
    \hfill
    \begin{subfigure}[t]{0.49\textwidth}
        \centering
        \includegraphics[width=\linewidth]
        {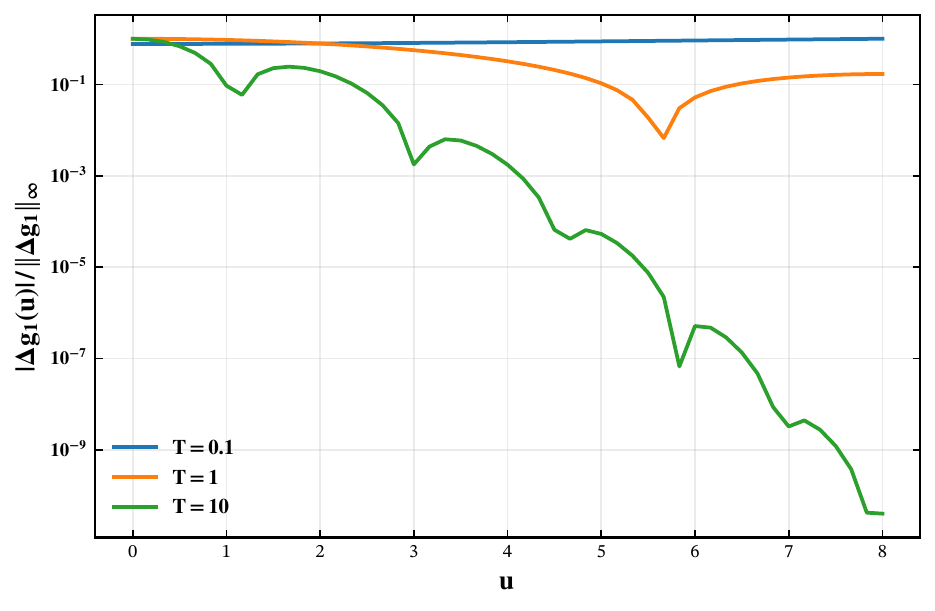}
        \caption{{First level difference,
        $|\Delta g_1(u)|/\|\Delta g_1\|_{\infty}$.}}
        \label{fig:GL_level_difference_support}
    \end{subfigure}
    \caption{{
    Relative magnitude of the level-zero integrand $g_0$ and the first level
    difference $\Delta g_1$ as functions of the integration variable $u$, for
    $T\in\{0.1,1,10\}$. The EuRos parameter set is used with the damping
    parameter $R\in\delta_V$ optimized separately for each maturity.
    }}
    \label{fig:GL_decay_maturity}
\end{figure}

\begin{figure}[!htbp]

\centering
\begin{subfigure}[t]{0.32\textwidth}
    \centering
    \includegraphics[width=\linewidth]
    {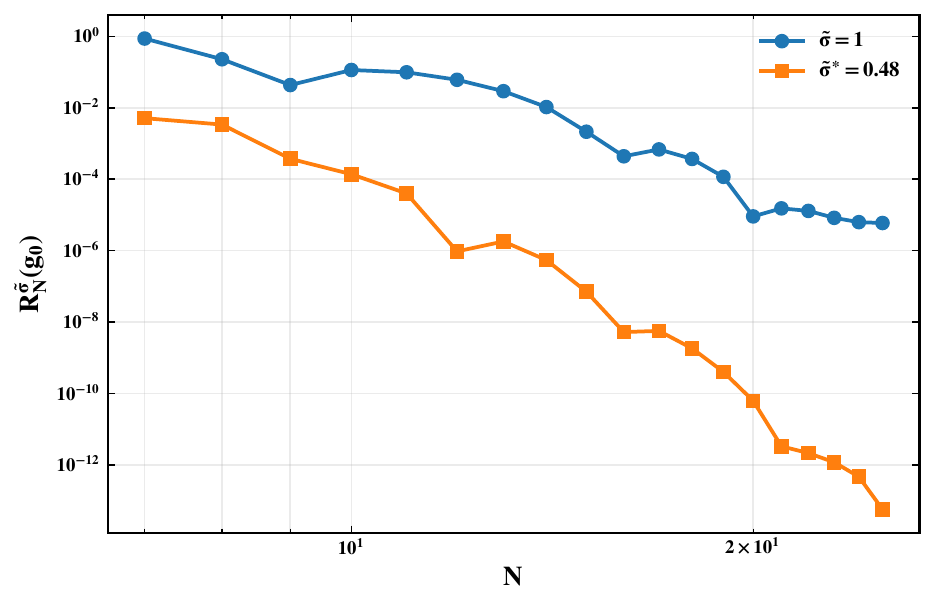}
    \caption{\(T=0.1\), \(\widetilde\sigma^\star=0.48\).}
\end{subfigure}
\hfill
\begin{subfigure}[t]{0.32\textwidth}
    \centering
    \includegraphics[width=\linewidth]
    {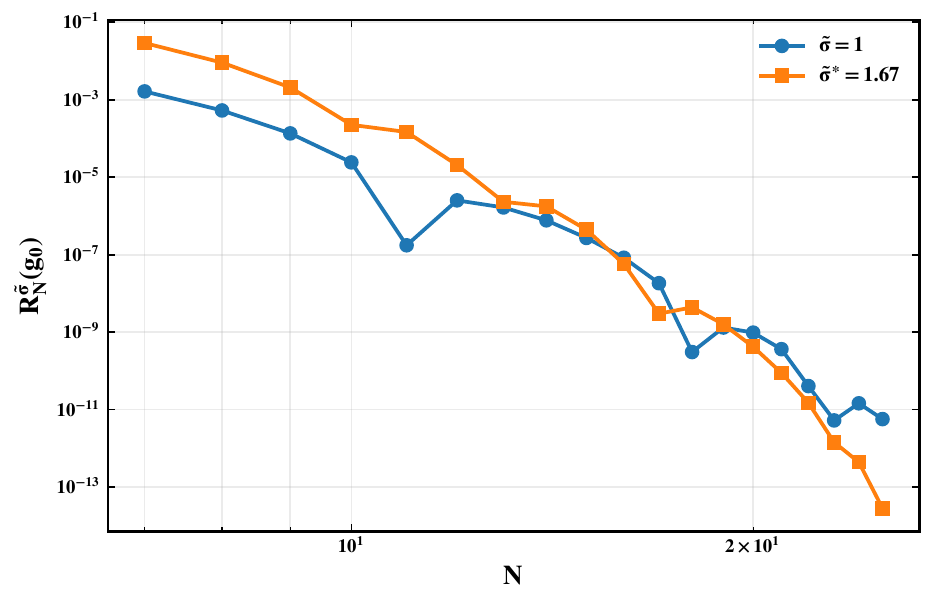}
    \caption{\(T=1\), \(\widetilde\sigma^\star=1.67\).}
\end{subfigure}
\hfill
\begin{subfigure}[t]{0.32\textwidth}
    \centering
    \includegraphics[width=\linewidth]
    {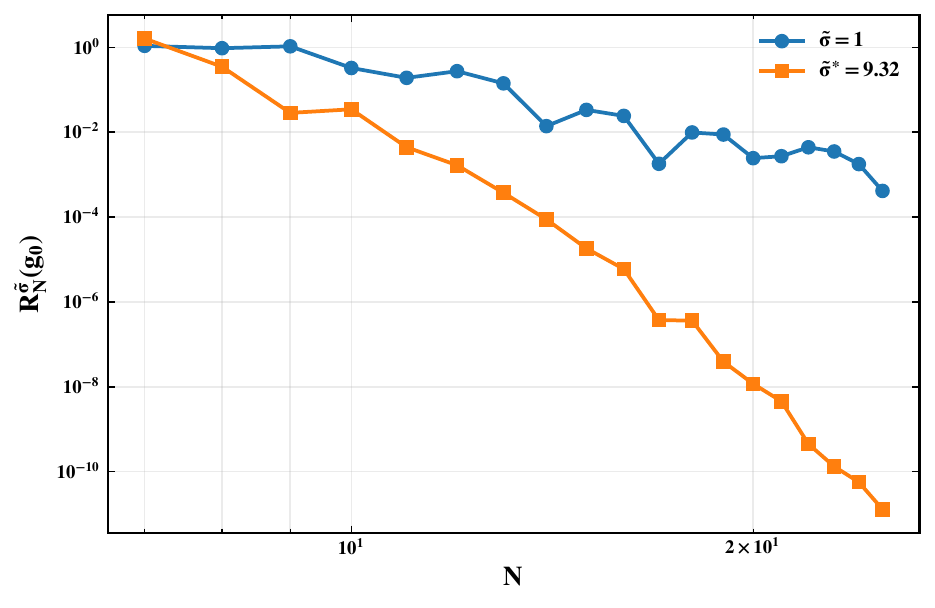}
    \caption{\(T=10\), \(\widetilde\sigma^\star=9.32\).}
\end{subfigure}
\caption{
Quadrature errors for the level-zero integrand $g_0$ at
$T\in\{0.1,1,10\}$, using the standard Gauss--Laguerre rule
$\widetilde\sigma=1$, marked in blue, and the scaling
in~\eqref{eq:sigma_choice}, marked in orange.
}
\label{fig:GL_scaling_comparison_level_zero}
\end{figure}

\begin{figure}[!htbp]
\centering
\begin{subfigure}[t]{0.32\textwidth}
    \centering
    \includegraphics[width=\linewidth]
    {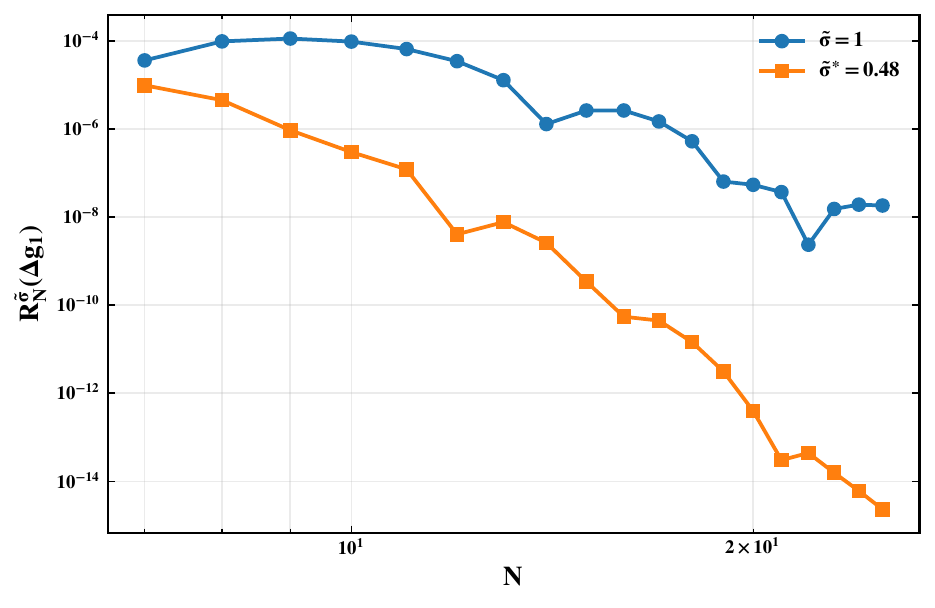}
    \caption{\(T=0.1\), \(\widetilde\sigma^\star=0.48\).}
\end{subfigure}
\hfill
\begin{subfigure}[t]{0.32\textwidth}
    \centering
    \includegraphics[width=\linewidth]
    {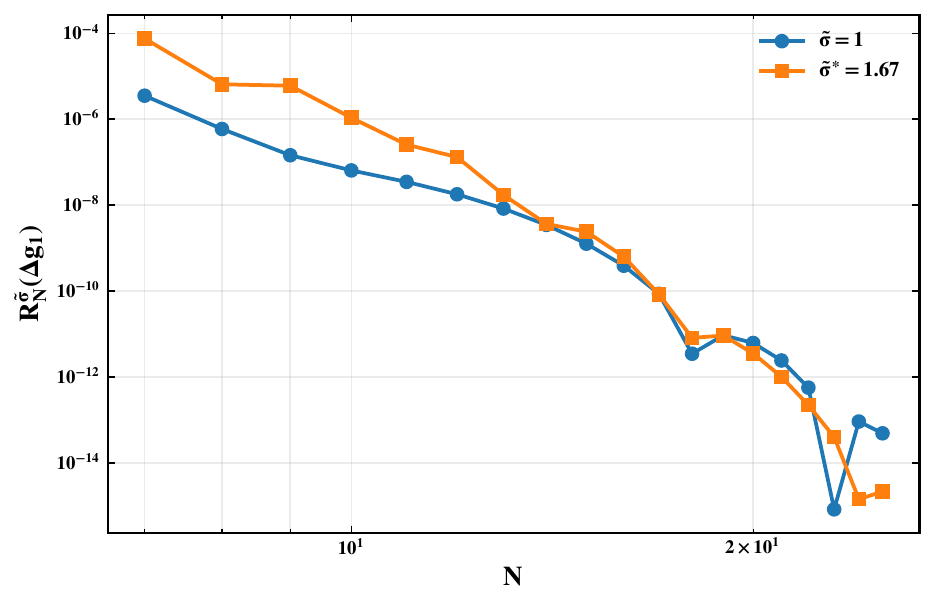}
    \caption{\(T=1\), \(\widetilde\sigma^\star=1.67\).}
\end{subfigure}
\hfill
\begin{subfigure}[t]{0.32\textwidth}
    \centering
    \includegraphics[width=\linewidth]
    {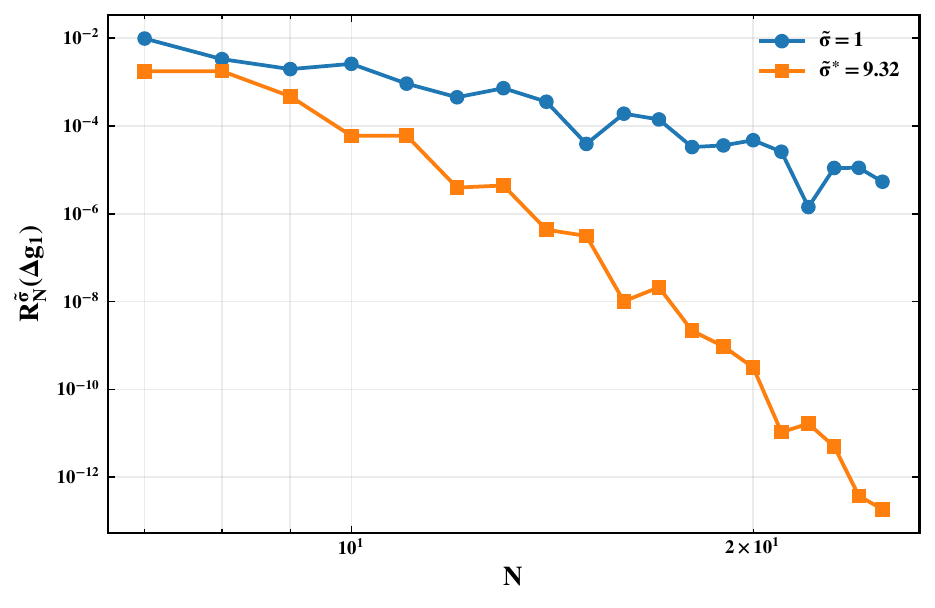}
    \caption{\(T=10\), \(\widetilde\sigma^\star=9.32\).}
\end{subfigure}
\caption{{
Quadrature errors for the first level difference $\Delta g_1$ at
$T\in\{0.1,1,10\}$, using the standard Gauss--Laguerre rule
$\widetilde\sigma=1$, marked in blue, and the scaling
in~\eqref{eq:sigma_choice}, marked in orange.
}}
\label{fig:GL_scaling_comparison_level_difference}
\end{figure}

{For both $g_0$ and $\Delta g_1$, the scaling
in~\eqref{eq:sigma_choice} yields faster convergence at short and long
maturities and performs comparably to the standard rule at $T=1$. We
therefore use this scaling in the remaining experiments.}

\subsubsection{{Quadrature error estimates}}
\label{subsubsec:quadrature_assumptions}

\red{Throughout this subsection, the fitted values of $s_\ell$, $A_\ell$,
and $B_\ell$ are numerical estimates obtained from finite-range quadrature
errors. They are used to assess consistency with the proposed quadrature
models and are not exact values of the theoretical regularity quantities
appearing in the assumptions.}

{The convergence analysis of the Gauss--Laguerre quadrature in
Section~\ref{sec:SL_method} includes the algebraic quadrature error estimate of
Theorem~\ref{thm:algebraic_error} and the root-exponential quadrature error
estimate of Theorem~\ref{thm:root_exponential}. We first compare the two
estimates for $g_0$ in Figure~\ref{fig:cover_model_comparison}. Over the range
of quadrature points considered, the root-exponential estimate provides a
closer description of the observed quadrature errors, whereas the algebraic
estimate is more conservative. \blue{Moreover, the algebraic cover is more
sensitive to the range of quadrature orders $N$ used in its construction.
We therefore show separate algebraic covers for the small- and large-$N$
regimes to assess this sensitivity and illustrate the effect of the selected
range.} The algebraic estimate, however, permits the explicit optimization and
complexity analysis developed in
{Propositions~\ref{prop:quadrature_order_algebraic},
\ref{prop:optimal_N_ml}, \ref{prop:work_scaling_single},
and~\ref{prop:work_scaling_multilevel}}, including the
closed-form expression for the minimum multilevel work
in~\eqref{eq:ml_work_optimal}. We therefore consider both error estimates in
the numerical experiments below.}

\begin{figure}[!htbp]
    \centering
    \includegraphics[width=0.72\linewidth]{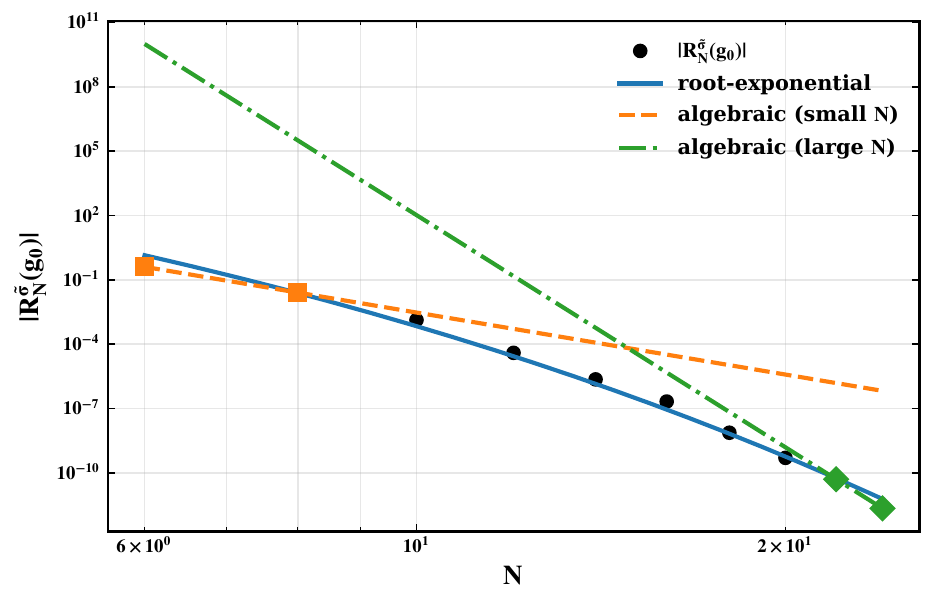}
    \caption{{
Quadrature error $|R_N^{\widetilde\sigma}(g_0)|$ as a function of the number
$N$ of quadrature points. \blue{The root-exponential cover is compared with
separate algebraic covers constructed in the small- and large-$N$ regimes.
The two algebraic covers illustrate their sensitivity to the range of
quadrature orders used in their construction.}
}}
    \label{fig:cover_model_comparison}
\end{figure}

\emph{{Transformed integrands.}}
We first examine the conditions used for the {transformed integrands}
\(\tilde g_\ell\). The algebraic
case is covered by Assumption~\ref{ass:weighted_smoothness} and the uniform
bounds in Proposition~\ref{prop:work_scaling_single}. The root-exponential
case is covered by Assumption~\ref{ass:root_exponential} and the uniform bounds
in Proposition~\ref{prop:work_scaling_single_root_exponential}. We use the
EuRos configuration
with \(T=2\) and the hierarchy
\red{\(\Delta t_\ell=T/(32 \times 2^\ell)\)}. At each level, the quadrature error is
measured relative to \(\bar N=64\) {quadrature points}.
{Fitting the parameters at every level is performed only to validate the
practical approximation described in the theory. The implemented method
estimates the level-zero parameters and uses them as surrogates for the
parameters at the selected finest level.}
\blue{Figure~\ref{fig:single_level_cover_fits} shows the fitted quadrature-error
models for the transformed integrands across the levels considered.}
The fitted {values of the root-exponential parameter} satisfy
\(B_\ell\approx 8.83\),
and a regression of the fitted {prefactors gives}
$$A_\ell^{\mathrm{re}}\approx C_{\mathrm{re}}\,\Delta t_\ell^{-0.019}.$$
The fitted {smoothness indices} satisfy \(s_\ell\approx 8.55\),
while a regression of the corresponding prefactors gives
$$A_\ell\approx C_{\mathrm{alg}}\,\Delta t_\ell^{\,0.001}.$$
{Thus, $B_\ell$, $s_\ell$, $A_\ell^{\mathrm{re}}$, and $A_\ell$
depend only weakly on the level over the range considered. These observations
are consistent with the bounds in
Propositions~\ref{prop:work_scaling_single}
and~\ref{prop:work_scaling_single_root_exponential}
being uniform across levels. {In particular, the estimated values of
$s_\ell$ are consistent with taking $s_{\mathrm{SL}}=8$ as the uniform lower
bound across these levels.}}

\begin{figure}[!htbp]
    \centering
    \begin{subfigure}[t]{0.49\textwidth}
        \centering
        \includegraphics[width=\linewidth]{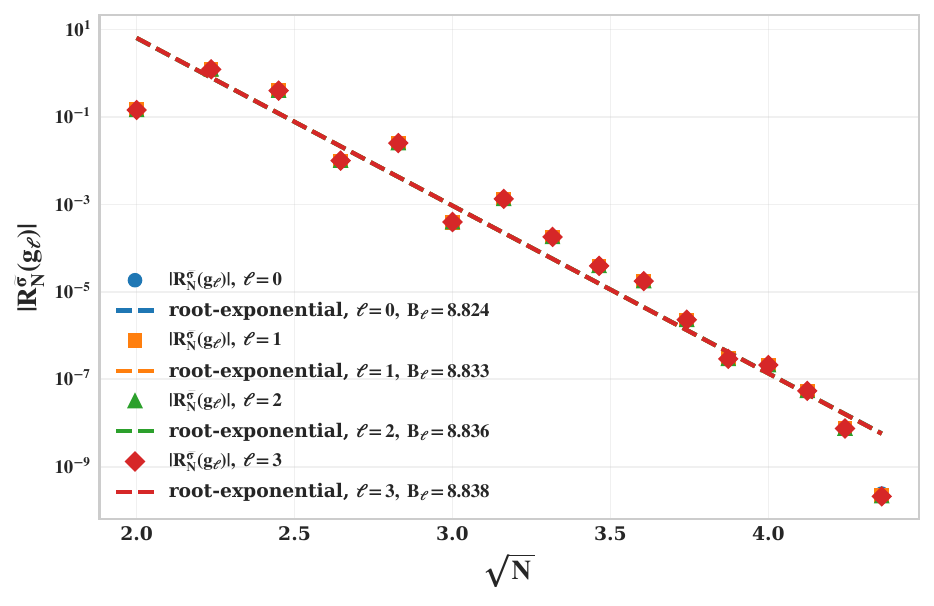}
        \caption{{Root-exponential fits of the quadrature error for
        \(g_\ell\).}}
    \end{subfigure}
    \hfill
    \begin{subfigure}[t]{0.49\textwidth}
        \centering
        \includegraphics[width=\linewidth]{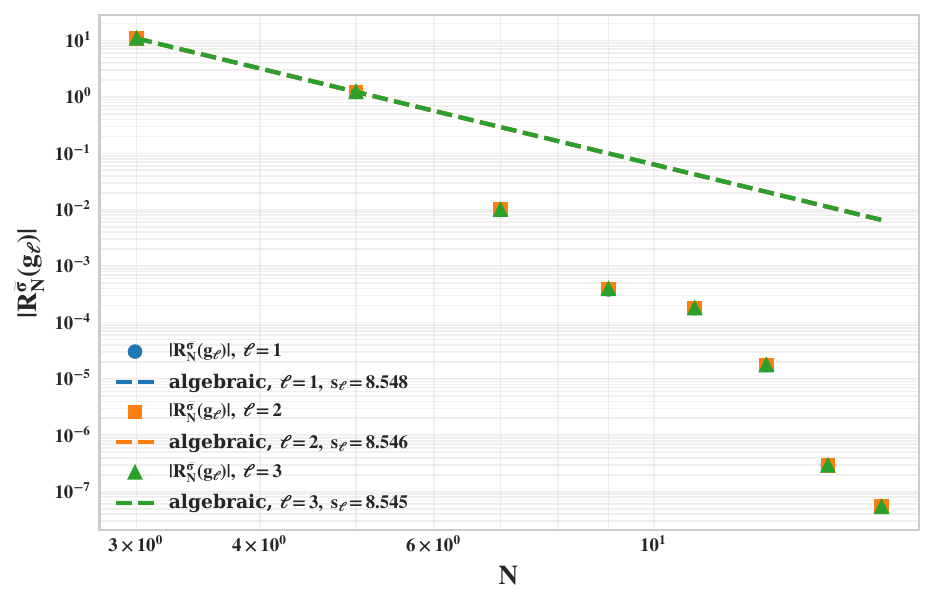}
        \caption{{Algebraic fits of the quadrature error for
        \(g_\ell\).}}
    \end{subfigure}
    \caption{{
    Quadrature errors for the discretized integrands \(g_\ell\), together with
    the root-exponential and algebraic fits. The fitted parameters and
    prefactors remain approximately constant across the levels considered.
    }}
    \label{fig:single_level_cover_fits}
\end{figure}

\emph{{Level differences.}}
We then examine the {level differences}
\(\Delta g_\ell=g_\ell-g_{\ell-1}\), which enter the multilevel
{method}. Assumptions~\ref{ass:level_difference_regularity}
and~\ref{ass:level_difference_decay} require {a common lower bound $s$
for the smoothness indices $s_\ell$} and the decay
\(A_\ell=\mathcal O(\Delta t_\ell^p)\).

\blue{Figures~\ref{fig:correction_exponential_cover} and
\ref{fig:correction_algebraic_cover} show the fitted correction models and
the scaling of their prefactors.}

For the root-exponential fits, the {fitted parameters} satisfy
\(B_\ell\approx8.26\), while a regression of the fitted {prefactors}
gives the decay rate
\(\Delta t_\ell^{\,1.613}\). For the algebraic fits, the {fitted
smoothness indices} satisfy \(s_\ell\approx7.60\), while a regression of the
fitted {prefactors} gives the decay rate
\(\Delta t_\ell^{\,1.675}\).

{Both $B_\ell$ and $s_\ell$ are nearly constant across the levels
considered, whereas the fitted prefactors decrease with an exponent close to
the observed discretization rate $p\approx1+\alpha=1.62$.}
{The observed algebraic decay
$A_\ell\approx C\Delta t_\ell^p$ supports
Assumption~\ref{ass:level_difference_decay}. The analogous root-exponential
decay $A_\ell^{\mathrm{re}}\approx C_{\mathrm{re}}\Delta t_\ell^p$ supports
the empirical level-scaling rule in
Remark~\ref{rem:practical_ml_quadrature_parameters}. Fitting all correction
levels here is only diagnostic; the implementation fits $A_1$ and
$A_1^{\mathrm{re}}$ and scales them across the finer levels as described
there.}
{The values of $s_\ell$ are consistent with taking $s=7$ as a common
integer lower bound across these levels. In this configuration, the fitted
smoothness indices of the level differences are slightly smaller than those
of the transformed integrands, but they remain sufficiently large for the
multilevel analysis. The reduction in multilevel work is therefore supported
primarily by the decay of the correction prefactors with the level and by the
level-dependent allocation of quadrature points.}

\begin{figure}[!htbp]
    \centering
    \begin{subfigure}[t]{0.49\textwidth}
        \centering
        \includegraphics[width=\linewidth]{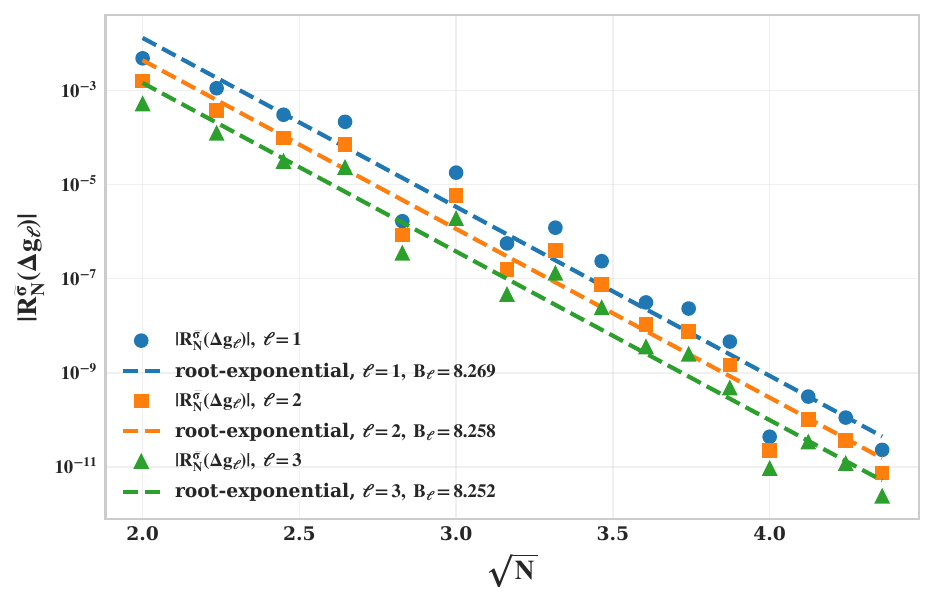}
        \caption{{Root-exponential fits of the quadrature error for
        \(\Delta g_\ell\).}}
    \end{subfigure}
    \hfill
    \begin{subfigure}[t]{0.49\textwidth}
        \centering
        \includegraphics[width=\linewidth]{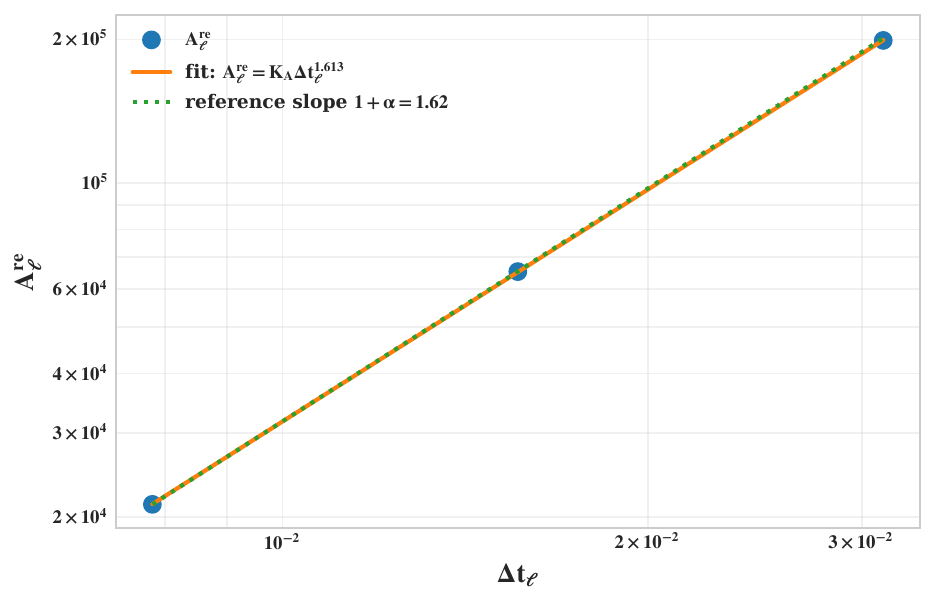}
        \caption{{Scaling of the fitted prefactors
        \(A_\ell^{\mathrm{re}}\).}}
    \end{subfigure}
    \caption{{
    Root-exponential fits of the quadrature error for the level differences
    \(\Delta g_\ell\). The fitted parameters \(B_\ell\) remain approximately
    constant across the levels considered, while the prefactors
    \(A_\ell^{\mathrm{re}}\) decrease approximately as
    \(\Delta t_\ell^{1.613}\).
    }}
    \label{fig:correction_exponential_cover}
\end{figure}

\begin{figure}[!htbp]
    \centering
    \begin{subfigure}[t]{0.49\textwidth}
        \centering
        \includegraphics[width=\linewidth]{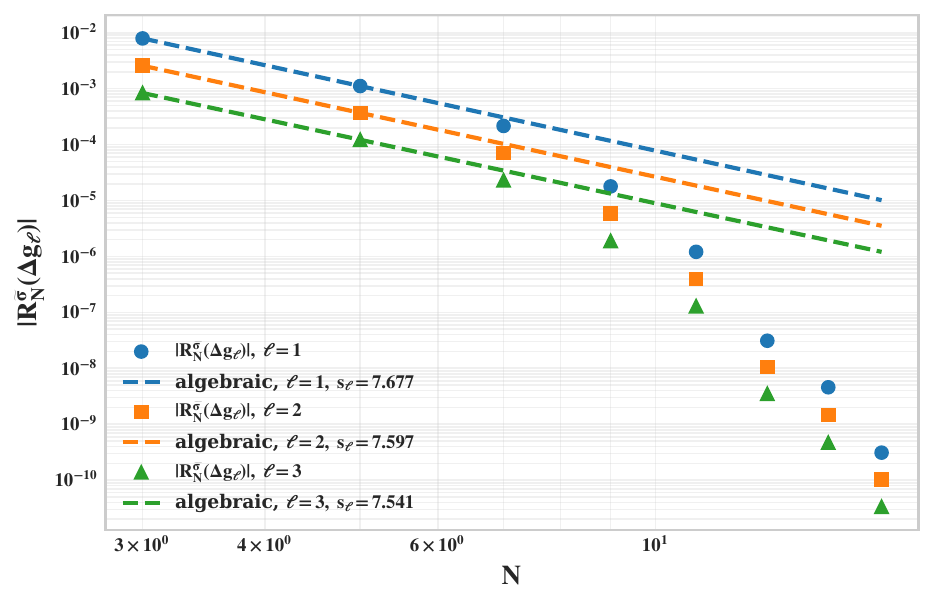}
        \caption{{Algebraic fits of the quadrature error for
        \(\Delta g_\ell\).}}
    \end{subfigure}
    \hfill
    \begin{subfigure}[t]{0.49\textwidth}
        \centering
        \includegraphics[width=\linewidth]{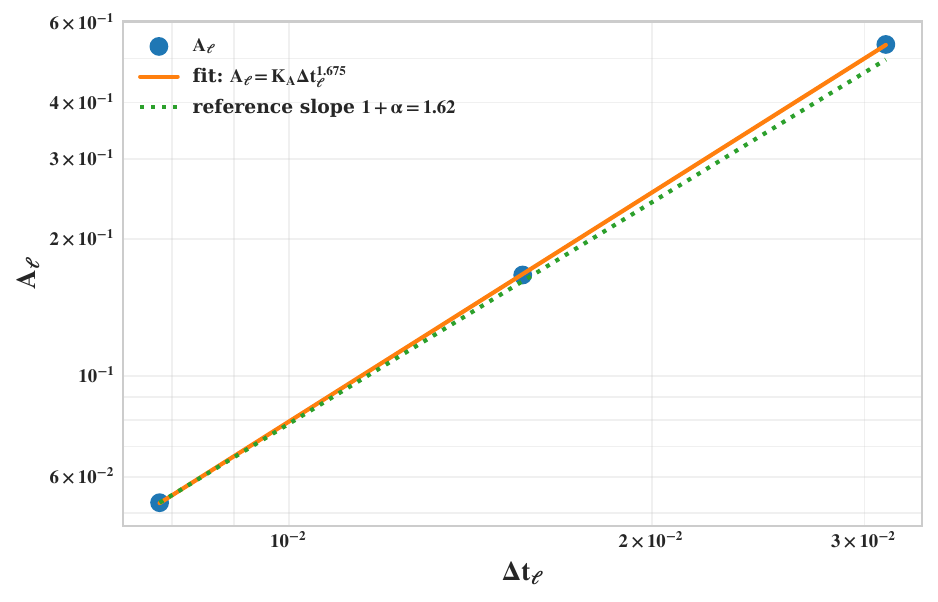}
        \caption{{Scaling of the fitted prefactors \(A_\ell\).}}
    \end{subfigure}
    \caption{{
    Algebraic fits of the quadrature error for the level differences
    \(\Delta g_\ell\). The fitted smoothness indices \(s_\ell\) show only weak
    dependence on the level, while \(A_\ell\) decreases approximately as
    \(\Delta t_\ell^{1.675}\), close to the observed rate
    \(\Delta t_\ell^p\) with \(p\approx1+\alpha=1.62\).
    }}
    \label{fig:correction_algebraic_cover}
\end{figure}

\subsection{{Computational Cost Comparison}}
\label{subsec:pricing_comparison}

{In this subsection, we first compare the SL and ML methods using the
algebraic and root-exponential quadrature error estimates in
Section~\ref{subsubsec:single_multilevel_comparison}. We then compare the ML
method with the BL2 Markovian approximation in
Section~\ref{subsubsec:bl2_comparison}. To distinguish the computational cost
of evaluating an option price from that of selecting the numerical parameters,
we report both the pricing CPU time and the total CPU time.

The pricing CPU time is the CPU time required to evaluate the option price
after all numerical parameters required by the method have been selected and
are reused. For SL, this is the CPU time required for the Gauss--Laguerre
quadrature at the selected finest discretization level and with the selected
number of quadrature points. For ML, it is the CPU time required for the
quadrature of $g_0$ and the level differences $\Delta g_\ell$, using the
selected discretization levels and numbers of quadrature points. The total CPU
time includes, in addition, the CPU time required for initialization, selection
of the damping parameter, computation and fitting of the quadrature error
estimates, selection of the discretization level, and selection of the numbers
of quadrature points.}

\subsubsection{Single-Level and Multilevel Methods Comparison}
\label{subsubsec:single_multilevel_comparison}
In this section, we compare the single-level and multilevel {methods} using
the algebraic and root-exponential quadrature {error estimates}. The
{time discretization} and {numbers of quadrature points} are selected as in
Sections~\ref{sec:SL_method} and~\ref{sec:ML_method}. The figures report the
scaled work, CPU time, and relative pricing error {with respect to a
high-accuracy benchmark price}. For the present comparison, we focus on
at-the-money options, \(S_0=K\).

\red{Let $C_0$ be the measured CPU time of one level-zero Riccati evaluation.
Using $\beta=2$, we estimate
$C_\ell=C_0(\Delta t_0/\Delta t_\ell)^\beta=C_0\,2^{\beta\ell}$.
The reported estimated times are
\[
\widehat W_{\mathrm{SL}}=C_LN_L,
\qquad
\widehat W_{\mathrm{ML}}
=C_0N_0+\sum_{\ell=1}^L(C_{\ell-1}+C_\ell)N_\ell.
\]}

{The relative pricing error is
\(\lvert V_{\mathrm{num}}-V_{\mathrm{ref}}\rvert/\lvert V_{\mathrm{ref}}\rvert\).
Thus, a relative target \(\varepsilon_{\mathrm{rel}}\) corresponds to the
absolute tolerance
\(\varepsilon=\varepsilon_{\mathrm{rel}}\lvert V_{\mathrm{ref}}\rvert\)
used in the theory. \red{This reference-based conversion is used only to
construct comparable relative-accuracy targets for the numerical benchmark;
it is not part of the pricing algorithm, whose input is an absolute tolerance
$\varepsilon>0$.} Since the benchmark price is fixed, this constant change
of scale does not alter the complexity exponents. The estimated work is
computed from the work proxies in~\eqref{eq:sl_work} and~\eqref{eq:ml_work}.}

\blue{The multilevel gain is not specific to this setting and is expected when
the integrand has low regularity, since the single-level method then requires
many quadrature points on the finest level. Greater regularity of the level
differences can further improve the gain, although it is not observed in the
present test. The gain depends on the model and numerical setup.}

\paragraph{Algebraic quadrature estimate.}

We first use the EuRos parameter set {in
Table~\ref{tab:numerical_benchmarks}}, with
\(T=2\), \(S_0=K=1000\), and \(r=0\).
\blue{The numerical results in
Section~\ref{subsec:time_discretization_numerics} give the observed convergence
rate \(p\approx1+\alpha=1.62\) for
\(\|g-g_\ell\|_{L^1(0,\infty)}\). The fitted smoothness indices \(s_\ell\)
of the transformed integrands \(\tilde g_\ell\) are approximately \(8.55\)
across the levels considered in
Figure~\ref{fig:single_level_cover_fits}, so we take
\(s_{\mathrm{SL}}=8\) as the uniform lower bound for the single-level
method in Proposition~\ref{prop:work_scaling_single}.}
\blue{For the multilevel method, the level-zero integrand $g_0$ is treated
separately and has $s_0=8$. The fitted smoothness index of the correction
integrand is approximately $7.68$, as shown
in Figure~\ref{fig:correction_algebraic_cover}. Consequently, the common
correction index in~\eqref{eq:common_order} is $s=7$.}
{The fitted prefactors
\(A_\ell\) for \(\ell\geq1\) decrease approximately as
\(\Delta t_\ell^p\), consistently with
Assumption~\ref{ass:level_difference_decay}. For the direct fractional Adams
implementation described in Section~\ref{sec:riccati_discretization}, the cost
exponent is \(\beta=2\).}

\blue{Here $2p/\beta=1.62$. Hence $s=7>2p/\beta$ and
$s_0=8\geq2p/\beta$, so both conditions required to use the simplified
multilevel complexity estimate~\eqref{eq:work_scaling_multilevel} are
satisfied.}

\blue{Substituting \(p=1.62\), \(\beta=2\), and
\(s_{\mathrm{SL}}=8\) into the
single-level complexity estimate~\eqref{eq:work_scaling_single}, and using the
multilevel complexity estimate~\eqref{eq:work_scaling_multilevel}, gives}
$${\color{blue}\frac{\beta}{p}+\frac{2}{s_{\mathrm{SL}}}=1.485\qquad\text{and}\qquad\frac{\beta}{p}=1.235.}$$
for the single-level and multilevel methods, respectively.

Figure~\ref{fig:single_multilevel_algebraic} {shows that the computational
work required by the ML method increases more slowly than that of the SL method
as the relative tolerance is reduced}. Both {methods} satisfy
the prescribed relative tolerance, and {the dependence of the
computational work on the tolerance is consistent with the complexity
estimates~\eqref{eq:work_scaling_single}
and~\eqref{eq:work_scaling_multilevel}.}
\blue{More precisely, the fitted exponents displayed in the figure are
$1.458$ for SL and $1.237$ for ML, compared with the predicted exponents
$1.485$ and $1.235$, respectively, over the tested tolerance range.}

\begin{figure}[!htbp]
    \centering
    \begin{subfigure}[t]{0.49\textwidth}
        \centering
        \includegraphics[width=\linewidth]
        {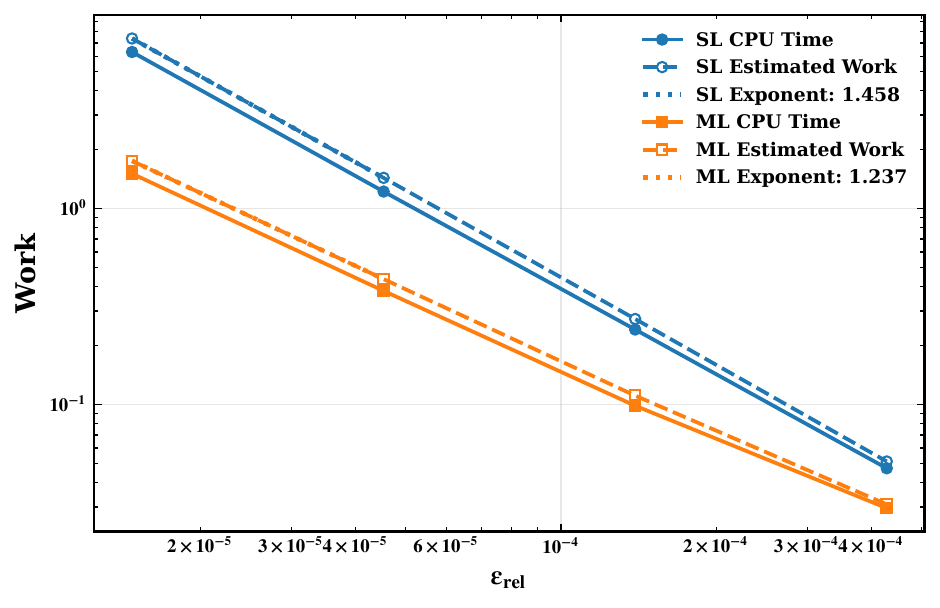}
        \caption{CPU time and estimated work.}
    \end{subfigure}
    \hfill
    \begin{subfigure}[t]{0.49\textwidth}
        \centering
        \includegraphics[width=\linewidth]
        {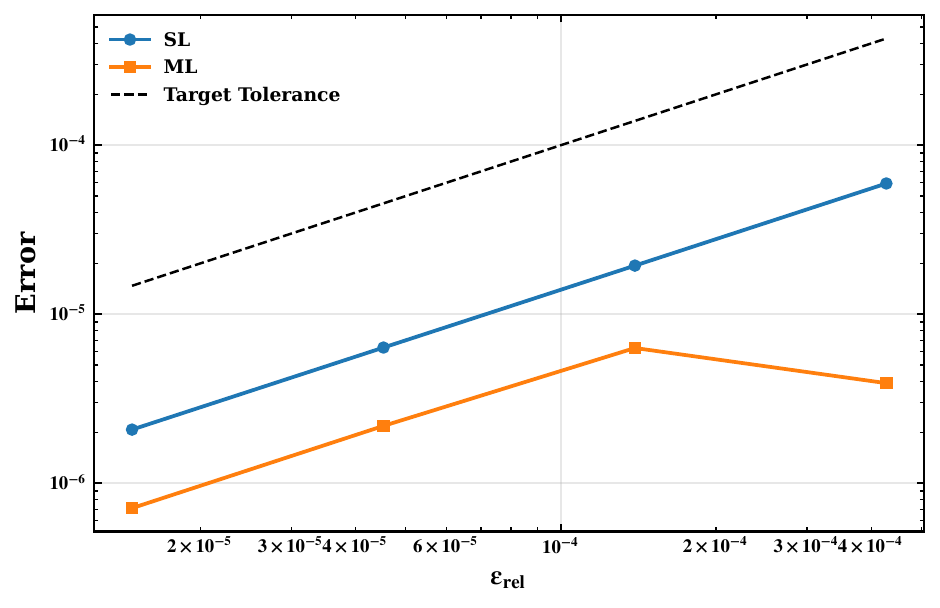}
        \caption{Achieved relative pricing errors.}
    \end{subfigure}
   \caption{{
    Comparison of the single-level and multilevel pricing methods for the
    EuRos configuration using the algebraic quadrature error estimate, for
    varying relative tolerances.
    }}
    \label{fig:single_multilevel_algebraic}
\end{figure}

\paragraph{{Root-exponential quadrature estimate.}}

\blue{The separation between SL and ML is generally expected to be less
pronounced under the root-exponential estimate because the difference in
convergence rates is only logarithmic. We therefore consider a challenging
case using the SPY parameter set in
Table~\ref{tab:numerical_benchmarks} with the short maturity}
\(T=2/365\), \(S_0=K=1\), and \(r=0\). {As shown in
Figure~\ref{fig:GL_decay_maturity}, the integrands decay more slowly as
\(u\) increases at short maturities. This makes the single-level quadrature
more demanding.} The single-level
{number of quadrature points} is selected
from~\eqref{eq:optimal_N_root_exponential}. {For the multilevel method
based on the root-exponential quadrature error estimate, the numbers of
quadrature points at the different levels are obtained by  solving
\eqref{eq:ml_root_exponential_optimization}.}
{The correction prefactors used in this optimization are constructed
with the level-scaling rule of
Remark~\ref{rem:practical_ml_quadrature_parameters}.}

Figure~\ref{fig:single_multilevel_exponential} shows that the multilevel method
\blue{requires less computational work and has more than ten times lower CPU
time at the smallest tested tolerance}, while both {methods satisfy the
prescribed relative tolerance}. {Section~\ref{sec:ML_method} does not
provide an asymptotic complexity result for the multilevel method based on the
root-exponential quadrature error estimate. The comparison in
Figure~\ref{fig:single_multilevel_exponential} therefore provides numerical
evidence that, for this configuration, selecting the number of quadrature
points separately at each level reduces both the computational work and the
CPU time relative to the single-level method.}
\begin{figure}[!htbp]
    \centering
    \begin{subfigure}[t]{0.49\textwidth}
        \centering
        \includegraphics[width=\linewidth]
        {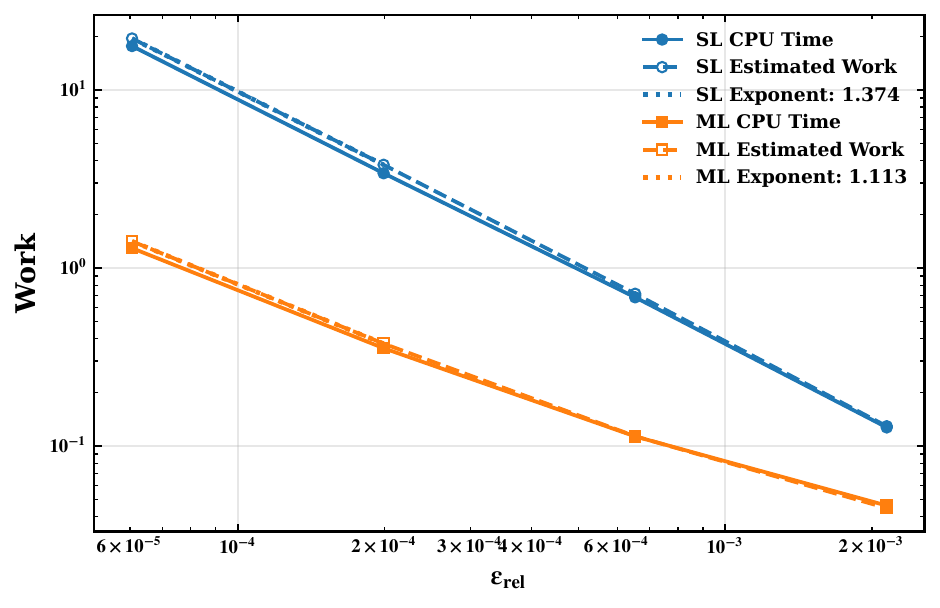}
        \caption{CPU time and estimated work.}
    \end{subfigure}
    \hfill
    \begin{subfigure}[t]{0.49\textwidth}
        \centering
        \includegraphics[width=\linewidth]
        {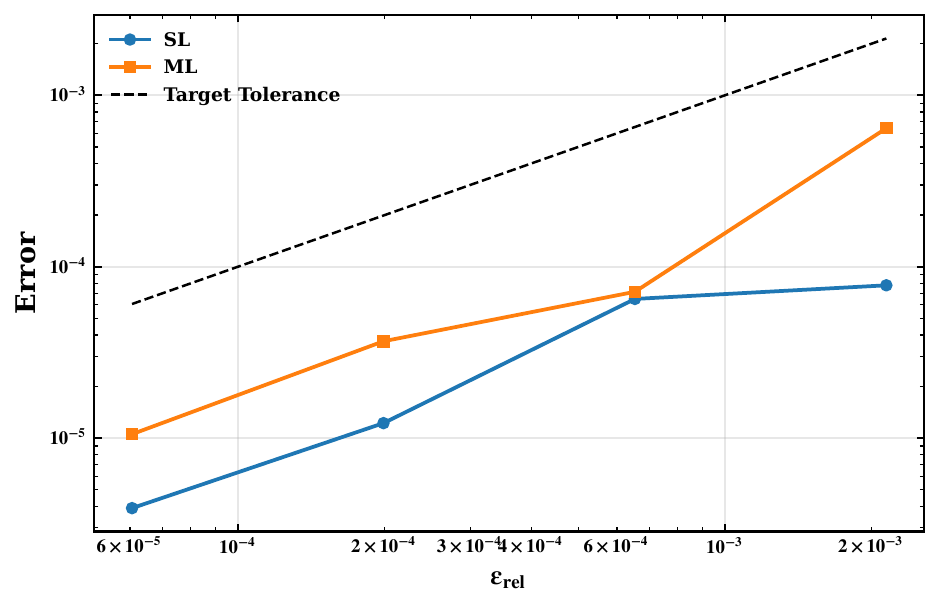}
        \caption{Achieved relative pricing errors.}
    \end{subfigure}
   \caption{{
    Comparison of the single-level and multilevel pricing methods for the
    two-day SPY parameter set using the root-exponential quadrature error
    estimate. 
    }}
    \label{fig:single_multilevel_exponential}
\end{figure}

\subsubsection{Comparison with the BL2 Method}
\label{subsubsec:bl2_comparison}

{We next compare the multilevel method with the BL2 Markovian
approximation of Bayer and Breneis~\cite{bayer2023weak}. We restrict the
comparison to the multilevel method because
Figures~\ref{fig:single_multilevel_algebraic}
and~\ref{fig:single_multilevel_exponential} show that it \blue{outperforms the
single-level method in terms of computational work} over the relative
tolerances considered. We use the SPY and EuRos parameter sets in
Table~\ref{tab:numerical_benchmarks} and consider maturities of two days and
one week. For each configuration, the two methods use the same model, payoff,
and market parameters. We choose the BL2 approximation because Bayer and
Breneis~\cite{bayer2023weak} report it as the \blue{best-performing} among the Markovian
approximations considered in their numerical comparison. For each target
relative tolerance \(\varepsilon_{\mathrm{rel}}\), we run the reference
implementation~\cite{breneisGithub} with internal tolerance
\(\varepsilon_{\mathrm{rel}}/2\) and select the smallest number of factors
\(n_{\mathrm{BL2}}\in\{1,2,3,4\}\) for which the relative pricing error with
respect to the high-accuracy reference price satisfies the prescribed
tolerance. The BL2 approximation and selection of
\(n_{\mathrm{BL2}}\) are described in Appendix~\ref{app:bl2_method}.}
\red{This factor-count selection uses $V_{\mathrm{ref}}$ and is therefore a
reference-assisted benchmarking procedure, not a production rule for selecting
$n_{\mathrm{BL2}}$.}

{Figures~\ref{fig:bl2_spy_comparison} and
\ref{fig:bl2_euros_comparison} compare the pricing and total CPU times of the
multilevel and BL2 methods. For BL2, the pricing CPU time is the CPU time
required to evaluate the option price after the number of factors
\(n_{\mathrm{BL2}}\) has been selected and the corresponding quadrature nodes
and weights have been computed. The total CPU time additionally includes the
selection of \(n_{\mathrm{BL2}}\). Starting from \(n_{\mathrm{BL2}}=1\), the
BL2 approximation is constructed and the option price is computed for each
candidate number of factors until the resulting relative pricing error
satisfies the prescribed tolerance. Each CPU time is the mean of three
repetitions, and the selected value of \(n_{\mathrm{BL2}}\) is displayed
beside the corresponding BL2 point.}
\begin{figure}[!htbp]
    \centering
    \begin{subfigure}[t]{0.49\textwidth}
        \centering
        \includegraphics[width=\linewidth]
        {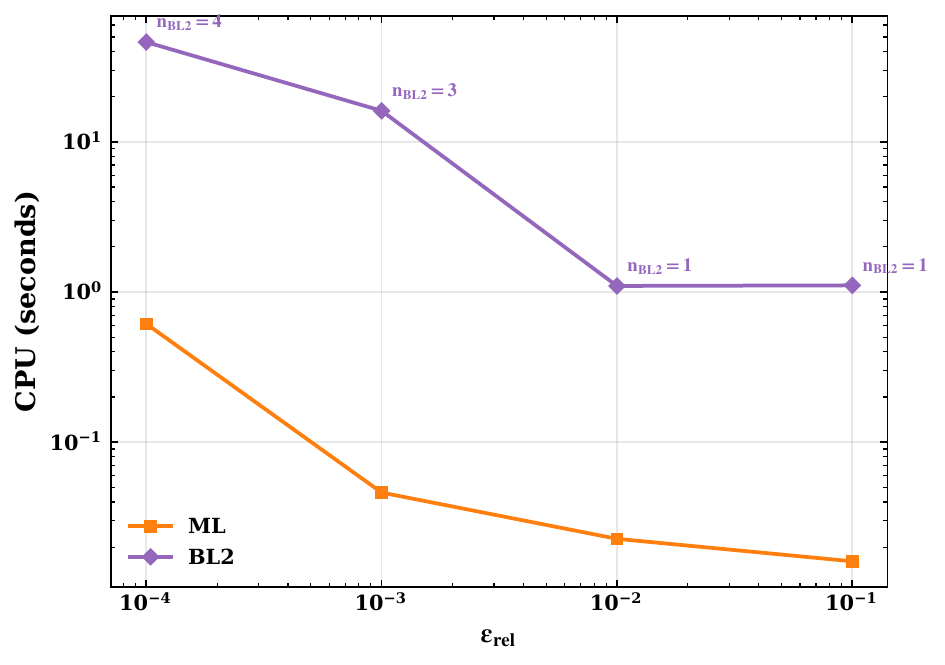}
        \caption{Two-day maturity: pricing CPU time.}
    \end{subfigure}
    \hfill
    \begin{subfigure}[t]{0.49\textwidth}
        \centering
        \includegraphics[width=\linewidth]
        {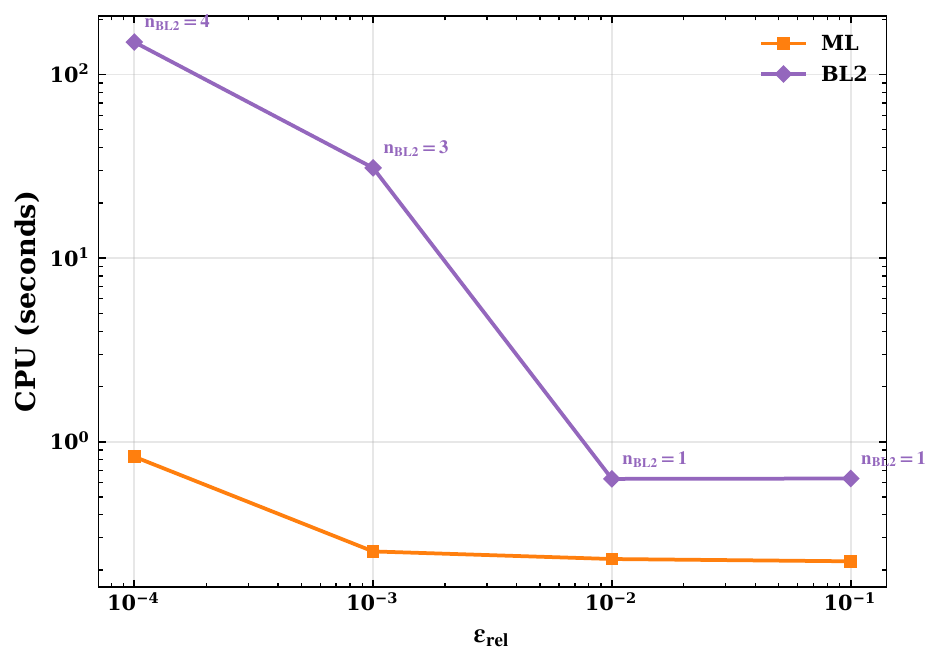}
        \caption{Two-day maturity: total CPU time.}
    \end{subfigure}
    \medskip
    \begin{subfigure}[t]{0.49\textwidth}
        \centering
        \includegraphics[width=\linewidth]
        {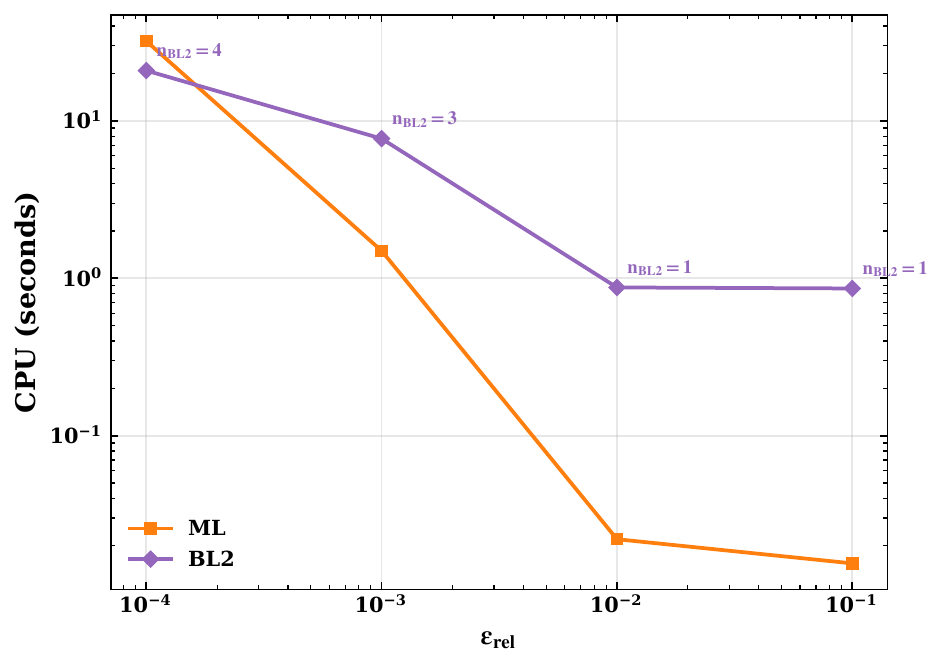}
        \caption{One-week maturity: pricing CPU time.}
    \end{subfigure}
    \hfill
    \begin{subfigure}[t]{0.49\textwidth}
        \centering
        \includegraphics[width=\linewidth]
        {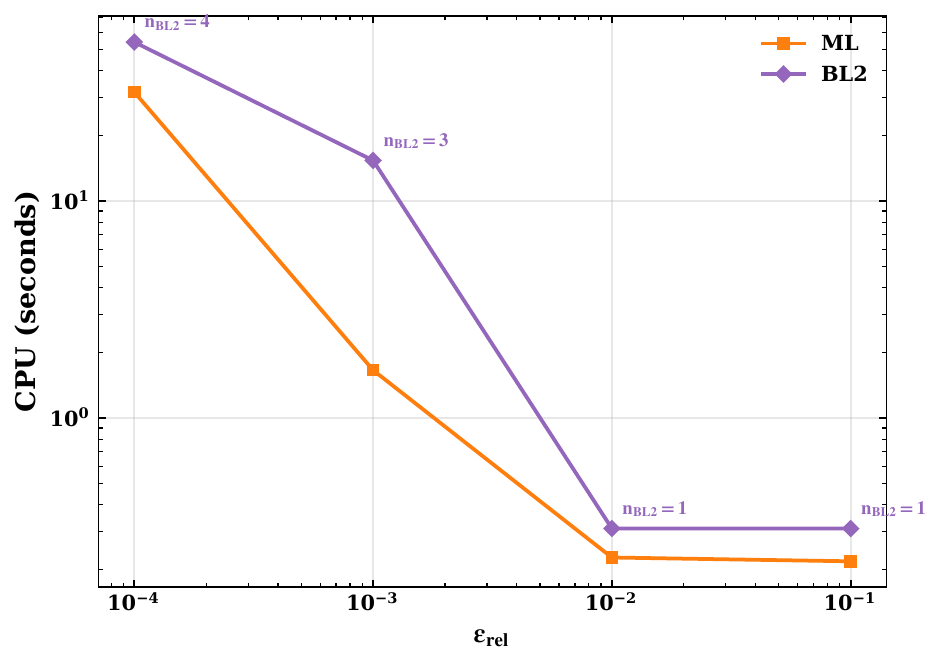}
        \caption{One-week maturity: total CPU time.}
    \end{subfigure}
    \caption{
    ML and BL2 CPU-time comparisons for the SPY parameter set using the
    pricing and total CPU times defined above.
    }
    \label{fig:bl2_spy_comparison}
\end{figure}

\begin{figure}[!htbp]
    \centering
    \begin{subfigure}[t]{0.49\textwidth}
        \centering
        \includegraphics[width=\linewidth]
        {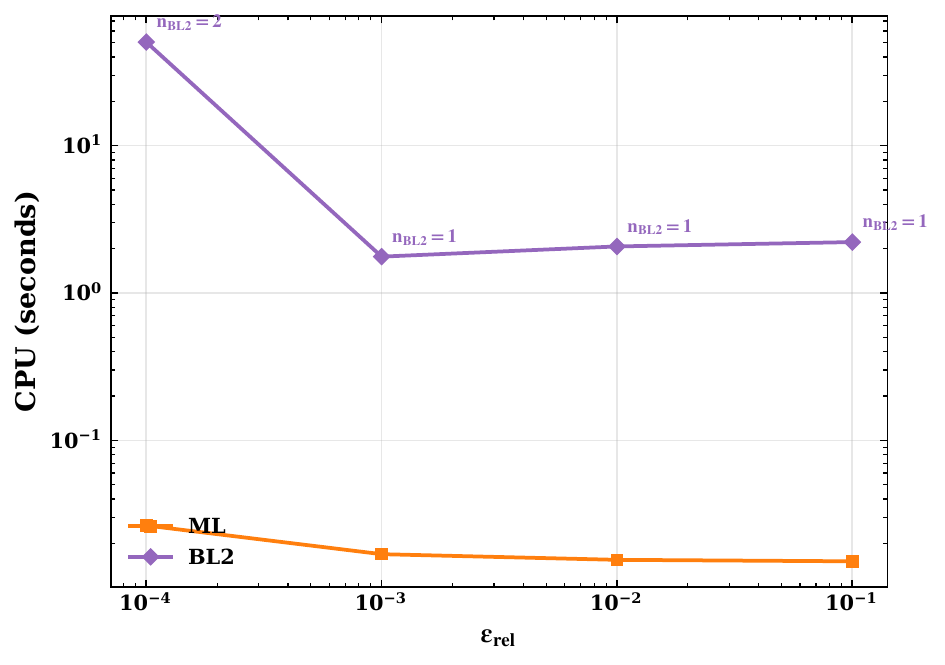}
        \caption{Two-day maturity: pricing CPU time.}
    \end{subfigure}
    \hfill
    \begin{subfigure}[t]{0.49\textwidth}
        \centering
        \includegraphics[width=\linewidth]
        {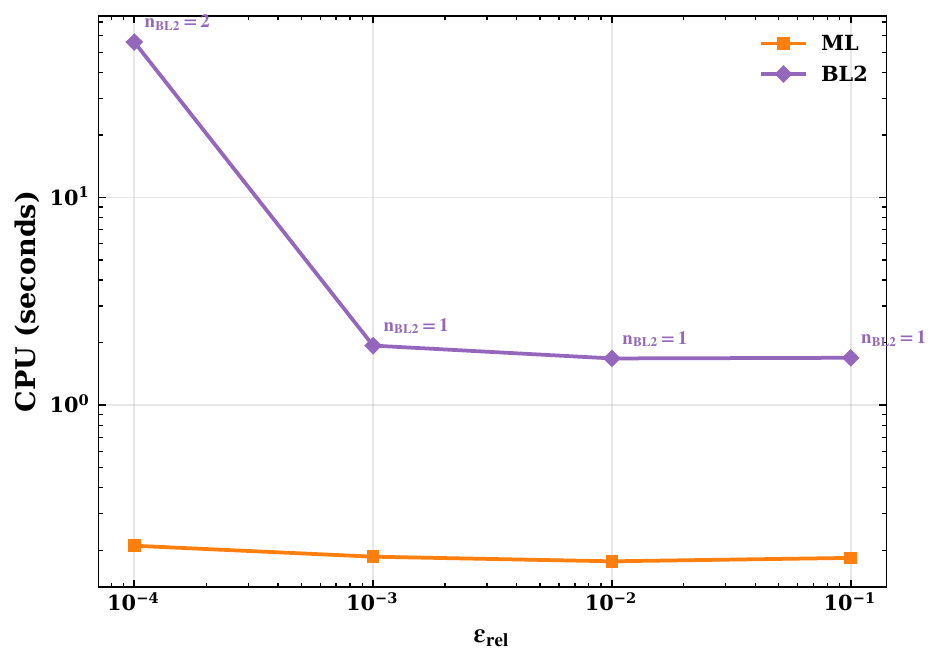}
        \caption{Two-day maturity: total CPU time.}
    \end{subfigure}
    \medskip
    \begin{subfigure}[t]{0.49\textwidth}
        \centering
        \includegraphics[width=\linewidth]
        {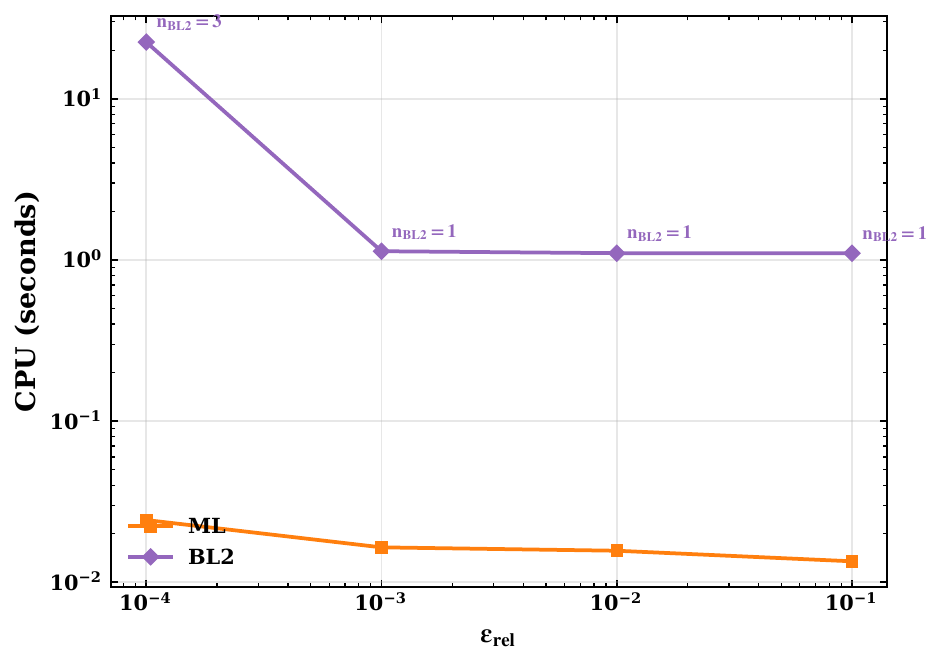}
        \caption{One-week maturity: pricing CPU time.}
    \end{subfigure}
    \hfill
    \begin{subfigure}[t]{0.49\textwidth}
        \centering
        \includegraphics[width=\linewidth]
        {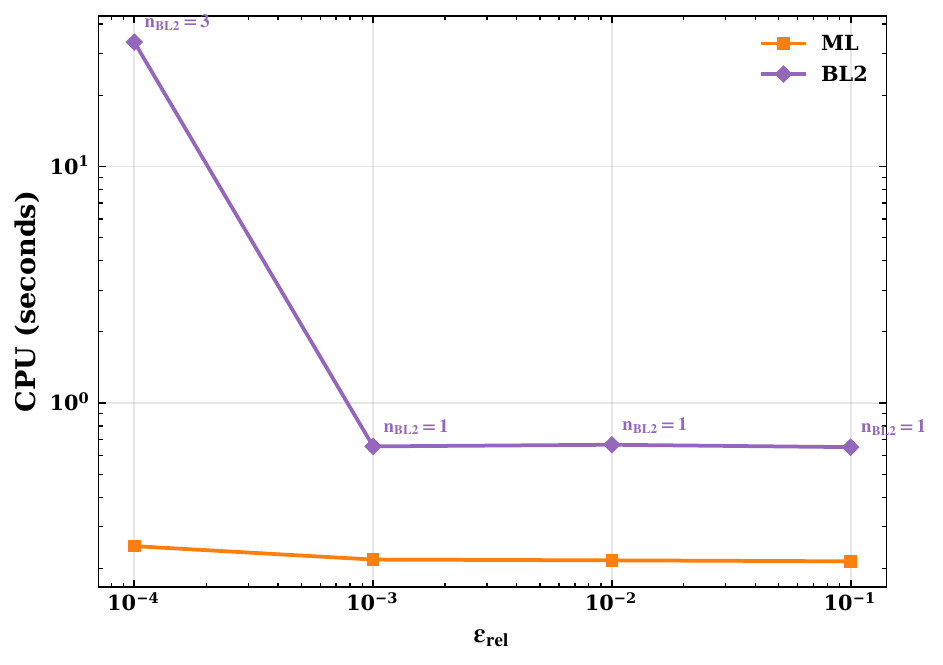}
        \caption{One-week maturity: total CPU time.}
    \end{subfigure}
    \caption{
    ML and BL2 CPU-time comparisons for the EuRos parameter set using the
    pricing and total CPU times defined above.
    }
    \label{fig:bl2_euros_comparison}
\end{figure}

\blue{Across the four configurations, ML is substantially faster in total CPU
time and is also faster in pricing CPU time except at the tightest tolerance
for the one-week SPY case.} These results demonstrate the effectiveness of the proposed hierarchical
framework, particularly in parameter regimes where the Fourier integrand has
limited effective regularity and is therefore more demanding to integrate.

\section{Conclusion and Future Work}
\label{sec:conclusion}

In this work, we introduced a hierarchical Fourier framework for pricing options under the rough Heston model. The main difficulty in this setting comes from the high cost of evaluating the characteristic function, which requires solving a fractional Riccati equation at every quadrature point in the Fourier domain. \blue{In the single-level construction, all Fourier nodes are evaluated on the same finest Riccati grid selected for the target accuracy. This can be costly when the required discretization is fine and the quadrature needs many nodes. Both the single- and multilevel constructions use scaled Gauss-Laguerre quadrature, with the scaling chosen to reflect the estimated decay of the Fourier integrand.}

{\color{blue}
For the single-level method, the conditional algebraic analysis gives
\[
W_{\mathrm{SL}}(\varepsilon)
=
\mathcal{O}\!\left(
\varepsilon^{-\beta/p-2/s_{\mathrm{SL}}}
\right),
\]
whereas, under the multilevel regularity and correction-decay assumptions of
Proposition~\ref{prop:work_scaling_multilevel}, the multilevel work satisfies
\[
W_{\mathrm{ML}}(\varepsilon)
=
\mathcal{O}\!\left(
\varepsilon^{-2/s_0}
+
\varepsilon^{-\beta/p}
\right),
\]
and reduces to $\mathcal{O}(\varepsilon^{-\beta/p})$ when
$s_0\geq 2p/\beta$. These are algebraic-cover complexity statements.
In the root-exponential regime used in the practical allocation, we do not
prove a separate multilevel complexity theorem. The numerical experiments
instead show that the level-dependent allocation reduces the number of
expensive fine-grid Riccati solves and provides substantial speedups while
maintaining the same level of accuracy as the single-level approach. The
numerical gains therefore depend on the quadrature regime, model/payoff
parameters, hierarchy, and setup cost.

The complexity results are stated for a generic Fourier-integrand convergence
rate $p$. For the fractional Adams discretization used in the numerical
experiments, the results in
Section~\ref{subsec:time_discretization_numerics} indicate the empirical rate
$p\approx1+\alpha$.

For the four reported SPY and EuRos benchmark configurations, and with the
BL2 factor count selected by the reference-assisted procedure described in
Section~\ref{subsubsec:bl2_comparison}, \red{the present ML implementation has
lower total CPU time in the reported measurements.}
}

The framework developed in this work suggests several directions for future work. One natural extension is the use of adaptive time-stepping schemes for the fractional Riccati equation, which may reduce the discretization cost even further. It would also be interesting to study alternative quadrature rules, such as Gauss-Kronrod or double-exponential formulas, and compare their performance with Gauss-Laguerre quadrature in rough volatility settings.

Also understanding the link between the {smoothness index} $s_L$ and the roughness parameter $\alpha$ would give deeper insight into the behavior of the Fourier integrand. Finally, the proposed method can be adapted to other models where the characteristic function {is} semi-explicit, such as signature volatility models or polynomial Ornstein-Uhlenbeck volatility models.

\section*{\blue{Use of artificial intelligence}}

\blue{AI assistance was used to improve writing clarity and presentation and to help identify relevant references. \red{The authors take responsibility} for the mathematical content, numerical results, and arguments.}

\newpage
\appendix

\section{Proof of Lemma~\ref{lemma:terminal_exponent_error}}
\label{sec:proof_terminal_exponent_error}

{\color{blue}
\begin{proof}
Fix $u\geq0$ and set $\xi=u+\mathrm{i}R$. Introduce the exponent
obtained by applying the time-integration rule to the exact Riccati solution
at the grid points:
\[
\widetilde G_\ell(\xi)
:=
\mathrm{i}\xi(X_0+rT)
+
\mathcal T_\ell
\left[
\left(
J\!\left(\xi,h(\xi,t_{\ell,j})\right)
\right)_{j=0}^{M_\ell}
\right].
\]
Then
\[
E_\ell(\xi)
=
\left(G(\xi)-\widetilde G_\ell(\xi)\right)
+
\left(\widetilde G_\ell(\xi)-G_\ell(\xi)\right).
\]
Assumption~\ref{ass:time_integration_error} gives
\begin{equation}
\left|G(\xi)-\widetilde G_\ell(\xi)\right|
\leq
C_{\mathrm{time},\xi}\Delta t_\ell^{\,r_T}.
\label{eq:exact_exponent_quadrature_bound}
\end{equation}

For $j=0,\ldots,M_\ell$, define
\[
d_{\ell,j}(\xi)
:=
h(\xi,t_{\ell,j})-h_{\ell,j}(\xi).
\]
Since $h(\xi,0)=h_{\ell,0}(\xi)=0$, we have $d_{\ell,0}(\xi)=0$.
Set
\[
a_\xi
:=
\theta\gamma+V_0\gamma(\mathrm{i}\xi\rho\nu-1),
\qquad
b
:=
V_0\frac{(\gamma\nu)^2}{2},
\qquad
H_\xi
:=
\|h(\xi,\cdot)\|_{L^\infty(0,T)}.
\]
By~\eqref{eq:exponent_integrand_J} and~\eqref{eq:F_def},
\[
\begin{aligned}
&J\!\left(\xi,h(\xi,t_{\ell,j})\right)
-J\!\left(\xi,h_{\ell,j}(\xi)\right)
\\
&\quad=
a_\xi d_{\ell,j}(\xi)
+
b
\left(
2h(\xi,t_{\ell,j})d_{\ell,j}(\xi)
-d_{\ell,j}(\xi)^2
\right),
\end{aligned}
\]
and therefore
\[
\begin{aligned}
&\left|
J\!\left(\xi,h(\xi,t_{\ell,j})\right)
-J\!\left(\xi,h_{\ell,j}(\xi)\right)
\right|
\\
&\quad\leq
\left(|a_\xi|+2|b|H_\xi\right)|d_{\ell,j}(\xi)|
+
|b||d_{\ell,j}(\xi)|^2.
\end{aligned}
\]
The positivity of the trapezoidal weights gives
\[
\begin{aligned}
\left|\widetilde G_\ell(\xi)-G_\ell(\xi)\right|
&\leq
\left(|a_\xi|+2|b|H_\xi\right)
\mathcal T_\ell
\left[
\left(
|d_{\ell,j}(\xi)|
\right)_{j=0}^{M_\ell}
\right]
\\
&\quad+
|b|
\mathcal T_\ell
\left[
\left(
|d_{\ell,j}(\xi)|^2
\right)_{j=0}^{M_\ell}
\right].
\end{aligned}
\]
By Assumption~\ref{ass:riccati_pointwise_error},
\[
\begin{aligned}
\mathcal T_\ell
\left[
\left(
|d_{\ell,j}(\xi)|
\right)_{j=0}^{M_\ell}
\right]
&\leq
C_\xi\Delta t_\ell^{q+1}
\sum_{j=1}^{M_\ell}t_{\ell,j}^{\alpha-1}
\\
&=
C_\xi\Delta t_\ell^{q+\alpha}
\sum_{j=1}^{M_\ell}j^{\alpha-1},
\end{aligned}
\]
and
\[
\begin{aligned}
\mathcal T_\ell
\left[
\left(
|d_{\ell,j}(\xi)|^2
\right)_{j=0}^{M_\ell}
\right]
&\leq
C_\xi^2\Delta t_\ell^{2q+1}
\sum_{j=1}^{M_\ell}t_{\ell,j}^{2\alpha-2}
\\
&=
C_\xi^2\Delta t_\ell^{2q+2\alpha-1}
\sum_{j=1}^{M_\ell}j^{2\alpha-2}.
\end{aligned}
\]
Since $\alpha\in(\frac12,1)$,
\[
\sum_{j=1}^{M}j^{\alpha-1}
\leq
\frac{M^\alpha}{\alpha},
\qquad
\sum_{j=1}^{M}j^{2\alpha-2}
\leq
\frac{M^{2\alpha-1}}{2\alpha-1}.
\]
Using $M_\ell\Delta t_\ell=T$, we obtain
\begin{equation}
\mathcal T_\ell
\left[
\left(
|d_{\ell,j}(\xi)|
\right)_{j=0}^{M_\ell}
\right]
\leq
C_\xi\frac{T^\alpha}{\alpha}\Delta t_\ell^q,
\label{eq:nodal_error_trapezoidal_l1}
\end{equation}
and
\begin{equation}
\mathcal T_\ell
\left[
\left(
|d_{\ell,j}(\xi)|^2
\right)_{j=0}^{M_\ell}
\right]
\leq
C_\xi^2\frac{T^{2\alpha-1}}{2\alpha-1}
\Delta t_\ell^{2q}.
\label{eq:nodal_error_trapezoidal_l2}
\end{equation}
Consequently,
\[
\begin{aligned}
\left|\widetilde G_\ell(\xi)-G_\ell(\xi)\right|
&\leq
\left(|a_\xi|+2|b|H_\xi\right)
C_\xi\frac{T^\alpha}{\alpha}\Delta t_\ell^q
\\
&\quad+
|b|C_\xi^2\frac{T^{2\alpha-1}}{2\alpha-1}
\Delta t_\ell^{2q}
\\
&\leq
C_{\mathrm R,\xi,T}\Delta t_\ell^q,
\end{aligned}
\]
where
\[
C_{\mathrm R,\xi,T}
:=
\left(|a_\xi|+2|b|H_\xi\right)
C_\xi\frac{T^\alpha}{\alpha}
+
|b|C_\xi^2\frac{T^{2\alpha-1}}{2\alpha-1}
\Delta t_0^q.
\]
Combining this estimate with~\eqref{eq:exact_exponent_quadrature_bound},
\[
\begin{aligned}
|E_\ell(\xi)|
&\leq
C_{\mathrm{time},\xi}\Delta t_\ell^{r_T}
+
C_{\mathrm R,\xi,T}\Delta t_\ell^q
\\
&=
\Delta t_\ell^{p_G}
\left(
C_{\mathrm{time},\xi}\Delta t_\ell^{r_T-p_G}
+
C_{\mathrm R,\xi,T}\Delta t_\ell^{q-p_G}
\right)
\\
&\leq
\left(
C_{\mathrm{time},\xi}\Delta t_0^{r_T-p_G}
+
C_{\mathrm R,\xi,T}\Delta t_0^{q-p_G}
\right)
\Delta t_\ell^{p_G}
\\
&=:
C_{E,\xi,T}\Delta t_\ell^{p_G}.
\end{aligned}
\]
Therefore,
\[
A_E(u)
=
\sup_{\ell\geq0}
\Delta t_\ell^{-p_G}|E_\ell(u+\mathrm{i}R)|
\leq
C_{E,\xi,T}
<
\infty,
\]
which proves~\eqref{eq:exponent_error_pG}.
\end{proof}
}

\section{Proof of Proposition~\ref{prop:discretization_error}}
\label{sec:proof_fourier_integrand_error}

\begin{proof}
Let $u\geq0$ and $\ell\geq0$. By Lemma~\ref{lemma:terminal_exponent_error},
\begin{equation}
\left|E_\ell(u+\mathrm{i}R)\right|
\leq
A_E(u)\Delta t_\ell^{\blue{p_G}}.
\label{eq:terminal_exponent_error_in_proof}
\end{equation}
Moreover, for every $z\in\mathbb C$,
\begin{equation}
\left|1-\exp(-z)\right|
\leq
|z|\exp(|z|).
\label{eq:complex_exponential_difference_bound}
\end{equation}
Combining~\eqref{eq:terminal_exponent_error_in_proof}
and~\eqref{eq:complex_exponential_difference_bound}, and using
$\Delta t_\ell\leq\Delta t_0$, gives
\begin{equation}
\Delta t_\ell^{-\blue{p_G}}
\left|1-\exp\!\left(-E_\ell(u+\mathrm{i}R)\right)\right|
\leq
A_E(u)\exp\!\left(A_E(u)\Delta t_0^{\blue{p_G}}\right).
\label{eq:normalized_exponential_difference_bound}
\end{equation}
Taking the supremum over $\ell\geq0$
in~\eqref{eq:normalized_exponential_difference_bound} gives
$\Psi(u)<\infty$. Since $E_\ell=G-G_\ell$,
\begin{equation}
\left|\Phi(u+\mathrm{i}R)-\Phi_\ell(u+\mathrm{i}R)\right|
=
\left|\Phi(u+\mathrm{i}R)\right|
\left|1-\exp\!\left(-E_\ell(u+\mathrm{i}R)\right)\right|
\leq
\left|\Phi(u+\mathrm{i}R)\right|\Psi(u)
\Delta t_\ell^{\blue{p_G}}.
\label{eq:characteristic_function_difference_bound}
\end{equation}
By the definition of the Fourier integrand,
$|\operatorname{Re}(z)|\leq|z|$ for every $z\in\mathbb C$, and
\eqref{eq:characteristic_function_difference_bound},
\begin{equation}
\begin{aligned}
\|g-g_\ell\|_{L^1(0,\infty)}
&\leq
\frac{e^{-rT}}{2\pi}
\int_0^\infty
\left|\widehat P(u+\mathrm{i}R)\right|
\left|\Phi(u+\mathrm{i}R)-\Phi_\ell(u+\mathrm{i}R)\right|\,du \\
&\leq
\frac{e^{-rT}}{2\pi}\Delta t_\ell^{\blue{p_G}}
\int_0^\infty
\left|\widehat P(u+\mathrm{i}R)\right|
\left|\Phi(u+\mathrm{i}R)\right|
\Psi(u)\,du.
\end{aligned}
\label{eq:fourier_integrand_difference_bound}
\end{equation}
\blue{By Conjecture~\ref{conj:contour_domination},
\[
C_g
:=
\frac{e^{-rT}}{2\pi}
\int_0^\infty \mathcal D_R(u)\,du
<
\infty.
\]
Therefore,~\eqref{eq:fourier_integrand_difference_bound} gives
\[
\|g-g_\ell\|_{L^1(0,\infty)}
\leq
C_g\Delta t_\ell^{p_G},
\]
which proves~\eqref{eq:g_error_pG}.}
\end{proof}

\section{Proof of Proposition~\ref{prop:richardson}}
\label{sec:proof_richardson}

\begin{proof}
{
Assumption~\ref{ass:sharp_price_expansion} gives
\begin{equation}
V_\ell-V
=
c_V\Delta t_\ell^{\,p}
+
o\!\left(\Delta t_\ell^{\,p}\right).
\label{eq:signed_discretization_error_expansion}
\end{equation}
Since $c_V\neq0$, it follows that
\begin{equation}
e_{\mathrm{disc}}(\ell)
=
|V-V_\ell|
=
|c_V|\Delta t_\ell^{\,p}
\left(1+o(1)\right).
\label{eq:discretization_error_expansion}
\end{equation}

Because $\Delta t_{\ell-1}=2\Delta t_\ell$, applying
\eqref{eq:signed_discretization_error_expansion} at levels $\ell$ and
$\ell-1$ gives
\begin{equation}
V_\ell-V_{\ell-1}
=
-c_V\left(2^p-1\right)\Delta t_\ell^{\,p}
+
o\!\left(\Delta t_\ell^{\,p}\right).
\label{eq:signed_level_difference_expansion}
\end{equation}
Hence
\begin{equation}
|V_\ell-V_{\ell-1}|
=
\left(2^p-1\right)|c_V|\Delta t_\ell^{\,p}
\left(1+o(1)\right).
\label{eq:level_difference_expansion}
\end{equation}
Comparing~\eqref{eq:discretization_error_expansion}
and~\eqref{eq:level_difference_expansion} proves
\eqref{eq:richardson_error_estimate}.
}
\end{proof}

\section{Proof of Proposition~\ref{prop:work_scaling_single}}
\label{app:complexity_proof_single}

\begin{proof}
{
We estimate separately the two factors in the single-level work
$W_{\mathrm{SL}}=N_{\mathrm{alg}}W_L$, starting with the cost $W_L$ of one
evaluation at level $L$.

Because $L$ is the smallest level satisfying
\eqref{eq:theoretical_level_bound}, the coarser level $L-1$ does not satisfy
the discretization constraint. Thus, with
$\delta_\varepsilon:=(\varepsilon_{\mathrm{disc}}/(2C_g))^{1/p}$, for all
sufficiently small $\varepsilon$,
\[
\Delta t_L
\leq
\delta_\varepsilon
<
\Delta t_{L-1}
=
2\Delta t_L.
\]
Hence $\Delta t_L\asymp\varepsilon^{1/p}$, and the cost model
$W_L=\mathcal O(\Delta t_L^{-\beta})$ gives
\[
W_L
=
\mathcal O\!\left(\varepsilon^{-\beta/p}\right).
\]

It remains to bound the number of quadrature points. By~\eqref{eq:optimal_N},
\[
N_{\mathrm{alg}}
\leq
1+
\left(
\frac{A_L}{\varepsilon_{\mathrm{quad}}}
\right)^{2/s_L}.
\]
\red{For sufficiently small $\varepsilon$, we have
$A/\varepsilon_{\mathrm{quad}}\geq1$.}
Since $A_L\leq A$ and {$s_L\geq s_{\mathrm{SL}}$}, this gives
\[
N_{\mathrm{alg}}
\leq
1+
\left(
\frac{A}{\varepsilon_{\mathrm{quad}}}
\right)^{{2/s_{\mathrm{SL}}}}
=
\mathcal O\!\left({\varepsilon^{-2/s_{\mathrm{SL}}}}\right).
\]

By construction, $L$ and $N_{\mathrm{alg}}$ satisfy both constraints
in~\eqref{eq:tolerance_allocation}, and therefore
$|V-V_{N_{\mathrm{alg}},L}|\leq\varepsilon$. Multiplying the two estimates
gives
\[
W_{\mathrm{SL}}(\varepsilon)
=
N_{\mathrm{alg}}W_L
=
\mathcal O\!\left(
\varepsilon^{-\beta/p}{\varepsilon^{-2/s_{\mathrm{SL}}}}
\right)
=
\mathcal O\!\left(
{\varepsilon^{-\left(\beta/p+2/s_{\mathrm{SL}}\right)}}
\right).
\]
}
\end{proof}

{
\section{Proof of Proposition~\ref{prop:work_scaling_single_root_exponential}}
\label{app:complexity_proof_single_root_exponential}

\begin{proof}
Recall that
\(
\varepsilon_{\mathrm{disc}}
=
\varepsilon_{\mathrm{quad}}
=
\varepsilon/2
\).
The choice of \(L\) gives
\[
e_{\mathrm{disc}}(L)
\leq
\varepsilon_{\mathrm{disc}},
\]
and Proposition~\ref{prop:quadrature_order_root_exponential} gives
\[
e_{\mathrm{quad}}(N_{\mathrm{re}},L)
\leq
\varepsilon_{\mathrm{quad}}.
\]
Therefore,
\[
|V-V_{N_{\mathrm{re}},L}|
\leq
\varepsilon.
\]

We now prove the asymptotic work estimate under the additional uniform bounds
on \(A_{L(\varepsilon)}^{\mathrm{re}}\) and \(B_{L(\varepsilon)}\). Because
\(L\) is the smallest level satisfying
\eqref{eq:theoretical_level_bound}, with
\[
\delta_\varepsilon
:=
\left(
\frac{\varepsilon_{\mathrm{disc}}}{2C_g}
\right)^{1/p},
\]
we have, for all sufficiently small \(\varepsilon\),
\[
\Delta t_L
\leq
\delta_\varepsilon
<
\Delta t_{L-1}
=
2\Delta t_L.
\]
Hence
\[
\Delta t_L
\asymp
\varepsilon_{\mathrm{disc}}^{1/p}
\asymp
\varepsilon^{1/p},
\]
and
\[
W_L
=
\mathcal O\!\left(\Delta t_L^{-\beta}\right)
=
\mathcal O\!\left(\varepsilon^{-\beta/p}\right).
\]

By~\eqref{eq:optimal_N_root_exponential},
\[
N_{\mathrm{re}}
\leq
1+
\frac{1}{B_L^2}
\log^2\!\left(
\frac{A_L^{\mathrm{re}}}{\varepsilon_{\mathrm{quad}}}
\right).
\]
Since \(A_L^{\mathrm{re}}\leq A^{\mathrm{re}}\) and \(B_L\geq B\),
\[
N_{\mathrm{re}}
\leq
1+
\frac{1}{B^2}
\log^2\!\left(
\frac{A^{\mathrm{re}}}{\varepsilon_{\mathrm{quad}}}
\right)
=
\mathcal O\!\left(\log^2(1/\varepsilon_{\mathrm{quad}})\right)
=
\mathcal O\!\left(\log^2(1/\varepsilon)\right).
\]

Consequently,
\[
W_{\mathrm{SL}}(\varepsilon)
=
N_{\mathrm{re}}W_L
=
\mathcal O\!\left(
\varepsilon^{-\beta/p}\log^2(1/\varepsilon)
\right).
\]
\end{proof}
}

{
\section{Proof of Proposition~\ref{prop:optimal_N_ml}}
\label{app:optimal_N_ml}

\begin{proof}
{
The split constraints in~\eqref{eq:ml_quadrature_constraint} decouple the
level-zero variable from the correction variables. At level zero, we solve
\[
\min_{N_0>0} W_0N_0
\qquad\text{subject to}\qquad
A_0N_0^{-s_0/2}
\leq
\frac{\varepsilon_{\mathrm{quad}}}{2}.
\]
Since the objective is increasing in \(N_0\), the constraint is active at the
minimizer, and therefore
\[
N_0^\star
=
\left(
\frac{2A_0}{\varepsilon_{\mathrm{quad}}}
\right)^{2/s_0}.
\]

For the correction levels, define
\[
c_\ell:=W_\ell+W_{\ell-1},
\qquad
\ell=1,\ldots,L.
\]
The remaining relaxed problem is
\[
\min_{N_1,\ldots,N_L>0}
\sum_{\ell=1}^L c_\ell N_\ell
\qquad\text{subject to}\qquad
\sum_{\ell=1}^L A_\ell N_\ell^{-s/2}
\leq
\frac{\varepsilon_{\mathrm{quad}}}{2}.
\]
Since \(N\mapsto N^{-s/2}\) is convex on \((0,\infty)\), this is a convex
optimization problem. Strictly feasible points exist by taking all \(N_\ell\)
sufficiently large, so the Karush-Kuhn-Tucker conditions characterize the
minimizer.

With a multiplier \(\lambda\geq0\), the Lagrangian is
\[
\mathcal L(\mathbf N,\lambda)
=
\sum_{\ell=1}^L c_\ell N_\ell
+
\lambda
\left(
\sum_{\ell=1}^L A_\ell N_\ell^{-s/2}
-
\frac{\varepsilon_{\mathrm{quad}}}{2}
\right).
\]
Stationarity gives
\[
c_\ell
=
\lambda\frac{s}{2}
A_\ell N_\ell^{-s/2-1},
\qquad
\ell=1,\ldots,L.
\]
Since \(c_\ell>0\) and \(A_\ell>0\), necessarily \(\lambda>0\), and the
correction constraint is active. Solving the stationarity condition gives
\[
N_\ell
=
\left(
\lambda\frac{s}{2}\frac{A_\ell}{c_\ell}
\right)^{2/(s+2)}.
\]
Substitution into the active constraint yields
\[
\left(\lambda\frac{s}{2}\right)^{-s/(s+2)}
\sum_{k=1}^L
A_k^{2/(s+2)}c_k^{s/(s+2)}
=
\frac{\varepsilon_{\mathrm{quad}}}{2},
\]
and hence
\[
\left(\lambda\frac{s}{2}\right)^{2/(s+2)}
=
\left(
\frac{2}{\varepsilon_{\mathrm{quad}}}
\sum_{k=1}^L
A_k^{2/(s+2)}c_k^{s/(s+2)}
\right)^{2/s}.
\]
Therefore,
\[
N_\ell^\star
=
\left(
\frac{A_\ell}{c_\ell}
\right)^{2/(s+2)}
\left(
\frac{2}{\varepsilon_{\mathrm{quad}}}
\sum_{k=1}^L
A_k^{2/(s+2)}c_k^{s/(s+2)}
\right)^{2/s},
\qquad
\ell=1,\ldots,L.
\]
Replacing \(c_\ell\) by \(W_\ell+W_{\ell-1}\), and combining this expression
with \(N_0^\star\), proves~\eqref{eq:optimal_N_ml}. Substitution
into~\eqref{eq:ml_work} gives~\eqref{eq:ml_work_optimal}.
}
\end{proof}
}

{
\section{Proof of Proposition~\ref{prop:work_scaling_multilevel}}
\label{app:complexity_proof_multilevel}

\begin{proof}
{
Since \(L\) is the smallest level satisfying~\eqref{eq:theoretical_level_bound},
for all sufficiently small \(\varepsilon\) we have \(L\geq1\) and
\[
\frac{1}{2}
\left(
\frac{\varepsilon_{\mathrm{disc}}}{2C_g}
\right)^{1/p}
<
\Delta t_L
\leq
\left(
\frac{\varepsilon_{\mathrm{disc}}}{2C_g}
\right)^{1/p}.
\]
Hence
\begin{equation}
\Delta t_L
\asymp
\varepsilon_{\mathrm{disc}}^{1/p}
\asymp
\varepsilon^{1/p}.
\label{eq:ml_level_scaling}
\end{equation}

For \(\ell\geq1\),~\eqref{eq:multilevel_cost_bound} and
\(\Delta t_{\ell-1}=2\Delta t_\ell\) give
\[
W_\ell+W_{\ell-1}
\leq
C_W\left(
\Delta t_\ell^{-\beta}
+
\Delta t_{\ell-1}^{-\beta}
\right)
=
C_W(1+2^{-\beta})\Delta t_\ell^{-\beta}.
\]
Together with Assumption~\ref{ass:level_difference_decay}, this gives
\[
A_\ell^{2/(s+2)}
\left(W_\ell+W_{\ell-1}\right)^{s/(s+2)}
=
\mathcal O\!\left(
\Delta t_\ell^{(2p-\beta s)/(s+2)}
\right).
\]
The assumption \(s>2p/\beta\) ensures that the exponent is negative. Since
\(\Delta t_\ell=\Delta t_0 2^{-\ell}\),
\[
\begin{aligned}
\sum_{\ell=1}^L
\Delta t_\ell^{(2p-\beta s)/(s+2)}
&=
\Delta t_0^{(2p-\beta s)/(s+2)}
\sum_{\ell=1}^L
2^{(\beta s-2p)\ell/(s+2)}
\\
&=
\mathcal O\!\left(
\Delta t_L^{(2p-\beta s)/(s+2)}
\right).
\end{aligned}
\]
\begin{equation}
\sum_{\ell=1}^L
A_\ell^{2/(s+2)}
\left(W_\ell+W_{\ell-1}\right)^{s/(s+2)}
=
\mathcal O\!\left(
\Delta t_L^{(2p-\beta s)/(s+2)}
\right).
\label{eq:ml_aggregate_sum_bound}
\end{equation}
Substituting~\eqref{eq:ml_aggregate_sum_bound}
into the correction term of~\eqref{eq:ml_work_optimal}, and using
\(\varepsilon_{\mathrm{quad}}=\varepsilon/2\) and
\eqref{eq:ml_level_scaling}, gives
\begin{equation}
\begin{aligned}
\left(
\frac{2}{\varepsilon_{\mathrm{quad}}}
\right)^{2/s}
\left(
\sum_{\ell=1}^L
A_\ell^{2/(s+2)}
\left(W_\ell+W_{\ell-1}\right)^{s/(s+2)}
\right)^{(s+2)/s}
&=
\mathcal O\!\left(
\varepsilon^{-2/s}
\Delta t_L^{\,2p/s-\beta}
\right)
\\
&=
\mathcal O\!\left(
\varepsilon^{-2/s}
\varepsilon^{(2p/s-\beta)/p}
\right)
=
\mathcal O\!\left(
\varepsilon^{-\beta/p}
\right).
\end{aligned}
\end{equation}

Because \(A_0\), \(W_0\), and the level-zero time step are independent of
\(\varepsilon\), the first term of~\eqref{eq:ml_work_optimal} satisfies
\[
W_0
\left(
\frac{2A_0}{\varepsilon_{\mathrm{quad}}}
\right)^{2/s_0}
=
\mathcal O\!\left(\varepsilon^{-2/s_0}\right).
\]
Consequently,
\begin{equation}
W_{\mathrm{ML}}(\mathbf N^\star,L)
=
\mathcal O\!\left(
\varepsilon^{-2/s_0}
+
\varepsilon^{-\beta/p}
\right).
\label{eq:ml_relaxed_work_complexity}
\end{equation}

It remains to account for the integer numbers of quadrature points. Set
\(N_\ell=\lceil N_\ell^\star\rceil\), \(\ell=0,\ldots,L\). Since rounding up
can only decrease the two quadrature-error bounds,
\[
A_0N_0^{-s_0/2}
\leq
A_0(N_0^\star)^{-s_0/2}
=
\frac{\varepsilon_{\mathrm{quad}}}{2},
\]
and
\[
\sum_{\ell=1}^L
A_\ell N_\ell^{-s/2}
\leq
\sum_{\ell=1}^L
A_\ell (N_\ell^\star)^{-s/2}
=
\frac{\varepsilon_{\mathrm{quad}}}{2}.
\]
Together with~\eqref{eq:ml_quadrature_split} and
\eqref{eq:common_order_estimate}, these inequalities give
\[
e_{\mathrm{quad}}(\mathbf N,L)
\leq
\varepsilon_{\mathrm{quad}}.
\]
By~\eqref{eq:theoretical_level_bound},
\(e_{\mathrm{disc}}(L)\leq\varepsilon_{\mathrm{disc}}\).
Therefore,
\[
|V-V_{\mathbf N,L}|
\leq
e_{\mathrm{disc}}(L)
+
e_{\mathrm{quad}}(\mathbf N,L)
\leq
\varepsilon.
\]

Finally, \(\lceil N_\ell^\star\rceil\leq N_\ell^\star+1\),
\eqref{eq:ml_work}, \eqref{eq:ml_relaxed_work_complexity},
\eqref{eq:multilevel_cost_bound}, the refinement relation
$\Delta t_{\ell-1}=2\Delta t_\ell$, and
\eqref{eq:ml_level_scaling} give
\[
\begin{aligned}
W_{\mathrm{ML}}(\mathbf N,L)
&\leq
W_{\mathrm{ML}}(\mathbf N^\star,L)
+
W_0
+
\sum_{\ell=1}^L
\left(W_\ell+W_{\ell-1}\right)
\\
&=
\mathcal O\!\left(
\varepsilon^{-2/s_0}
+
\varepsilon^{-\beta/p}
\right)
+
\mathcal O\!\left(
\sum_{\ell=1}^L\Delta t_\ell^{-\beta}
\right)
\\
&=
\mathcal O\!\left(
\varepsilon^{-2/s_0}
+
\varepsilon^{-\beta/p}
\right).
\end{aligned}
\]
This proves~\eqref{eq:work_scaling_multilevel_general}. If
\(s_0\geq2p/\beta\), then
\(\varepsilon^{-2/s_0}=\mathcal O(\varepsilon^{-\beta/p})\), which
gives~\eqref{eq:work_scaling_multilevel}.
}
\end{proof}
}

\section{{Numerical evidence for Conjecture~\ref{conj:contour_domination}}}
\label{app:contour_domination_numerics}

{
The exact function $\mathcal D_R$ in~\eqref{eq:domination_function} cannot be
evaluated because the characteristic function $\Phi$ is not available in
closed form. We therefore compute a finer-grid approximation $\bar\Phi$ and,
for each tested level $\ell$, define the computable normalized function
\begin{equation}
\overline{\mathcal D}_{\ell,R}(u)
:=
\red{\Delta t_\ell^{-q_{\mathrm{obs}}}}
\left|\widehat P(u+\mathrm{i}R)\right|
\left|\bar\Phi(u+\mathrm{i}R)-\Phi_\ell(u+\mathrm{i}R)\right|,
\qquad u\geq0.
\label{eq:numerical_domination_function}
\end{equation}
Here the overbar indicates that the finer-grid characteristic function is used
as the reference quantity. We use the observed rate
\red{$q_{\mathrm{obs}}=1+\alpha$} reported in
Section~\ref{subsec:time_discretization_numerics}.

We evaluate~\eqref{eq:numerical_domination_function} for several sufficiently
fine levels on a common Fourier grid. On the same grid, we compute the
cumulative trapezoidal approximations
\begin{equation}
\overline{I}_{\ell,R}(U)
:=
\int_0^U\overline{\mathcal D}_{\ell,R}(u)\,du.
\label{eq:numerical_domination_cumulative}
\end{equation}
}

\FloatBarrier
\ifconvertedfigures
\begin{figure}[ht]
\centering

\begin{subfigure}[t]{0.48\textwidth}
    \vspace{0pt}
    \centering
    \includegraphics[width=\linewidth]{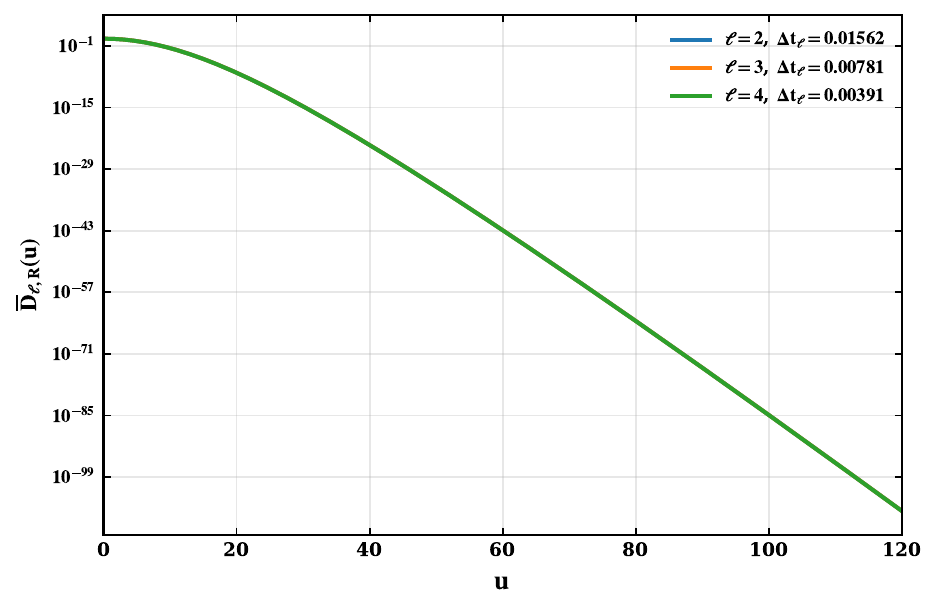}
    \caption{{
    The normalized functions
    $\overline{\mathcal D}_{\ell,R}(u)$ against $u$ for several sufficiently
    fine levels $\ell$.
    }}
    \label{fig:contour_decay}
\end{subfigure}
\hfill
\begin{subfigure}[t]{0.48\textwidth}
    \vspace{0pt}
    \centering
    \includegraphics[width=\linewidth]{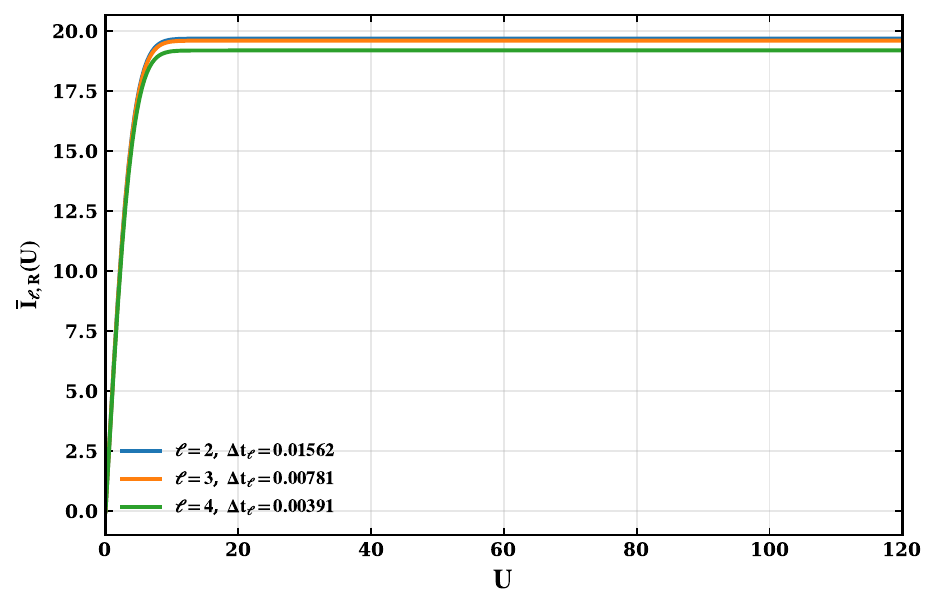}
    \caption{{
    The cumulative integrals $\overline I_{\ell,R}(U)$ against $U$ for the
    same levels.
    }}
    \label{fig:contour_integral}
\end{subfigure}

\caption{{
Numerical evidence for Conjecture~\ref{conj:contour_domination} using the
finer-grid approximation $\bar\Phi$. The normalized functions in
\subref{fig:contour_decay} decay with $u$ and stabilize under refinement,
while the cumulative integrals in~\subref{fig:contour_integral} approach a
common plateau.
}}
\label{fig:contour_domination_evidence}
\end{figure}
\fi
\FloatBarrier

{
The observed stabilization of $\overline{\mathcal D}_{\ell,R}$ and the
plateau of $\overline I_{\ell,R}$ provide numerical support for the decay and
integrability required in Conjecture~\ref{conj:contour_domination}. Because
the experiment uses a finite Fourier interval, finitely many levels, and the
approximation $\bar\Phi$ in place of $\Phi$, it does not constitute a proof of
integrability on $(0,\infty)$.
}

\section{\blue{Damping parameters used in the numerical experiments}}
\label{app:numerical_damping_values}

{\color{blue}
Table~\ref{tab:numerical_damping_values} lists every distinct combination of
the model parameter set, maturity, and market/payoff parameters for which the
practical damping optimization is performed. The model parameters associated
with the EuRos and SPY labels are given in
Table~\ref{tab:numerical_benchmarks}. All options are at the money and have
zero interest rate. The values of $R^\star$ are obtained from the one-sided
constrained optimization described in
the paragraph on practical damping selection in
Section~\ref{sec:numerical_results}.

\begin{table}[H]
    \centering
    \small
    \caption{Damping parameters selected for the distinct numerical
    configurations by the constrained damping optimization. Payoff-transform
    admissibility is enforced, whereas model-moment admissibility remains
    heuristic.}
    \label{tab:numerical_damping_values}
    \begin{tabular}{lccccc}
        \toprule
        Case & $T$ & $S_0$ & $K$ & $r$ & $R^\star$ \\
        \midrule
        EuRos & $0.1$ & $1000$ & $1000$ & $0$ & $-23.144$ \\
        EuRos & $1$ & $1000$ & $1000$ & $0$ & $-6.817$ \\
        EuRos & $2$ & $1000$ & $1000$ & $0$ & $-4.673$ \\
        EuRos & $10$ & $1000$ & $1000$ & $0$ & $-2.064$ \\
        EuRos & $2/365$ & $1000$ & $1000$ & $0$ & $-100.241$ \\
        EuRos & $7/365$ & $1000$ & $1000$ & $0$ & $-53.563$ \\
        SPY & $2/365$ & $1$ & $1$ & $0$ & $-62.432$ \\
        SPY & $7/365$ & $1$ & $1$ & $0$ & $-47.635$ \\
        \bottomrule
    \end{tabular}
\end{table}
}

\section{BL2 Markovian approximation used in the benchmarks}
\label{app:bl2_method}

The BL2 method of Bayer and Breneis~\cite{bayer2023weak} replaces the
fractional kernel in the rough Heston variance equation by a finite sum of
exponentials. With \(H=\alpha-\tfrac12\), it approximates
\[
K_\alpha(t)
=
\frac{t^{\alpha-1}}{\Gamma(\alpha)}
\qquad\text{by}\qquad
K_\alpha^{(n)}(t)
=
\sum_{j=1}^n w_j e^{-x_jt},
\]
where \((x_j,w_j)_{j=1}^n\) are positive nodes and weights. The BL2 rule is the
bounded \(L^2(0,T)\) approximation proposed in~\cite{bayer2023weak}; the
implementation constructs the rules successively, using the optimized
\((n-1)\)-factor rule to initialize the \(n\)-factor optimization. Hence,
\(n_{\mathrm{BL2}}\) is the number of exponential factors, or equivalently the
number of kernel-quadrature nodes. It is not the number of Fourier nodes or
Riccati time steps.

Introducing \(U_0^{(j)}=0\), the resulting Markovian approximation is
\[
V_t^{(n)}:=V_0+\sum_{j=1}^n w_jU_t^{(j)}.
\]
Together with the asset dynamics in~\eqref{eq:rHeston_model}, its factors
satisfy
\begin{align*}
dU_t^{(j)}
&=
\left[-x_jU_t^{(j)}
+\gamma\bigl(\theta-V_t^{(n)}\bigr)\right]dt
+\gamma\nu\sqrt{V_t^{(n)}}\,dW_t^2,
\qquad j=1,\ldots,n,
\end{align*}
so that the non-Markovian variance equation is replaced by an
\(n\)-dimensional Markovian system. In the notation of the reference
implementation~\cite{breneisGithub}, the parameter conversion is
\[
H=\alpha-\tfrac12,
\qquad
\lambda=\gamma,
\qquad
\nu_{\mathrm{BL2}}=\gamma\nu,
\qquad
\theta_{\mathrm{BL2}}=\gamma\theta.
\]

For a prescribed factor count, the implementation calls
\texttt{quadrature\_rule} with \texttt{mode="european"}, which invokes the
\texttt{european\_rule} construction and returns the kernel nodes and weights.
For this fixed rule, the Fourier truncation and the Riccati and Fourier
resolutions are then selected adaptively. In our comparison, we precompute one
rule for every \(n\in\{1,2,3,4\}\). Given a target relative tolerance
\(\varepsilon_{\mathrm{rel}}\), each candidate is priced using the internal
tolerance
\[
\texttt{rel\_tol}_{\mathrm{BL2}}
=
\frac{\varepsilon_{\mathrm{rel}}}{2}.
\]
This half-tolerance controls the adaptive Fourier and Riccati calculation for
the fixed kernel rule; it does not control the finite-factor kernel error. We
therefore compute the realized relative pricing error
\[
e_n
:=
\frac{|V_{\mathrm{BL2}}^{(n)}-V_{\mathrm{ref}}|}
{|V_{\mathrm{ref}}|}.
\]
We report the smallest factor count satisfying
\begin{equation}
n_{\mathrm{BL2}}
:=
\min\left\{
n\in\{1,2,3,4\}:e_n\leq\varepsilon_{\mathrm{rel}}
\right\}.
\label{eq:bl2_factor_selection}
\end{equation}
Thus, we select the minimum number of BL2 kernel-quadrature nodes that attains
the target relative error. Since this selection uses \(V_{\mathrm{ref}}\), it
is an oracle procedure used only for benchmarking. In our experiments, at
most four factors were sufficient to achieve all tested tolerances.

\newpage
\bibliographystyle{plain}
\bibliography{references}

\end{document}